\documentclass[letterpaper,11pt]{article}
\usepackage[margin=1in]{geometry}
\usepackage{amsmath}
\usepackage{amsthm}
\usepackage{amssymb}
\usepackage{bbm,url}
\usepackage[colorlinks=true,urlcolor=black,linkcolor=black,citecolor=black]{hyperref}
\usepackage{graphicx}
\usepackage[ruled,linesnumbered]{algorithm2e}
\usepackage{color}
\usepackage{subcaption}

\newtheorem{theorem}{Theorem}
\newtheorem*{theorem*}{Theorem}

\newtheorem{lemma}[theorem]{Lemma}
\newtheorem{claim}{Claim}
\newtheorem{remark}[theorem]{Remark}
\newtheorem{observation}[theorem]{Observation}
\newtheorem{example}[theorem]{Example}

\theoremstyle{definition}
\newtheorem{definition}{Definition}

\theoremstyle{condition}
\newtheorem{condition}{Condition}

\newcommand{\bR}{\mathbb{R}}

\newcommand{\cR}{\mathcal{R}}

\newcommand{\cP}{\mathcal{P}}
\newcommand{\CP}{\mathcal{P}}
\newcommand{\cI}{\mathcal{I}}
\newcommand{\cE}{\mathcal{E}}

\newcommand{\classP}{{\sf P}}

\newcommand{\classPPAD}{{\sf PPAD}}

\newcommand{\x}{\mbox{\boldmath $x$}}
\newcommand{\f}{\mbox{\boldmath $f$}}
\renewcommand{\c}{\mbox{\boldmath $c$}}
\renewcommand{\v}{\mbox{\boldmath $v$}}
\newcommand{\s}{\mbox{\boldmath $s$}}
\newcommand{\uu}{\mbox{\boldmath $u$}}
\newcommand{\qq}{\mbox{\boldmath $q$}}
\newcommand{\q}{\mbox{\boldmath $q$}}
\newcommand{\0}{\mbox{\boldmath $0$}}
\newcommand{\vv}{\mbox{\boldmath $v$}}
\renewcommand{\r}{\mbox{\boldmath $r$}}
\newcommand{\p}{\mbox{\boldmath $p$}}
\newcommand{\y}{\mbox{\boldmath $y$}}

\newcommand{\w}{\mbox{\boldmath $w$}}
\newcommand{\ff}{\mbox{\boldmath $f$}}

\newcommand{\argmax}{{\operatorname{\mathrm{arg\,max}}}}
\newcommand{\argmin}{{\operatorname{\mathrm{arg\,min}}}}
\newcommand{\zeros}{\mathbf{0}}
\newcommand{\zz}{\mathbf{z}}

\newcommand{\ps}{p^*}

\newcommand{\bAPR}{{\overline{N_{PR}}}}

\newcommand{\MPB}{{\mathit{MPB}}}
\newcommand{\MBB}{{\mathit{MBB}}}
\newcommand{\balpha}{\boldsymbol{\alpha}}
\newcommand{\bbeta}{\boldsymbol{\beta}}

\title{Competitive Allocation of a Mixed Manna}
\author{Bhaskar Ray Chaudhury\thanks{MPI for Informatics, Saarland Informatics Campus, Graduate School of Computer Science, Saarbr\"ucken, Germany}\\ \texttt{\small braycha@mpi-inf.mpg.de} \and Jugal Garg\thanks{University of Illinois at Urbana-Champaign. Supported by NSF Grant CCF-1942321 (CAREER)}\\ \texttt{\small jugal@illinois.edu}  \and Peter  McGlaughlin \thanks{University of Illinois at Urbana-Champaign. Supported by NSF Grant CCF-1942321 (CAREER)}\\ \texttt{\small mcglghl2@illinois.edu} \and Ruta Mehta\thanks{University of Illinois at Urbana-Champaign. Supported by NSF Grant CCF-1750436 (CAREER)}\\ \texttt{\small rutameht@illinois.edu}}
\date{}

\begin{document}
	
\maketitle
\thispagestyle{empty}
\pagestyle{empty}

\begin{abstract}
We study the fair division problem of allocating a mixed manna under additively separable piecewise linear concave (SPLC) utilities. A mixed manna contains goods that everyone likes and bads that everyone dislikes, as well as items that some like and others dislike. The seminal work of Bogomolnaia et al.~\cite{BogomolnaiaMSY17} argue why allocating a mixed manna is genuinely more complicated than a good or a bad manna, and why competitive equilibrium is the best mechanism. They also provide the existence of equilibrium and establish its peculiar properties (e.g., non-convex and disconnected set of equilibria even under linear utilities), but leave the problem of computing an equilibrium open. This problem remained unresolved even for only bad manna under linear utilities.

Our main result is a simplex-like algorithm based on Lemke's scheme for computing a competitive allocation of a mixed manna under SPLC utilities, a strict generalization of linear. Experimental results on randomly generated instances suggest that our algorithm will be fast in practice. The problem is known to be $\classPPAD$-hard for the case of \emph{good} manna~\cite{ChenDDT09}, and we also show a similar result for the case of \emph{bad} manna. Given these $\classPPAD$-hardness results, designing such an algorithm is the only non-brute-force (non-enumerative) option known, e.g., the classic Lemke-Howson algorithm (1964) for computing a Nash equilibrium in a 2-player game is still one of the most widely used algorithms in practice.

Our algorithm also yields several new structural properties as simple corollaries. We obtain a (constructive) proof of existence for a far more general setting, membership of the problem in $\classPPAD$, rational-valued solution, and odd number of solutions property. The last property also settles the conjecture of~\cite{BogomolnaiaMSY17} in the affirmative.

Furthermore, we show that if either the number of agents or the number of items is a constant, then the number of pivots in our algorithm is strongly polynomial when the mixed manna contains all bads, providing additional evidence to the practicality of our approach.
\end{abstract}
\newpage
\tableofcontents

\newpage
\pagestyle{plain}
\setcounter{page}{1}

\section{Introduction}\label{sec:intro}
Fair division is the problem of allocating a set of items among a set of agents in a \emph{fair} and \emph{efficient} way. This age-old problem, mentioned even in the Bible, arises naturally in a wide range of real-life settings such as division of family inheritance~\cite{PrattZ90}, partnership dissolutions, divorce settlements~\cite{BramsT96}, spectrum allocation~\cite{EtkinPT05}, airport traffic management~\cite{Vossen02}, office space between co-workers, seats in courses~\cite{SonmezU10,BudishC10}, computing resources in peer-to-peer platforms~\cite{GhodsiZHKSS11} and sharing of earth observation satellites~\cite{Bataille99}. The formal study of this problem dates back to the seminal work of Steinhaus~\cite{steinhaus1948problem} where he introduced the cake-cutting problem for more than two agents. Since then it has been an active research area in many disciplines. 

The vast majority of work in both Economics and Computer Science focuses on the case of \emph{disposable} goods, i.e., items that agents enjoy, or at least can throw away at no cost. However, many situations contain \emph{mixed manna} where some items are positive goods (e.g., cake), while others are undesirable bads (e.g., house chores and job shifts). Potentially, agents might disagree on whether a specific item is a good or a bad. Examples include: dividing tasks among various team members, deciding teaching assignments between faculty, managing pollution among firms, or splitting assets and liabilities when dissolving a partnership. 

Clearly, bads are \emph{nondisposable} and must be allocated. At first glance, it seems that the tools and techniques developed for the case of all goods might apply, but the mixed manna case turns out to be significantly more complex. The seminal work of Bogomolnaia et al.~\cite{BogomolnaiaMSY17} initiated the study of mixed manna, where they argue why allocating a mixed manna is genuinely more complicated than a good or a bad manna, and why an allocation based on \emph{competitive equilibrium with equal incomes} (CEEI) is the best mechanism.\footnote{E.g., competitive allocation not only achieves the standard notions of fairness called \emph{envy-freeness} and \emph{proportionality}, but it is also (Pareto) efficient and core stable.} They show the existence of equilibrium, and investigate some of their peculiar properties. Namely, they establish that even the simplest case of linear utility functions generally admits multiple equilibria, and the set of equilibria is non-convex and disconnected.\footnote{A similar result is shown for only bad manna~\cite{BogomolnaiaMSY19}. We also refer to an excellent survey article by Moulin~\cite{Moulin19}.} In sharp contrast, in the same setting with all goods, an equilibrium is captured by a convex program. 
Designing fast algorithms for mixed manna, even for linear utilities, is an important open question -- the abstract of Bogomolnaia et al.~\cite{BogomolnaiaMSY17} mentions,\footnote{Spliddit~\cite{spliddit} is a user friendly online platform for computing fair allocation in a variety of problems, which have drawn tens of thousands of visitors in the last five years~\cite{GoldmanP14}. Spliddit uses linear utilities.}

{\small
\begin{quote}
\emph{$\dots$ the implementation of competitive fairness under linear preferences in interactive platforms like SPLIDDIT will be more difficult when the manna contains bads that overwhelm the goods.}
\end{quote}
}

Recently, ~\cite{BranzeiS19,GargM20} make progress on this problem by designing polynomial-time algorithms for computing competitive allocation under linear utilities when either the number of agents or the number of items is a constant. These algorithms are based on clever \emph{enumeration-based exhaustive search}, which may not be fast in practice in the general case. 

\paragraph{Our Contributions} 
In this paper, we design a simplex-like algorithm for computing a competitive allocation of a mixed manna when agents' utility functions have a fairly general form: \emph{separable piecewise linear concave} (SPLC), a strict generalization of linear; see Section~\ref{sec:splc} for a formal definition. In economics, it is customary to assume that utility functions of goods are concave since they capture the important condition of decreasing marginal utilities. Likewise, this assumption is also natural for bads to capture increasing marginal disutility, e.g., considering the chore of reducing pollution from a plant where driving emissions toward zero likely comes at a rising cost. The SPLC functions are also important for the fair division problems to capture natural situations when there are limitations on the maximum amount of an item that can be assigned to an agent due to rationing and other restrictions.

Experimental results on randomly generated instances suggest that our algorithm will be fast in practice, answering the question raised by Bogomolnaia et al.~\cite{BogomolnaiaMSY17}. Our algorithm follows a systematic path rather than a brute force enumeration of every \emph{configuration}; see Section~\ref{sec:strongly}. We also show that the problem is $\classPPAD$-complete even when all items are bads. As a result, a polynomial time algorithm is not possible unless $\classPPAD=\classP$. We note that SPLC utilities are extensively studied in the case of good manna; see e.g.,~\cite{ChenDDT09,ChenT09,VaziraniY11,GargMSV15}. To the best of our knowledge, they have not been studied before for a bad (or mixed) manna. We also note that,~\cite{BogomolnaiaMSY17,BogomolnaiaMSY19} mention about reducing the bads under linear utilities into goods under SPLC utilities, however this may not always work; see Appendix~\ref{sec:bads_to_goods}.

Our approach is based on Lemke's complementary pivoting on a polyhedron~\cite{Lemke65}, which is similar in spirit to simplex algorithm for linear programming~\cite{Dantzig63} and classical Lemke-Howson algorithm for computing a Nash equilibrium of a 2-player game~\cite{LemkeH64}. A common phenomenon in these algorithms is that they perform well in practice even though their worst case behavior is exponential; the latter is exhibited via intricately doctored up instances that are designed to make the algorithm perform poorly, e.g., see~\cite{KleeM72} and ~\cite{SavaniS06} for simplex and Lemke-Howson, respectively. Given the $\classPPAD$-completeness of our problem, such a pivoting-based algorithm is the only non-brute-force (non-enumerative) option known. 

The most striking feature of this approach is that it not only gives a fast algorithm but also provides several new structural results as simple corollaries. First, it yields the first (constructive) proof of the existence of a competitive allocation of a mixed manna under SPLC utilities. Second, it shows that a rational-valued equilibrium exists if all input parameters are rational.  Third, together with the result of Todd~\cite{Todd76}, it gives a proof of membership of this problem in $\classPPAD$. Fourth, this shows that the number of equilibria is odd in a nondegenerate instance. We note that none of these results were known even for linear utilities. The last property also settles the conjecture of~\cite{BogomolnaiaMSY17} in affirmative, which shows the odd property for 2 agents (or 2 items) under linear utilities and conjectures the same for any number of agents and items. 

Furthermore, we show that if either the number of agents or the number items is a constant, then the number of pivots in our algorithm is polynomial when the mixed manna contains only bads. All our results also extend to a more general setting of \emph{exchange}; see Section~\ref{sec:prelim} for definition. To the best of our knowledge, the exchange setting was not studied before despite its natural applications, e.g., exchange of tasks among agents in which a group of university students teaching subjects or sports to each other, or some landlords providing shelter to apartment seekers in their houses in exchange of help in household chores~\cite{mitwohnen}. 

\paragraph{Techniques}
Our approach requires two steps. First, we need to derive a linear complementarity program (LCP) formulation for the problem whose solutions capture competitive equilibria. Second, we must show that the algorithm always terminates at a competitive equilibrium -- this is usually shown by proving no \emph{secondary-rays} (special kind of unbounded edges) in the LCP polyhedron; see Section \ref{sec:prel-lcp} for details. 

This approach has been extensively utilized for computing equilibria in markets (with only goods) and in games; see, e.g.,~\cite{Eaves76,GargMSV15,GargV14,GargMV18,KollerMS94,Sorensen12,HansenL18}. Each of them first obtains an LCP formulation that \emph{exactly} captures equilibria, and then shows that there are no secondary-rays. Despite significant efforts, no such LCP was found for competitive allocation of mixed manna. The LCP we design has ``non-equilibrium'' solutions, and furthermore, it has secondary-rays. 

We first note that both the above steps must work \emph{simultaneously}. In fact, it is not difficult to come up with an LCP formulation for only bads by extending the LCP for only goods~\cite{Eaves76,GargMSV15}. However, it does not yield an algorithm. Hence, we first come up with a different LCP for only bads. The case of mixed manna turns out to be even more challenging as simply merging the two LCPs does not work. This is due to the single utility maximization over all items for each agent, and it is apriori not clear how much an agent wants to spend on only goods (or bads). Using new ideas, we derive an LCP formulation that captures competitive allocation of a mixed manna, but it also captures some non-equilibrium solutions that we deal with in the second step.  

The second step presents the most significant challenge. The major issue with the Lemke's scheme is that, in general, it is not guaranteed to find a solution. This happens when the path followed by the algorithm leads to a secondary ray. 
 
As mentioned before, the standard way to show convergence of a complementary pivot algorithm to a solution is by proving that there are no secondary rays in the LCP polyhedron. However, our LCP formulation has secondary rays. Therefore, we must show that the algorithm never reaches a secondary ray to guarantee its termination to a competitive equilibrium. In addition, we must show that the final output of the algorithm is an equilibrium, rather than a non-equilibrium solution to the LCP. This makes the analysis of our algorithm more challenging than the previous works. 

For the hardness result, we show that finding a competitive allocation of a bad manna under SPLC utilities is $\classPPAD$-hard, even a $\tfrac{1}{poly(n)}$-approximation, where $n$ is the number of agents (Theorem \ref{thm:ceei-hardness}). A similar result is known for the case of good manna~\cite{ChenDDT09,chen2009spending,ChenPY13}. We obtain a reduction from the problem of computing a Nash equilibrium in two-player games, a known $\classPPAD$-complete problem~\cite{ChenDT09}. At a high-level, our approach follows the same approach as in the case of good manna. However, one major issue is {\em how to prevent an agent from consuming some bad}. This is easy in case of goods --  we can just set the corresponding utility function to zero. For bads this amounts to setting the disutility value to infinity, which may lead to the non-existence of equilibrium~\cite{ChaudhuryGMM20}. Therefore, we stick to finite disutility values and circumvent the issue through a careful choice of parameters and analysis. Our construction implies that the PPAD-hardness holds even if the PLC function for every pair of agent and bad has at most three segments with constant slopes.

\paragraph{Further Related Work}
The fair division literature is too vast to survey here, so we refer to the excellent books~\cite{BramsT96,RobertsonW98,Moulin03} and restrict attention to previous work that appears most relevant. 

Most of the work in fair division is focused on allocating a good manna with a few exceptions of bad manna~\cite{Su99,AzrieliS14, BramsT96, RobertsonW98}. The seminal paper of Bogomolnaia et al.~\cite{BogomolnaiaMSY17} is the first to study the case of mixed manna. While linear is the most studied utility function to model agents' preferences~\cite{BogomolnaiaMSY17}, SPLC is its natural extension to capture important generalizations. For these models, competitive allocation of a good manna is very well-understood. 
Two most ideal economic models to study competitive allocation are of Fisher and Exchange. In Fisher setting, the celebrated Eisenberg-Gale convex program captures equilibrium when utility functions are homothetic, concave and monotone, which includes linear~\cite{EisenbergG59,Eisenberg61}. The program maximizes the product of the agents' utilities (i.e., the Nash welfare) on all feasible utility profiles, and implies existence, convexity, uniqueness (of utility profile), and polynomial time computation; there are faster algorithms for some special cases~\cite{DevanurPSV08,Orlin10,Vegh14,Vegh17}. For exchange, polynomial time algorithms are known for subclasses of homothetic functions including linear~\cite{Jain07,Ye07,DuanM15,DuanGM16,GargV19}. Although the SPLC case is known to be $\classPPAD$-complete even in Fisher setting~\cite{ChenT09,ChenPY13}, the complementary pivot algorithm~\cite{GargMSV15} works well in practice and the only non-brute-force option known. 

The fair allocation of \emph{indivisible} items is also an intensely studied problem for the case when all items are goods with a few recent exceptions~\cite{AzizRSW17,AzizCL19,HuangL19,AzizCL19a,AzizCIW19,SandomirskiyH19}. Since the standard notions of fairness such as envy-freeness are not applicable, alternate notions have been defined for this case; see~\cite{LiptonMMS04,Budish11,CaragiannisKMPSW16,BarmanM17,GhodsiHSSY17,GargKK20} for a subset of notable work and references therein. The Nash welfare continues to serve as a major focal point in this case as well, for which approximation algorithms have been obtained under several classes of utility functions including linear and SPLC~\cite{ColeG15,ColeDGJMVY17,AnariGSS17,AnariMGV18,BarmanKV18,GargHM18,ChaudhuryCGGHM18}.

In our other related work~\cite{ChaudhuryGMM20}, we show that a slight variant of the problem of computing a competitive allocation of a bad manna under linear utilities is already $\classPPAD$-hard. This, together with the non-convex and disconnected set of solutions, suggests that our algorithm in this paper is likely to be the best one can hope for this problem even under linear utilities. 

\paragraph{Organization of the paper} We introduce notation and preliminaries in the following Section~\ref{sec:prelim}. In Section~\ref{sec:linear} we give a high level overview of our LCP formulation and algorithm in the special case of linear utilities to highlight the main issues and challenges that arise in computing competitive allocation of a mixed manna. In Section~\ref{sec:splc}, we formally define SPLC utilities and extend our LCP formulation to this more general problem. Our algorithm and its analysis appear in Section~\ref{sec:algo}. A precise description of all the results are presented in Section~\ref{sec:results}. In Section~\ref{sec:strongly}, we show a strongly polynomial bound of the algorithm for all bads when the number of agents (or items) is a constant. We show $\classPPAD$-hardness of the {\em bads only} problem in Section~\ref{sec:hardness}. 
Section~\ref{sec:experiments} summarizes our numerical experiments on randomly generated instances. Appendix~\ref{sec:bads_to_goods} presents a counterexample showing that bads cannot be reduced into goods, Appendix~\ref{app:sr} illustrates that the Lemke's scheme fails if we try a naive adaption of the LCP of~\cite{Eaves76,GargMSV15}, which is specialized to all goods, and all the omitted proofs appear in Appendix~\ref{sec:proofs}. Finally, Appendix~\ref{sec:convergence_all_bads} shows the convergence of the algorithm for only bad manna. 

\section{Preliminaries}\label{sec:prelim}
Let $M$ be the set of $m$ divisible items that needs to be divided among the set $N$ of $n$ agents. An item can be a good or a bad for an agent as discussed earlier. Each agent $i$ has a utility function $u_i: \bR^m_+ \rightarrow \bR$ over bundles of items. Let $\x_i = (x_{ij})_{j\in M}$ denote agent $i$'s assigned bundle containing $x_{ij}$ amount of item $j$. The standard notions of fairness and efficiency are \emph{envy-freeness} and \emph{Pareto optimality}, defined as follows:

\begin{itemize}
\item \emph{Envy-freeness:} An allocation $X = (\x_1, \dots, \x_n)$ is said to have no envy, if each agent weakly prefers her allocation over any other agents' allocation, i.e., $u_i(\x_i) \ge u_i(\x_j), \forall i, j\in N$. Envy-freeness also implies another standard fairness notion called \emph{proportionality}, where every agent receives at least a $1/n$ share of all items, i.e., $u_i(\x_i) \ge \frac{1}{n}u_i(M), \forall i\in N$. 

When agents have different weights (unequal rights/responsibilities), say $\eta_i$ is the weight of agent $i$, then we say that an allocation $X$ has no envy if $\frac{u_i(\x_i)}{\eta_i} \ge \frac{u_i(\x_j)}{\eta_j}, \forall i, j\in N$. 
\item \emph{Pareto optimality:} An allocation $X'=(\x'_1, \dots, \x'_n)$ Pareto dominates another allocation $X=(\x_1, \dots, \x_n)$ if $u_i(\x'_i)\geq u_i(\x_i),\forall i$ and $u_k(\x'_k)>u_k(\x_k)$ for some $k$. An allocation $X$ is Pareto optimal if no allocation $X'$ dominates $X$. 
\end{itemize}

Competitive allocations are well-known to be not only envy-free and Pareto optimal, but also core stable\footnote{No coalition of agents, by standing alone, can allocate better shares to each agent in the coalition.}. Two most ideal economic models to study competitive allocations are of Fisher and Exchange.\footnote{These are two fundamental economic models, introduced by Walras~\cite{Walras74} and Fisher~\cite{BrainardS00} in the late nineteenth century, respectively.} An exchange model is like a barter system, where each agent comes with an initial endowment of items and exchanges them with others to maximize her utility function. Fisher is a special case of exchange model where each agent has a fixed proportion of each item. Competitive Equilibrium with Equal Incomes (CEEI)~\cite{Varian74} is a special case of Fisher where each agent has the same endowment. 

\subsection{Competitive Equilibrium} 
Let $\w_i=(W_{ij})_{j\in M}$ denote the agent $i$'s initial endowment containing $W_{ij}\ge 0$ amount of item $j$. In Fisher, $W_{ij} = \eta_i, \forall i\in N, j\in M$, where $\eta_i$ is the budget (entitlement/weight) of agent $i$. In CEEI, $\eta_{i} = 1, \forall i\in N$. Given prices of items, each agent demands a utility maximizing (optimal) bundle by spending her budget (earned by selling the initial endowment). At (competitive) equilibrium, prices $\p=(p_j)_{j\in M}$ and allocation $(\x_i)_{i\in N}$, satisfy two conditions: 

\begin{enumerate}
\item \emph{Optimal bundle}, $\x_i$ maximizes agent $i$'s utility at $\p$, i.e., $\x_i \in \{\arg\max u_i(\y) \text{ s.t. } \sum_j y_{j}p_j = \sum_j W_{ij}p_j;\ y_j \ge 0, \forall j\}$, and 
\item \emph{Demand meets supply (market clearing)}, demand of each item equals its supply, i.e., $\sum_{i} x_{ij} = \sum_{i} W_{ij}, \forall j\in M$. 
\end{enumerate}

Observe that equilibrium prices are scale invariant, i.e., if $\p$ is an equilibrium price vector, then so is $\alpha \p, \forall \alpha >0$. Further, the prices of items, which are bads for all agents, will be negative at equilibrium. We can assume without loss of generality that each agent brings some fraction of some item, and that there is a unit amount of each item, i.e., $\sum_{i\in N} W_{ij} = 1, \forall j\in M$. 
\subsection{Linear Complementarity Problem and Lemke's Scheme}\label{sec:prel-lcp}
Linear Complementary Problem (LCP) is a generalization of Linear Programming (LP) complementary slackness conditions: Given an $n \times n$ matrix $A$ and an $n$-dimensional vector $\qq$, the problem is to find $\y$ such that 
\begin{equation}\label{0lcp}
\forall i \in [n]:\ \ \ (A\y)_i \le q_i; \ \ \ y_i \ge 0;\ \ \  y_i(A\y-\qq)_i =0 \enspace . 
\end{equation}
Clearly, the problem is only interesting when $q_j <0$ for some $j \in [n]$, otherwise $\mathbf{y} = \mathbf{0}$ offers a trivial solution. 
Let $\cP$ denote the $n$-dimensional polyhedron defined by the first two constraints of \eqref{0lcp}. We assume that $\cP$ is nondegenerate. That is, exactly $n-d$ constraints hold with equality on any $d$ dimensional face of $\cP$. Under this assumption, each solution to \eqref{0lcp} corresponds to a vertex of $\cP$ since exactly $n$ equalities must be satisfied. 

The LCPs are general enough to capture (strongly) NP-hard problems~\cite{CottlePS92} and therefore may not have a solution. Lemke's scheme first \emph{augments} the LCP by adding a scalar variable $z$, to create easily accessible solutions, and considers the formulation:
\begin{equation}\label{alcp}
\begin{array}{c}
\forall i\in [n]: \ \ \   (A\y)_i - z \le  q_i; \ \ \ y_i \ge 0; \ \ \  y_i(A\y - \qq)_i = 0\\ 
z \ge 0 
\end{array}\enspace .
\end{equation}

Observe that a solution $(\y, z)$ with $z=0$ of~\eqref{alcp} gives a solution $\y$ of~\eqref{0lcp} and vice versa. Let $\cP'$ be the polyhedron defined by the first two linear constraints for each $i\in [n]$, and $z\ge 0$ constraint. The dimension of $\cP'$ is $n+1$. Assuming that $\cP'$ is nondegenerate, solutions to \eqref{alcp} must still satisfy $n$ constraints. Therefore, the set of solutions $S$ is a subset of the 1-skeleton of $\cP'$, i.e., solutions consist of edges (1-dimensional faces) and vertices (0-dimensional faces) of $\cP'$. Further, any solution to \eqref{0lcp} must be a vertex of $\cP'$ with $z=0$.

Solutions $S$ to the augmented LCP have some important structural properties. We say that label $i$ is present at $(\y,z)\in \cP'$ if $y_i =0$ \emph{or} $(A\y)_i - z = q_i$. Every solution in $S$ is \emph{fully labeled} since label $i$ is present for all $i\in [n]$. A solution $s\in S$ contains \emph{double label} $i$ if $y_i =0$ \emph{and} $(A\y)_i - z =q_i$ for $i\in [n]$. Further, there are two edges of $S$ incident to $s$ since there are only two ways to relax the double label while keeping all the other labels. Obviously, any solution $s$ to \eqref{alcp}, which satisfies $z=0$, contains no double labels. Relaxing $z=0$ yields the unique edge incident to $s$ at this vertex.

From the above observations, it follows that $S$ consists of paths and cycles. We note that some of the edges in $S$ are unbounded. An unbounded edge of $S$ incident to vertex $(\y^*,z^*)$ with $z^* >0$ is called a {\em ray}. Formally, a ray $\cR$ has the form
\[ \cR = \{(\y^*, z^*) + \alpha (\mathbf{y}', z')\ |\ \alpha\ge 0\} \enspace ,\]
where $(\mathbf{y}',z') \neq \mathbf{0}$ solves \eqref{alcp} with $\qq = \mathbf{0}$ (the direction vector). Among all rays, one is special. Observe that $\y=0, z \ge |\min_{i} q_i|$ gives a solution to \eqref{alcp}, which forms an unbounded edge of $S$, known as {\em primary-ray}.
All other rays are called \emph{secondary-rays}. Starting from the primary-ray, Lemke's scheme follows a path on 1-skeleton of $\cP'$ with a guarantee that it never repeats a vertex. Therefore, either it reaches a vertex with $z=0$ that is a solution of the original LCP \eqref{0lcp}, or it ends up on a {\em secondary-ray}. In the latter case, the algorithm fails to find a solution, and in fact problem may not have a solution. Observe that we can replace $\mathbf{1} z$ with $\mathbf{c} z$ where $c_i = 0$ when $b_i > 0 $, and $c_i > 0$ when $b_i < 0$, without changing the role of $z$. 

In what follows, for simplicity, we use the shorthand notation of \[(A\y)_i \le q_i \ \ \perp \ y_i\] to represent $\{(A\y)_i - z \le  q_i; \  \ y_i \ge 0; \ \  y_i(A\y - \qq)_i = 0\}$ while defining LCPs as in \eqref{0lcp}.

\section{Warm up: Linear Utilities}\label{sec:linear}
In this section, we provide a high-level technical overview of our algorithm to convey the main ideas and challenges. For simplicity, we will assume linear utilities. We deal with the more involved case of SPLC utilities in Section~\ref{sec:splc}. Since linear is a special subcase of SPLC, all the formal proofs presented in later sections for SPLC simply apply to the linear case, so we do not present them separately here. 

Linear utility function is defined as $u_i(\x_i) := \sum_j U_{ij}x_{ij}$, where $U_{ij}$ is the utility of agent $i$ for a unit amount of item $j$. Clearly, $U_{ij}\ge 0$ if item $j$ is a good for $i$ and $U_{ij}<0$ if it is a bad. For bads, we also use $D_{ij} := |U_{ij}| > 0$ to denote the \emph{disutility} of agent $i$ for a unit amount of bad $j$. If $U_{ij}=0$, then we set $x_{ij}:= 0$ and do not introduce corresponding variable in the formulation. 

We show in Section~\ref{sec:splc} that it is without loss of generality to assume that all agents agree on whether an item $j$ is good or bad. Therefore, $M$ can be partitioned into a set $M^+$ of goods and a set $M^-$ of bads. We also show in Section~\ref{sec:splc} that competitive equilibrium  prices of bads are negative, and those of goods are positive. We may also assume without loss of generality that the total supply of every item is $1$.\footnote{This is like redefining the unit of items by appropriately scaling utility values.}  As discussed in Section~\ref{sec:prelim}, we derive our results for the most general exchange setting, which is a strict generalization of Fisher and CEEI. 

Competitive equilibrium in the exchange setting does not always exist. We show that it is guaranteed to exist under the \emph{strong connectivity} assumption (defined in Section~\ref{sec:splc}). This assumption implies the existence of equilibrium for \emph{all instances} of Fisher (and hence CEEI) setting under linear utilities. We will show that our algorithm converges to a competitive equilibrium under strong connectivity, thereby also implying a constructive proof of existence.\footnote{We note that the strong connectivity assumption is vacuous in case of only bad manna. For this case, we provide a separate convergence proof in Appendix~\ref{sec:convergence_all_bads} without any assumptions.}

Like the Simplex algorithm for linear programming (LP), our algorithm is based on Lemke's complementary pivoting scheme, which follows a path on a polyhedron and therefore easy to implement and fast in practice (see Section~\ref{sec:experiments} for experimental results). The complementary pivoting is a powerful tool to design non-enumerative algorithms for (PPAD-)hard problems; a prominent example is most widely used Lemke-Howson algorithm~\cite{LemkeH64} for computing a Nash equilibrium in a two-player game. Such an approach follows two steps: 
\begin{enumerate}
\item Design a Linear Complementarity Problem (LCP) formulation that {\em exactly} captures the solutions.
\item Show that a complementary pivoting scheme converges to a solution. It essentially boils down to showing that the algorithm will not reach an infinite edge on the LCP polyhedron -- a {\em secondary-ray} (see Section \ref{sec:prel-lcp} for the definitions).  
\end{enumerate}

The main challenge here is to make both the steps work \emph{simultaneously}. In fact, it is not difficult to come up with an LCP formulation for only bads (i.e., $M^+ = \emptyset$) by extending the LCP for only goods~\cite{Eaves76,GargMSV15}. However, it does not yield an algorithm. Hence, we first come up with a different LCP for only bads. The case of mixed manna turns out to be even more challenging as simply merging the two LCPs does not work. This is due to a single utility maximization over all items, and it is apriori not clear how much an agent wants to spend on only goods (or bads) as shown in Example~\ref{example1}. 

The standard way, used in all related works under all goods case~\cite{Eaves76,GargMSV15,GargV14,GargMV18}, to show convergence of Lemke's scheme to a solution is by proving that there are no secondary-rays in the LCP polyhedron as discussed in Section~\ref{sec:prel-lcp}. Despite significant efforts, no such LCP was found for competitive allocation of mixed manna. We then switched our attention to showing convergence to a solution even though there are secondary-rays. 

\subsection{LCP Formulation}
\subsubsection{Only Bads}
In this section, we derive an LCP formulation for only bads, i.e., $M^+=\emptyset$. Recall from Section~\ref{sec:prelim} that at competitive equilibrium, every agent receives their optimal bundle and demand meets supply. Thus, the LCP need to capture both these conditions. Since LCP allows only non-negative variables, we use $p_j \ge 0$ even for bad $j$ and interpret it as the {\em payment to agents per unit of bad} done. At prices $\p$ and allocation $\x$, the money {\em earned} by agent $i$ on bad $j$ is $x_{ij}p_j$, which is a quadratic term. To ensure linearity of equations, we use $f_{ij}$ to denote the money {\em earned} by $i$ on $j$. Then, the following linear equations capture the demand meets supply condition:
\begin{equation}\label{eq.chores-mc}
\forall j \in M,\ \ \ \sum_{i \in N} f_{ij} = p_j, \ \ \ \ \mbox{ and } \ \ \ \ \forall i \in N, \ \ \ \sum_{j \in M} f_{ij} =\sum_{j \in M} W_{ij} p_j
\end{equation}

At prices $\p$, agent $i$'s optimal bundle $\x_i \in \{\argmin \sum_{j} D_{ij} x_{ij}\ \text{ s.t. } \ \sum_{j} x_{ij}p_j = \sum_j W_{ij}p_j;\ x_{ij} \ge 0, \forall j\}$. That is, agent $i$ wants to minimize her total disutility (pain), and in the linear case, her optimal bundle consists of only those bads that minimizes the pain per unit of money. Let $\MPB_i(\p)$ denote the bads with minimum-pain-per-buck for agent $i$ at prices $\p$, i.e., $\MPB_i(\p)=\argmin_j D_{ij}/p_j$. Then, $f_{ik}>0$ only if $k \in \MPB_i(\p)$. Let us now introduce a variable $r_i$ to capture the inverse of MPB for agent $i$. Then, the following captures the optimal bundle condition: 
\begin{equation}\label{eq.chores-ob}
\forall i\in N, \forall j \in M, \ \ \ p_j - D_{ij}r_i \le 0, \ \ \  f_{ij}\ge 0,\ \ \ \ f_{ij}(D_{ij}r_i - p_j) = 0
\end{equation}

Using \eqref{eq.chores-mc} and \eqref{eq.chores-ob}, we obtain the following LCP (very similar to the one for the goods only case~\cite{Eaves76}): Variables are, $p_j$ representing price of bad $j$, $f_{ij}$ representing earning of agent $i$ from bad $j$, and $r_i$ representing inverse of minimum pain-per-buck (MPB) of agent $i$. 
\begin{subequations}\label{eq.chores-lcp}
\begin{eqnarray}
\forall i \in N:  & \sum_{j\in M} W_{ij} p_j - \sum_{j \in M} f_{ij}\le 0  & \ \ \perp  \ \ r_i \\[10pt]
\forall j \in M:  & \sum_{i \in N} f_{ij} - p_j\le 0  & \ \ \perp  \ \ p_j\\[10pt]
\forall i \in N, \forall j \in M: &  p_j - D_{ij}r_i \le 0  & \ \ \perp  \ \ f_{ij}
\end{eqnarray} 
\end{subequations}

It is easy to show that all competitive equilibria are solutions of LCP~\eqref{eq.chores-lcp}. One issue, common to all related LCPs~\cite{Eaves76,GargMSV15,GargV14,GargMV18}, is that the LCP~\eqref{eq.chores-lcp} has more solutions, e.g., setting all variables to $0$ is a (trivial) solution. The fix used by all the previous works is: Since equilibrium prices are non-zero and scale-invariant, it is without loss of generality to assume $p_j \ge 1, \forall j\in M$ at equilibrium, and therefore consider $(1+p_j)$ as the price of item $j$ and modify the LCP accordingly. However, this fix fails miserably for the LCP~\eqref{eq.chores-lcp}. In particular, it indeed does give an LCP that exactly captures all competitive equilibria, but it does not yield to an algorithm. If we apply Lemke's scheme on such an LCP, it encounters an infinite edge ({\em secondary-ray}) in a couple of steps and hence fails to find a solution to the LCP.\footnote{If we replace $p_j$ with $(1+p_j)$ in LCP~\eqref{eq.chores-lcp}, and then augment it by adding scalar variable $-z$ in inequalities with negative right hand side to apply Lemke's scheme, we get:
\[
\begin{array}{cccl}
\forall i \in N, & \sum_{j\in M} W_{ij} p_j - \sum_{j \in M} f_{ij}-z \le -\sum_{j\in M} W_{ij} & \perp & r_i\\
\forall j \in M, & \sum_{i \in N} f_{ij} - p_j\le 1 & \perp & p_j\\
\forall i \in N,\ \forall j \in M, &  p_j - D_{ij}r_i -z \le -1 & \perp & f_{ij}
\end{array}
\]
Suppose agent $k\in N$ has the highest total endowment $\sum_j W_{kj}$ that is more than $1$, then for $\p,\f,\r=0$ and $\forall z\in [\sum_j W_{kj}, \infty)$ are solutions of the above LCP, forming the primary-ray. The vertex at the end of this primary-ray has $z=\sum_j W_{kj}$ where the first inequality above becomes tight for agent $k$. Lemke's scheme does complementary pivot by increasing the corresponding variable $r_k$. Note that, while fixing the remaining variables to the current value, any $r_k\ge 0$ is a solution. Therefore, the algorithm will increase $r_k$ infinitely without finding the next vertex. This is another unbounded edge of the LCP, a secondary-ray. Thus, Lemke's scheme gets stuck in the first pivoting step itself and fails to find a solution; see Appendix \ref{app:sr} for further discussion.}

We need a fix for LCP \eqref{eq.chores-lcp} so that the Lemke's scheme works, and besides, the resulting LCP is extendable to allow goods as well. Extension to mixed manna case has to capture \emph{negative} prices (and money allocation) for bads while combining optimal bundle conditions of bads with those of goods. We next discuss the general mixed manna setting to show an approach that handles all these issues.

\subsubsection{Mixed Manna}

Since bads incur disutility, no agent wants to consume (do) them unless there is a valid reason. At given prices, an agent $i$ is willing to do a bad only if either her income from endowment, i.e., $(\sum_{j\in M^+} W_{ij} p_j - \sum_{j \in M^-} W_{ij}p_j)$, is negative, implying she needs to earn, or she wants to buy a good from money earned because the utility from the good outweighs the disutility of doing the bad. The latter condition can be formally stated as: Recall minimum-pain-per-buck bads as $\MPB_i(\p)=\argmin_{j \in M^-} D_{ij}/(-p_j)$, and similarly define maximum-bang-per-buck (MBB) goods as $\MBB_i(\p)= \argmax_{j \in M^+} U_{ij}/p_j$. Naturally, agent $i$ consume goods only from $\MBB_i(\p)$ and bads only from $\MPB_i(\p)$, if at all. And, if $\MBB_i(\p)\ge \MPB_i(\p)$, then agent $i$ may want to consume bad $b\in \MPB_i(\p)$ so that from the earned money she can buy a good $g\in \MBB_i(\p)$. At equilibrium, the inequality has to hold with equality, otherwise $i$ will demand infinite amounts of both $g$ and $b$. By capturing both the max and min ratios for goods and bads respectively in $\frac{1}{r_i}$, the optimal bundle condition can be stated as:
\[
\begin{array}{lll}
\forall i \in N, \forall j \in M^+:&\ \ \  \frac{U_{ij}}{p_j} \le \frac{1}{r_i}& \ \ \text{and}\ \ \  f_{ij}>0 \Rightarrow \frac{U_{ij}}{p_j} = \frac{1}{r_i} \\[10pt]
\forall i \in N, \forall j \in M^-:&\ \ \ \frac{D_{ij}}{p_j} \ge \frac{1}{r_i}&\ \ \text{and} \ \ \  f_{ij}>0 \Rightarrow \frac{D_{ij}}{p_j} = \frac{1}{r_i}
\end{array}\enspace .
\]

Here, $f_{ij}$ for bad $j$ should be thought of as earning of agent $i$. Combining the above with an appropriate extension of~\eqref{eq.chores-mc}, we get the following LCP,
\[
\begin{array}{rccl}
\forall i\in N: & \displaystyle\sum_{j \in M^+} f_{ij} - \displaystyle\sum_{j \in M^-} f_{ij} \le \displaystyle\sum_{j \in M^+} W_{ij} p_j - \displaystyle\sum_{j \in M^-} W_{ij} p_j &  \perp & r_i\\[20pt]
\forall j\in M^+: & p_j \le \sum_{i \in N} f_{ij}&   \perp & p_j\\[10pt]
\forall j\in M^-: & \sum_{i \in N} f_{ij} \le p_j  & \perp&  p_j\\[10pt]
\forall i \in N,  \forall j \in M^+:& U_{ij} r_i - p_j \le 0;& \perp &  f_{ij}\\[10pt]
\forall i \in N,  \forall j \in M^-:& p_j - D_{ij} r_i \le 0;& \perp &  f_{ij}
\end{array}
\]

Again, all competitive equilibria are solutions of the above LCP, but it has more solutions such as all-zeros. Let us now replace $p_j$ with $(P-p_j)$ and $r_i$ with $(R-r_i)$, where $P,R>0$ are large constants such that $R > P/U_{min}$ for $U_{min}=\min_{(i,j): U_{ij}\neq 0} |U_{ij}|$. Think of $P$ as an upper bound on prices and $R$ an upper bound on $r_i$'s (up to scaling) at any equilibrium. With this, we get the following LCP:
\begin{subequations}\label{eq:mixed-lcp}
\begin{eqnarray}
\hspace{-10pt}\forall i\in N:\hspace{1cm} & \hspace{-1cm}\displaystyle\sum_{j \in M^+} f_{ij} - \displaystyle\sum_{j \in M^-} f_{ij} \le \displaystyle\sum_{j \in M^+} W_{ij} (P-p_j) - \displaystyle\sum_{j \in M^-} W_{ij} (P-p_j) &  \perp \ \   r_i \label{eq:mca}\\[10pt]
\forall j\in M^+: & (P-p_j) \le \sum_{i \in N} f_{ij}&   \perp \ \   p_j \label{eq:mcg}\\[10pt]
\forall j\in M^-: & \sum_{i \in N} f_{ij} \le (P-p_j)  & \perp\ \   p_j \label{eq:mcb}\\[10pt]
\forall i \in N,  \forall j \in M^+:& U_{ij} (R-r_i) - (P-p_j) \le 0& \perp \ \   f_{ij}\label{eq:obg}\\[10pt]
\forall i \in N,  \forall j \in M^-:& (P-p_j) - D_{ij} (R-r_i) \le 0& \perp \ \   f_{ij}\label{eq:obb}
\end{eqnarray}
\end{subequations}

In the next section, we will discuss how solutions of LCP~\eqref{eq:mixed-lcp} with $\p < P$ and $\r < R$ maps to competitive equilibrium.\footnote{With $\p < P$, we mean $p_j < P, \forall j$, and so on.}  However, there are still two crucial issues that the algorithm needs to handle: $(i)$ ``dummy solutions'' where $\p \not< P$ or $\r \not< R$, and $(ii)$ the augmented LCP to apply Lemke's scheme has secondary-rays (easy to construct). We show in Section \ref{sec:lin-algo} that the algorithm, starting from the {\em primary-ray}, will never reach either. 

Every equation in LCP~\eqref{eq:mixed-lcp} represents three constraints of the LCP, namely the linear inequality constraint, non-negativity of the corresponding variable, and complementarity condition which requires either the inequality to be tight or the variable to be zero. To avoid ambiguity, now on we will use equation number to refer to the linear constraint, and equation number with a prime to refer to the complementarity constraint. For example, \eqref{eq:mcg} refers to $(P-p_j)\le \sum_{i\in N} f_{ij}$ and (\ref{eq:mcg}') refers to $p_j((P-p_j)- \sum_{i\in N}f_{ij})=0$.

\subsubsection{Correctness}\label{sec:lin-lcp-correct}
\noindent{\bf Proof (sketch).} Although not obvious, it is not too difficult to show that a competitive equilibrium gives a solution of LCP~\eqref{eq:mixed-lcp}. Given equilibrium prices $\p^*$ and corresponding money allocation $\f^*$, construct a solution of the LCP as follows: Assume $p^*_j < P, \forall j$ due to scale invariance, and set
\[
\begin{array}{rl}
\forall j \in M^+, \forall i\in N: & \  p_j = P - p^*_j \ \ \text{ and } \  \ f_{ij} = f^*_{ij}\\[10pt]
\forall j \in M^-, \forall i\in N: &\ p_j = P- (-p^*_j) \ \ \text{ and } \  \  f_{ij} = -f^*_{ij}\\[10pt]
\forall i \in N: & \ r_i = R - r^*_i,\ \  \mbox{ where } \frac{1}{r^*_i} =  \left\{\begin{array}{ll}\max_{j \in M^+} \frac{U_{ij}}{p^*_j} & \mbox{ if } \exists j \in M^+,\ f^*_{ij}>0\\  \min_{j \in M^-} \frac{D_{ij}}{(-p^*_j)} & \mbox{ otherwise } \end{array}\right.  
\end{array}\enspace .
\]
The more difficult part is to map LCP solutions to competitive equilibrium (CE). First observe that if $(\p,\f,\r)$ is a solution of the LCP, then $\p \le P$ and $\r \le R$: Since all variables are non-negative, \eqref{eq:mcb} ensures that $p_j\le P$ for all bads $j$, and this together with \eqref{eq:obb} ensures that $r_i\le R$ for all $i$. Then, \eqref{eq:obg} ensures that even for all goods $j, p_j \le P$. 

Unfortunately, the LCP does have ``dummy solutions'', ones that do not give CE. For example, setting $p_j=P, \forall j \in M$, $r_i=R,\ \forall i\in N$, and $\f=\zeros$ is a solution of the LCP that gives no information about equilibrium. There may be more such dummy solutions that we are unable to discard, however in the next section we will argue that the algorithm has to find a ``desired solution'' before it encounters any such dummy solution (this is in addition to avoiding the secondary-rays). 

Now we argue that if $\p < P$ and $\r < R$ at $(\p,\f,\r)$ then it can be mapped to a competitive equilibrium. For every good $j \in M^+$, set $p^*_j = P - p_j$ and $f^*_{ij}=f_{ij},\ \forall i \in N$, and for every bad $j\in M^-$, set $p^*_j = -(P-p_j)$ and $f^*_{ij} = -f_{ij}$. For every agent $i$, set $r_i^* = R-r_i$. Clearly, $|\p^*|>0$ and $\r^*>0$. Since the inequalities of \eqref{eq:mca}, \eqref{eq:mcg}, \eqref{eq:mcb} are satisfied simultaneously, we have
\[
\sum_{j \in M} p^*_j \le \sum_{j\in M, i \in N} f^*_{ij} \le \sum_{j \in M} p^*_j \enspace . 
\]

Thereby, all of them hold with equality implying that every agent spends/earns exactly her budget, and every item is allocated completely. To show that $\f^*$ indeed allocates optimal bundle to every agent at prices $\p^*$, first we note that \eqref{eq:obg} and \eqref{eq:obb} ensures that for every agent $U_{ig}/p^*_g \le 1/{r^*_i} \le D_{ib}/{(-p^*_b)}$ for any good-bad pair $(g,b)$. Therefore, agents do not have to demand infinite amount of any item. Next, the corresponding (\ref{eq:obg}') and (\ref{eq:obb}') ensure that $i$ is allocated goods only from $\MBB_i(\p^*)$ and bads only from $\MPB_i(\p^*)$. This together with the fact that she exactly spends her net earning implies $(\p^*,\f^*)$ forms a competitive equilibrium.

\begin{theorem}[Informal]\label{thm:lcp}
For mixed manna under linear utilities, solutions of LCP~\eqref{eq:mixed-lcp} with $\p < P$ and $\r < R$ are in one-to-one correspondence with competitive equilibria.\footnote{
In Section~\ref{sec:splc}, we extend LCP~\eqref{eq:mixed-lcp} to capture equilibria under more general SPLC utilities. A number of new issues arise, e.g., the characterization of optimal bundle turns out to be much more complex. We show that it can be captured through linear and complementary conditions that still uses a single variable $r_i$ to tie them together.} 
\end{theorem}

In summary, solutions of LCP~\eqref{eq:mixed-lcp} satisfy $\p \le P$ and $\r\le R$, and those with $\p < P$ and $\r < R$ exactly captures competitive equilibria. Many issues still remain: 1) LCP has ``dummy'' non-equilibrium solutions, 2) augmented LCP has secondary-rays, and 3) LCP polyhedron has multiple inherent degeneracies. We note that none of these issues arise in previous algorithms~\cite{Eaves76,GargMSV15,GargV14,GargMV18}. 

\subsection{Algorithm}\label{sec:lin-algo}
To apply Lemke's scheme, we need to add $(-z)$ term to all the inequalities with possibly negative right hand side (rhs). Observe that these are inequalities of \eqref{eq:mca},\eqref{eq:mcg},\eqref{eq:obg}, 
where for \eqref{eq:obg} rhs is $P - U_{ij} R \le P - U_{min} R < 0$. 

\paragraph{Handling Degeneracies} Every inequality of LCP is paired up with a variable, and for every pair, either the variable is zero or the inequality is tight at a solution. If both are true for some pair then it is called a double label (details in Section \ref{sec:prel-lcp}). Lemke's scheme follows a path of vertices and edges in the solution space, and at every vertex pivots by either making the variable non-zero or relaxing the tight inequality corresponding to the double label (complementary pivot). Therefore to avoid ambiguities, it is important to ensure a unique double label at every vertex that the algorithm encounters. This follows if the LCP polyhedron is \emph{nondegenerate}. In general, there are standard ways to handle degeneracy by symbolic or numerical perturbation of the input parameters. However, in our case they are not sufficient. 

To avoid degeneracies and to facilitate the final convergence proof, we need to add carefully chosen coefficients to the $(-z)$ terms.\footnote{This issue does not arise in the known LCPs for goods only case~\cite{Eaves76,GargMSV15,GargV14,GargMV18}.} For every good $j\in M^+$, define $\delta_j=(1+\epsilon_j)$ where $\epsilon_j>0$ is a uniform random value from $(0,1/m)$. We show in Section~\ref{sec:alcp} that if the input parameters of the mixed manna, namely $U_{ij}$'s and $W_{ij}$'s do not have any polynomial relation, then $\epsilon_j$'s can be carefully chosen so that $\CP$ is indeed nondegenerate. The augmented LCP is:

\begin{subequations}\label{eq:mixed-alcp}
\begin{eqnarray}
\hspace{-10pt}\forall i\in N: \hspace{1.2cm} & \hspace{-1.2cm}\displaystyle\sum_{j \in M^+} f_{ij} - \displaystyle\sum_{j \in M^-} f_{ij} - z\le \displaystyle\sum_{j \in M^+} W_{ij} (P-p_j) - \displaystyle\sum_{j \in M^-} W_{ij} (P-p_j) &  \perp \ \   r_i \label{eq:amca}\\[5pt]
\forall j\in M^+: & (P-p_j) \le \sum_{i \in N} f_{ij} + \delta_j z &   \perp \ \   p_j \label{eq:amcg}\\[5pt]
\forall j\in M^-: & \sum_{i \in N} f_{ij} \le (P-p_j)  & \perp\ \   p_j \label{eq:amcb}\\[5pt]
\forall i \in N,  \forall j \in M^+:& U_{ij} (R-r_i) - (P-p_j) - z \le 0& \perp \ \   f_{ij} \label{eq:aobg}\\[5pt]
\forall i \in N,  \forall j \in M^-:& (P-p_j) - D_{ij} (R-r_i) \le 0& \perp \ \   f_{ij} \label{eq:aobb}\\
& z\ge 0 & 
\end{eqnarray}
\end{subequations}

By construction and Theorem \ref{thm:lcp} we get:

\begin{lemma}\label{lem:1} 
Every solution of LCP~\eqref{eq:mixed-alcp} with $z=0$, $\p < P$, and $\r <R$ gives a competitive equilibrium.
\end{lemma}

Henceforth, we will use $\y$ to represent vector $(\p,\ff,\r)$. Let $\CP$ denote the polytope defined by the linear inequalities of LCP~\eqref{eq:mixed-alcp} including $\y, z\ge0$. 
\medskip

As discussed in Section \ref{sec:prel-lcp}, the {\em primary-ray} of the LCP is defined as the unique unbounded edge where $\y=0$ and $z\ge 0$. The algorithm (Lemke's scheme) starts on this primary-ray of \eqref{eq:mixed-alcp} where $z$ varies from $\infty$ to a positive value so that one of the inequalities becomes tight giving the double label at the vertex reached. Since $R$ is big relative to $P$, it ensures that the double label corresponds to \eqref{eq:aobg}. Then on, the algorithm pivots at the double label by relaxing one constraint, and traveling along the corresponding edge of $\CP$ to the next vertex solution -- see Algorithm \ref{algo:mixed} in Section \ref{sec:algo}. Crucially, the complementary pivot ensure that it never revisits any vertex or edge \cite{Lemke65}. Therefore, the algorithm terminates when it either encounters a vertex $(\y,z)$ with $z=0$, or a {\em secondary-ray}, an unbounded edge other than the primary-ray. We next show that in the former case we get a competitive equilibrium and the latter case never happens.

To ensure competitive equilibrium in the former case, we need to show that $\p < P$ and $\r < R$ as well (Lemma \ref{lem:1}). Furthermore, \eqref{eq:mixed-alcp} has an obvious secondary-ray, namely $\p=P$, $\r=R$, $\ff=0$ and $z$ positive, and may have many more where a subset of $p_j$'s and $r_i$'s are set to $P$ and $R$ respectively. The algorithm needs to avoid all of these rays. We handle both of these issues by showing that the algorithm can never encounter a point where $p_j=P$ for any $j\in M$, and $r_i=R$ for any $i\in N$. The proof is involved and requires a number of steps that we briefly discuss next. A detailed proof for the general SPLC utilities is given in Section \ref{sec:algo_convergence}.
\medskip

We prove four main claims: 
\begin{itemize}
\item[$(a)$] $\p \le P$, $\r \le R$, and if $r_i=R$ for an $i\in N$ then $p_j=P, \forall j\in M^-$. 
\item[$(b)$] If $p_g=P$ for some good $g\in M^+$ then $p_j=P$ for all $j\in M$. 
\item[$(c)$] If $p_b=P$ for some chore $b\in M^-$ then $p_j=P$ for all $j\in M^-$. 
\item[$(d)$] It can not be that $p_b=P$ for all $b\in M^-$ unless $p_g=P$ for some good $g\in M^+$ and $z=0$. 
\end{itemize}

Using these, let us first argue that the algorithm can never reach a point $\uu=(\y,z)$ (a vertex or on an edge) where $p_j=P$ for some item $j$ or $r_i=R$ for some agent $i$. 

The four claims above imply that if $p_j=P$ for some $j$ or $r_i=R$ for some $i$ then all $p_j$'s take value $P$ simultaneously. Consider the first point $\vv$ before $\uu$ where $0 < \p < P$. Observe that such a $\vv$ exists because on the primary-ray, where the algorithm starts, $\p=\zeros$. At $\vv$, claim $(a)$ above imply $\r <R$, and the {\em demand meets supply} conditions of \eqref{eq:amca}, \eqref{eq:amcg} and \eqref{eq:amcb} together with unit supply for every item gives,
\[
\sum_{j \in M^+} (P-p_j) - \sum_{j \in M^-} (P-p_j) + nz = \sum_{i, j \in M^+} f_{ij} - \sum_{i, j \in M^-} f_{ij} = \sum_{j \in M^+} (P-p_j) - \sum_{j \in M^-} (P-p_j) - z\sum_{j \in M^+} \delta_j 
\]

By canceling the price terms we are left with $z(n+\sum_{j\in M^+} \delta_j)=0$, implying $z=0$ at $\vv$. Thus, by Lemma~\ref{lem:1}, $\vv$ itself gives an equilibrium and therefore the algorithm stops at $\vv$ (or before) and never reaches $\uu$. Full details are in Lemma~\ref{lem:no_z_0_solution_2}. 

Coming back to the four claims, $(a)$ is relatively easy to show; see Lemma~\ref{lem:bound_on_p_and_r}. We next explain the idea behind $(b)$. Let $p_g=P$ for some good $g\in M^+$, then (\ref{eq:amcg}') forces the inequality to be tight for $g$, implying $\delta_j z+\sum_{i\in N} f_{ig}=0\Rightarrow z=0$. With $z=0$, \eqref{eq:amca}, \eqref{eq:amcg}, \eqref{eq:amcb} together forces all of them to hold with equality (by similar arguments as discussed in Section \ref{sec:lin-lcp-correct}). Furthermore, for all agents $a\in N$ who likes $g$, i.e., $U_{ag}>0$, \eqref{eq:aobg} implies $r_a=R$ since $r_a\le R$ by $(a)$. Replacing $r_a=R$ in \eqref{eq:aobb} for agent $a$ and all the chores, we get that $p_j=P$ for all $j\in M^-$. This in turn ensures $f_{aj}=0$, for all $j\in M^-$ due to~\eqref{eq:amcb}. 

Now consider the optimal bundle conditions~\eqref{eq:aobg} for agent $a$. Since $z=0$, these essentially are $(R-r_a) \le \frac{(P-p_j)}{U_{aj}}$ and $(R-r_a) = \frac{(P-p_j)}{U_{aj}}$ if $f_{aj}>0$. Given that $(R-r_a)=0$, the latter implies if $f_{aj}>0$ then $p_j=P$, but (\ref{eq:amcg}') forces $\sum_{i\in N} f_{ij}=(P-p_j)=0$ implying $f_{aj}$ to be zero. In summary, agent $a$ neither spends on goods nor earns from chores, and therefore her net income $(\sum_{j \in M^+} W_{aj} (P-p_j) - \sum_{j\in M^-} W_{aj} (P-p_j))$ has to be zero to satisfy (\ref{eq:amca}'). Since we have shown $p_j=P$ for all the chores the second term is already zero, and hence we have that for all goods $j$ if $W_{aj}>0$ then $p_j=P$. 

Under the strong connectivity assumption, there are agents who like the goods that $a$ brings, namely $j\in M^+$ such that $W_{aj}>0$. Since for all these goods we already proved $p_j=P$, applying the argument again, we can show that for all goods that these agents bring we have $p_j=P$. Applying this argument repeatedly using the strong connectivity of the {\em economy graph} assumption (Assumption 3 in Section~\ref{sec:splc}) we can propagate $p_j=P$ to all the goods $j\in M^+$. For the formal proof with all the details, see Lemma \ref{lem:no_p_0_goods} that argues for the more general SPLC utilities. For claims $(c)$ and $(d)$ too, see Lemmas \ref{lem:no_sec_ray_first_step} and \ref{lem:no_sec_ray_first_step_2} respectively.

\paragraph{No other secondary-ray} The above argument takes care of secondary-rays where $p_j=P$ for some item $j$ or $r_i=R$ for some agent $i$. Next, we show that there are no other secondary-rays, or in other words no secondary-rays where $\p < P$ and $\r < R$. Recall that a {\em secondary-ray} is an unbounded edge of polytope $\CP$ and the entire edge is a solution of LCP \eqref{eq:mixed-alcp} with $z>0$. Let $(\y^*,z^*)$ be the vertex where the ray starts, and $(\y',z')$ be its direction vector then the ray can be formally defined as $\cR=\{(\y^*,z^*) + \gamma (\y',z')\ |\ \gamma>0\}$. We will argue that the only possibility for $\cR$ is that it is the {\em primary-ray} where the algorithm started. And since the algorithm never revisits any point, this is a contradiction.

Observe that the non-negativity of variables in $\CP$ ensures $(\y',z')\ge 0$, and for any variable if its coordinate in the direction vector is positive then it increases to infinity on $\cR$. However, we have $\p<P$ and $\r<R$ on $\cR$, hence $\p',\r'=\zeros$. Then, it can be shown that \eqref{eq:amcb} and \eqref{eq:amca} together will not let any of the $f_{ij}$'s increase infinitely implying $\ff'=0$. In summary, $\y'=0$, and since the direction vector $(\y',z')$ can not be all zeros, $z'>0$.

At any point on $\cR$, we have $\y=\y^*+\gamma \y' =\y^*$ and $z>0$. In fact, $z$ goes to infinity on $\cR$ while $\y$ is fixed to $\y^*$. Therefore, inequalities with $(-z)$ terms, namely \eqref{eq:amca}, \eqref{eq:amcg}, \eqref{eq:aobg}, are all strict on $\cR$, and therefore their paired-up variables have to be zero. This gives that $\r^*=0$, $p^*_j=0,\ \forall j\in M^+$, and $f^*_{ij}=0,\ \forall i\in N, \forall j\in M^+$. Then, for $f^*_{ij}$'s for chores, observe that \eqref{eq:aobb} is also strict since $\forall i\in N, \forall j \in M^-,\ (P-p_j) \le P < D_{ij} R$ and therefore $f^*_{ij}=0$. This makes all the \eqref{eq:amcb} strict since $(P-p_j)>0$ for all items $j$, and in turn to satisfy (\ref{eq:amcb}') $p^*_j=0$, for all chores $j\in M^-$. In summary, we get $\y^*=0$ at the vertex of the ray. 

Using $\y'=0$ and $\y^*=0$ we have that on the entire ray $\cR$, $\y=\y^*+\gamma\y'=0$. Further, $z'>0, z^*\ge0$ implies $z=z^*+\gamma z\ge0$. However, by definition a ray with $\y=0$ and $z\ge 0$ is the {\em primary-ray} where the algorithm starts and can never revisits. This contradicts that the algorithm terminates on a {\em secondary-ray} with $\p < P$ and $\r < R$. We refer to Theorem \ref{thm:no_secondary_rays} for the formal proof under more general SPLC utilities.

Our construction also implies odd number of equilibria for nondegenerate instances. Starting from the primary-ray, the algorithm is guaranteed to reach an equilibrium. All other equilibria are paired up because the above arguments also imply that if we start the algorithm from any equilibrium point by relaxing $z=0$, it will end up on either the primary-ray or another equilibrium (see Theorem \ref{thm:odd_number}).

Putting everything together, we get (discussed in Section \ref{sec:results} for SPLC utilities): 
\begin{theorem}[Informal]
\label{thm:main-algo}
For mixed-manna under linear utilities, there is a complementarity pivot algorithm that finds a competitive equilibrium. Thereby, implying that the problem is in PPAD and the number of equilibria (up to scaling) for nondegenerate instances is odd. 
\end{theorem}

\section{Separable Piecewise Linear Concave Utilities}\label{sec:splc}
Recall from Section \ref{sec:prelim} that the set of agents and items are respectively denoted by $N$ and $M$. Agent $i\in N$ has $W_{ij}$ amount of item $j\in M$, and it is without loss of generality to assume that total amount of every item is $1$, i.e., $\sum_i W_{ij} = 1, \forall j$. In this section, we consider additively separable piecewise linear concave (SPLC) utility functions. That is, agent $i$'s utility function is additively separable over the items $u_i(\x_i) = \sum_{j\in M} u_{ij}(x_{ij})$, where for each agent $i$ and each item $j$, the function $u_{ij}:\bR_+ \rightarrow \bR$ is monotone piecewise linear and concave. The function is either non-negative and increasing representing a {\em good}, or it is non-positive and decreasing representing a \emph{chore/bad}. We call each linear piece of $u_{ij}$ a segment. Let $|u_{ij}|$ be the number of segments of $u_{ij}$, and let the triple $(i,j,k)$ denote the $k$-th segment. The slope of a segment gives the utility received per each additional unit of the item. Let $(i,j,k)$ be a segment with domain $[a,b] \subseteq \bR_+$ and slope $c$. Define $U_{ijk} = c$, and $L_{ijk} = b-a$. Note that the length of last segment is infinite. However, since there is unit amount of each item, we can assume without loss of generality that the length of the last segment is 1 plus some small constant. Note that linear is a special case of SPLC where each $u_{ij}$ has exactly one segment with infinite length. 

Our assumptions on the function $u_{ij}$ implies the following. If agent $i$ receives positive utility from item $j$, then $U_{ijk} > U_{ijk'} \geq 0$ for all $k < k'$, capturing the standard economic assumption of decreasing marginal returns on goods. Otherwise, $0 \ge U_{ijk} > U_{ijk'}$ for all $k <k'$ which models scenarios where the cost of completing a chore increases with the percentage required to be performed, e.g., cutting emissions from a plant. In the latter case, we use the notation $D_{ijk} = |U_{ijk}|$ for agent $i$'s disutility on the $k$-th segment of $u_{ij}$. Figure \ref{fig:splc} provides an illustration of SPLC utility functions.

\begin{figure}
	\centering
	\begin{subfigure}{.5\textwidth}
		\centering
		\includegraphics[scale=.58]{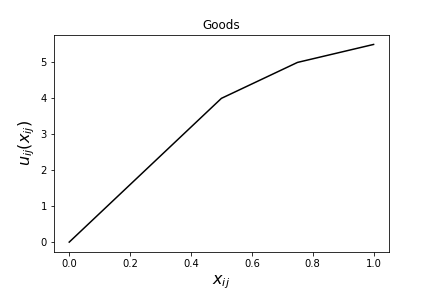}
	\end{subfigure}%
	\begin{subfigure}{.5\textwidth}
		\centering
		\includegraphics[scale=.58]{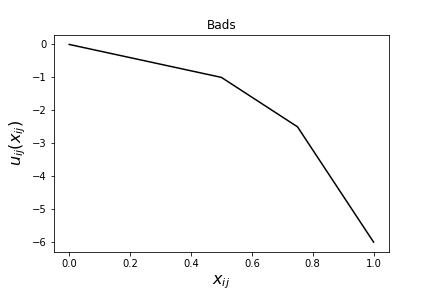}
	\end{subfigure}
	\caption{An example of SPLC utility functions for a good and a bad.}
	\label{fig:splc}
\end{figure}

\paragraph{Identifying Goods and Bads} \label{sec:goods_and_bads}
We begin with an important observation. Examining the first segment of each agent's utility function reveals the sign of the item prices at equilibrium. If there exists an agent $i\in N$ such that $ U_{ij1} > 0$, then $p_j \ge 0 $. This follows since agent $i$ would demand infinite amount of item $j$ if $p_j<0$, and then demand will not meet supply. Therefore, in \emph{any} equilibrium, if there exists an agent $i$ such that $U_{ij1} > 0$, then $p_j \geq 0$. Similarly, if $U_{ij1} \leq 0, \ \forall i\in N$, then $p_j \leq 0$, as at any positive price the demand of $j$ is zero. In view of the above, we refer to items with non-negative price as goods, and items with non-positive price as bads.
Here, a negative price for a bad implies an agent can {\em earn} by doing (consuming) the chore.

We can further refine the above observations to identify situations where there exists an equilibrium where an item's price is 0. For any good $j$, i.e., $p_j \geq 0$, we define the desire for $j$ as
\begin{equation}\nonumber
	\text{desire}_j = \sum_{i\in N}\sum_{k: U_{ijk}>0} L_{ijk}.
\end{equation}
In words, $\text{desire}_j$ is the maximum possible demand for good $j$ at any price $p_j > 0$. Suppose that $\text{desire}_j < 1$, then observe that there exists an equilibrium where $p_j = 0$, since there is a unit amount of each item. Thus, for any good $j$ with $\text{desire}_j < 1$ we may set $p_j = 0$, allocate the segments which provide positive utility for agents, i.e., $U_{ijk} > 0$, and assign any remaining fraction of the good to any zero utility segments.

Similarly, for any bad $j$ , i.e., $p_j \leq 0$, we define the indifference to $j$ as
\begin{equation}\nonumber
	\text{indifference}_j = \sum_{i\in N} \sum_{k:U_{ij1} = 0} L_{ij1}.
\end{equation}
The indifference to $j$ is the maximum amount of $j$ that can be assigned without causing any agent to lose utility. If $\text{indifference}_j \geq 1$, then \emph{all} equilibria set $p_j = 0$, and the item can be allocated among the agents along segments with $U_{ijk} = 0$. 

Henceforth, we assume that desire for every good is more than $1$ and indifference to every bad is less than $1$. Note that spending on bads `costs' a negative amount of money, since the price is negative for any bad. The natural economic interpretation is as follows. Suppose agent $i$ accepts some portion of bad $j$ she dislikes. As the price of $j$ is negative, this decreases her overall spending. Equivalently, she increases her budget by accepting responsibility for handling some universally disliked chore in order to spend more on goods she enjoys. Thus, the negative spending on bads can be viewed as receiving payment on some chore $j$ to increase the agent's budget. 

\paragraph{Characterizing Optimal Bundles.} \label{sec:optimaL_bundles}
At any prices, for each $u_{ij}$ function, clearly segment $k\ge 1$ is more attractive to agent $i$ than any later segment $k'>k$ due to the concavity of $u_{ij}$. Therefore, even if agent $i$ is allowed to buy ``segments'' of $u_{ij}$, she will buy them in increasing order. Formally, given a vector of prices $\p$, an optimal bundle of items for agent $i$, {\em i.e.,} the bundle that maximizes her utility subject to the {\em budget constraint}, solves the following linear program (LP).
\begin{align*}
\max \ &\sum_{j,k} U_{ijk} x_{ijk} \ \ \  \text{s.t.} \ \sum_{j,k}x_{ijk} p_j \leq \sum_{j} W_{ij} p_j;\ \ \  0\leq x_{ijk} \leq L_{ijk}, \ \forall (i,j,k), 
\end{align*}
where $x_{ijk}$ is the fraction of item $j$ allocated to agent $i$ on the $k$'th segment of $u_{ij}$. However, we require a more explicit characterization for later analysis. 

For any good, $p_j > 0$, define the bang per buck $(bpb)$ of agent $i$ on segment $(j,k)$ as
\begin{equation*}
	bpb_{ijk} = \frac{U_{ijk}}{p_j}.
\end{equation*} 
Note that, $bpb_{ijk}$ is the utility gained per unit spending on the $k$th segment of good $j$. Similarly, for any bad, $p_j < 0$, define the pain per buck $(ppb)$ of agent $i$ on segment $(j,k)$ as 
\begin{equation*}
	ppb_{ijk} = \frac{U_{ijk}}{p_j}.
\end{equation*}
Note that, for a bad $j$ since $p_j<0$ and $U_{ijk}\le 0$, we have $ppb_{ijk} \ge 0$ and it is the disutility per unit earning on the $k$th segment of bad $j$.

Intuitively, optimal bundles for any agent consist of segments with maximum $bpb$ for goods, which yield highest utility per unit spending, and minimum $ppb$ for bads, which minimizes disutility per unit spending. This can be easily verified through KKT conditions on the above LP. These segments may be computed as follows. Sort agent $i$'s segments for goods in decreasing order of $bpb_{ijk}$, and increasing order of $ppb_{ijk}$ for bads. Define the equivalence classes $G_1,\dots,G_l$ for goods with equal $bpb_{ijk}$, and $B_1, \dots, B_{l'}$ with equal $ppb_{ijk}$ for bads. Given the prices $\p$, each segment in $G_d$ adds an equal amount of utility per unit spending, while each segment in $B_{d'}$ adds an equal amount of disutility per unit earning. Obviously, agent $i$ demands $G_d$'s and $B_{d'}$'s in the increasing order to maximize her utility subject to the budget constraint. By abuse of notation, we will use $bpb(G_d)$ ($ppb(B_{d'})$) to denote the $bpb$ (resp. $ppb$) of the segments in equivalence class $G_d$ (resp. $B_{d'}$). 

Since an agent's utility decreases by consuming chores, she would consume one only if she needs the money earned to either satisfy her budget constraint (pay for the chores she owns), or use it to buy goods that (over) compensate for the disutility. Therefore, for agent $i$ if $\sum_{j\in M} W_{ij} p_j \ge 0$ then she consumes a segment from $B_{d'}$ only if there exists a $G_d$ such that $bpb(G_d) \ge ppb(B_{d'})$. If the latter inequality is strict, then agent $i$ would choose to accept as much of bads as possible from $B_{d'}$ to buy goods from $G_d$.

Suppose agent $i$ stops buying goods and bads at equivalence classes $G_d$ and $B_{d'}$ respectively -- $G_d$ (resp. $B_{d'}$) is the first partition that is not fully consumed. We note that, if $\sum_{j\in M} W_{ij} p_j <0$ then agent $i$ may consume only chores to earn the desired money. For all $k < d$ and $k' < d'$ we call the segments of equivalence classes $G_k$ and $B_{k'}$ {\em forced}. All the segments of equivalence classes $G_d$ and $B_{d'}$ are called {\em flexible}. And the for all $k>d$ and $k'>d'$ we call the segments of $G_k$ and $B_{k'}$ {\em undesirable}. For all agents, $ppb \geq bpb$ in their flexible partition.

\paragraph{Assumptions} \label{sec:Assumptions}
Even in the special case of all goods, equilibria in exchange setting need not exist~\cite{DevanurGV16}. We need to assume certain sufficiency conditions to allow an equilibrium to exist. We note that our conditions follows the previous works of~\cite{GargMSV15,ChenDDT09,ChenT09} that consider only goods, and is one of the weakest sufficiency conditions to guarantee an equilibrium exists in the case of all goods. First, we include our basic assumptions. 
\begin{condition}\label{cond1}
Each agent brings a positive amount of some good and positive amount of some bad.
\end{condition}

\begin{definition}
For any good $j\in M$, we say that agent $i$ is \emph{non-satiated} for $j$, if $U_{ijk} > 0$ where $k$ is the last segment of good $j$. 
\end{definition}

\begin{definition}
Define the economy graph as a directed graph $G$ with vertices $N$, with directed edges from $i$ to $j$ if agent $i$ is non-satiated for some good $l$ that agent $j$ brings. We call the instance \emph{strongly connected} if the economy graph $G$ is strongly connected. 
\end{definition}

\begin{condition}\label{cond3}
Economy graph of the input instance is strongly connected.
\end{condition}

Note that, Condition~\ref{cond3} is needed to ensure the existence of equilibrium even for the goods only case~\cite{GargMSV15,ChenDDT09}. We refer to Conditions \ref{cond1} and \ref{cond3} together as \emph{strong connectivity}. We will show that our algorithm in Section~\ref{sec:algo} converges to a competitive equilibrium under strong connectivity, hence we get a constructive proof of the existence. Observe that this implies the existence of equilibrium in all instances of the Fisher (and hence CEEI) setting under linear utilities and under non-satiated SPLC utilities. 

\subsection{LCP Formulation for All Bads}\label{sec:alL_bads}
In this section, we derive a linear complementary program to capture competitive equilibria for the case when mixed manna contains only bads, i.e., $U_{ijk}\le 0, \forall (i, j, k)$. We build on the approaches of Eaves~\cite{Eaves76} and Garg et al.~\cite{GargMSV15} for only goods. Our task consists of two steps. First, we need to design constraints to ensure that market clears (i.e., all bads are fully allocated, and each agent earns exactly the required budget). Second, we need to ensure agents earn their budget on optimal bundles of bads. We note that most of the proofs are deferred to Appendix \ref{sec:proofs}. 

The first problem, market clearing, is straightforward and does not even require complementarity. Note that the LCP formulation requires non-negative variables. However, prices and spending on bads are negative. Therefore, we create  non-negative variables $p_j$ for all $j\in M$, and $f_{ijk}$ for all segments $(i,j,k)$. For every $U_{ij1}\le 0$ for a good $j$, then we set $f_{ijk}:=0$ at the beginning itself, and we do not introduce the corresponding variables in our formulation. We will use $(-p_j)$ as the price of bad $j\in M$, and $(-f_{ijk})$ as the amount agent $i$ earns on the segment $(j,k)$. We also let $D_{ijk} = |U_{ijk}|$ denote $i$'s disutility on segment $(i,j,k)$. Also, for each agent $i$, we create a variable $r_i$. Eventually, $1/r_i$ will be the pain per buck of agent $i$'s flexible partition. 

Let $\perp$ denote a complementarity constraint between the inequality and the variable (e.g., $\sum_{j} W_{ij} p_j \leq \sum_{j,k} f_{ijk}\ \perp\ r_i$ is a shorthand for $\sum_{j} W_{ij} p_j \leq \sum_{j,k} f_{ijk};\ r_i \ge 0;\ r_i (\sum_{j} W_{ij} p_j - f_{ijk}) =0$). We ensure market clearing with the following constraints, where each variable is paired with a constraint by complementarity conditions to yield a standard LCP formulation. 
\begin{subequations}\label{lcp:1}
\begin{eqnarray}
	\forall i \in N: & \ \sum_{j} W_{ij} p_j \leq \sum_{j,k} f_{ijk} & \perp  \ \ r_i  \label{eq:budget_1}\\
	\forall j\in M: & \ \sum_{j,k} f_{ijk} \leq p_j & \perp \ \ p_j \enspace .\label{eq:spending_1}
\end{eqnarray}
\end{subequations}

We refer to each constraint by the equation number, and the corresponding complementarity condition by the equation number prime. 
Next, we design constraints to ensure agents purchase optimal bundles of bads. Recall the characterization of optimal bundles from Section~\ref{sec:optimaL_bundles}. Let $(i,j,k)$ be a segment of agent $i$'s flexible partition. We want the variable $r_i$ to satisfy
\begin{equation} \label{eq:lambda_ppb}
	ppb_{ijk} = \frac{1}{r_i} = \frac{D_{ijk}}{p_j} >0\enspace .
\end{equation}
For any forced segment $(i,j',k')$, we have $ppb_{ij'k'} < ppb_{ijk}$. We compensate for this by adding another variable $s_{ij'k'}\geq 0$ for each segment $(i,j',k')$ of $i$'s utility function. We want $s_{ijk} > 0$ for any forced segment, and $s_{ijk} =0$ otherwise. The new variables can be interpreted as supplementary prices for each segment of $i$'s utility function. This leads to the following constraints and complementarity conditions
\begin{align}
	\hspace{4cm} \forall (i,j,k): &  \ \ \ p_j -s_{ijk}\leq D_{ijk} r_i & \perp \ \ f_{ijk} \ \ \hspace{4cm}\label{eq:ppb_1}\tag{8c} \\
	\hspace{4cm} \forall (i,j,k): & \ \ \ \ \ f_{ijk} \leq L_{ijk} p_j  & \perp \ \ s_{ijk} \enspace . \hspace{4cm}\label{eq:segment_1}\tag{8d}
\end{align}
Note that complementarity condition (\ref{eq:segment_1}') ensures that forced segments are fully purchased. 
The next lemma shows that LCP~\eqref{lcp:1} captures all the competitive equilibrium (see Appendix \ref{sec:lem:clearing_1} for proof). 

\begin{lemma} \label{lem:market_solution_1}
	Any competitive equilibrium gives a solution to LCP~\eqref{lcp:1}.
\end{lemma}

LCP~\eqref{lcp:1} suffers from a serious problem. The vector $\q$ of in \eqref{0lcp} representation contains all zeros, meaning that it admits the trivial solution $\p = \f =\r = \s = \0$. We address this issue by a change of variables. For any equilibrium price vector $\p^*$, there exists a largest price (in magnitude) $P = \max_j |p_j^*|$. Since equilibrium prices are scale invariant, we can assume that $P$ is a positive constant. Changing variables to define prices relative to $P$ makes $-(P - p_j)$ the price of bad $j\in M$. Observe that bounding the maximum price (in absolute value) also bounds each agent's $ppb$ in her flexible partition, i.e., $1/ppb \leq P /D_{min}, \ \forall i\in N$, where $D_{min}= \min_{i,j,k: D_{ijk}>0} D_{ijk}$. Let a constant $R > (P /D_{min})$, and replace $r_i$ with $(R - r_i)$. That is, we want $1/(R-r_i) = ppb_i$ for $i$'s flexible partition. Substituting the new variables (i.e., $p_j$ with $(P-p_j)$ and $r_i$ with $(R-r_i)$) into LCP~\eqref{lcp:1} yields 
\begin{subequations}\label{lcp:2}
\begin{eqnarray}
	\forall i \in N: & \ -\sum_{j} W_{ij} p_j - \sum_{j,k} f_{ijk} \leq -P \sum_{j} W_{ij} & \perp \ \ r_i \label{eq:budget_2} \\
	\forall j\in M: & \ \sum_{i,k} f_{ijk} + p_j \leq P & \perp  \ \ p_j \label{eq:spending_2}\\
	\forall (i,j,k): & D_{ijk} r_i - p_j -s_{ijk}\leq D_{ijk}R - P & \perp \ \ f_{ijk} \label{eq:ppb_2} \\
	\forall (i,j,k): & f_{ijk} + L_{ijk}p_j \leq L_{ijk} P & \perp \ \ s_{ijk} \label{eq:segment_2} \enspace .
\end{eqnarray} 
\end{subequations}

LCP~\eqref{lcp:2} still allows one non-competitive equilibrium. Observe that setting $p_j = P, \ \forall j\in M$, $r_i=R,\ \forall i\in N$, and all other variables $(\f,\s)=\0$, solves LCP~\eqref{lcp:2}, but this solution is not a competitive equilibrium. Rather, this degenerate `equilibrium' proposes to make the price of each bad 0, since the price of bad $j$ is $-(P-p_j)$. In turn, this makes each agent's budget equal to 0, and prevents them from earning on anything. Ultimately, this leaves all bads unallocated and the market doesn't truly clear. We call this the degenerate solution. We show in Section \ref{sec:algo} that the algorithm never reaches this solution. Assuming $p_j < P, \ \forall j\in M$ and $r_i < R, \ \forall i\in N$, it is straightforward to verify that Lemma \ref{lem:market_solution_1} still holds. 

\begin{theorem} \label{thm:equiL_soL_correspondence}
	The solutions to LCP~\eqref{lcp:2} with $p_j < P, \ \forall j\in M$ and $r_i < R, \ \forall i\in N$ exactly captures all competitive equilibria (up to scaling).
\end{theorem}

The proof of the above theorem is in Appendix~\ref{sec:lem:equiv_sol}.

\subsection{LCP Formulation for Mixed Manna}
We now extend the LCP formulation to the general mixed manna case. Given the known LCP formulation for SPLC utilities for all goods due to Garg et al.~\cite{GargMSV15}, and LCP~\eqref{lcp:2} for all bads, a natural question is: Can we simply combine an LCP for goods and an LCP for bads to obtain an LCP for mixed manna? Note that this treats the mixed manna case as two separate subproblems: one for goods and one for bads. Such a formulation requires separate budget constraints for goods and bads, i.e., each agent's spending on goods (bads) is at least as much as her earnings on goods (bads), similar to constraint \eqref{eq:budget_1}. However, a simple example illustrates that, in general, this is not possible. \medskip

\begin{example}\label{example1}
Consider an instance with two agents $A$ and $B$, and two items $1$ and $2$. Agents' utilities are as follows: $u_A(x_A) = x_{A1} - 2 x_{A2}$, and $u_B(x_B) = x_{B1} - 3 x_{B2}$. Assume each agent brings an equal amount of each item, i.e., $W_{A1} = W_{A2} = W_{B1}= W_{B2} =0.5$. 

There are a few important things to note. Both agents like item 1, so it is a good and $p_1 >0$, and since both agents dislike item 2, it is a bad and $p_2 < 0$. Clearly, both agents must purchase some of bad 2 at equilibrium. A portion of item 1 can not be purchased by both agents since optimal bundles require $bpb = ppb$. Thus, if both agents purchase some of item 1, then $u_{A1}/p_1 = u_{A2}/p_2$, or $p_2 = -2p_1$, but we also have the requirement $p_2 = -3p_1$, a contradiction. Therefore, only one agent purchases good 1. One can verify that the prices $p_1 = 2$ and $p_2 = -4$, along with the allocation $x_{A1} = 1$, $x_{A2} = 3/4$, $x_{B1} = 0$, and $x_{B2} = 1/4$ are an equilibrium where each agents' initial budget is set to $-1$. Note that agent 1's total spending on the good is 2, and agent 2's spending on the good is 0. However, the total value of good in each agent's initial bundle is 1. Thus, neither agent's spending on the good equals the value of good in her initial endowment. 
\end{example}

\subsubsection{Basic Formulation}
Similar to Section \ref{sec:alL_bads}, we start by designing an LCP whose solutions capture competitive equilibria. This requires that the market clears, and that agents purchase optimal bundles of goods and bads. Although specialized to the case of all bads, the derivation of LCP~\eqref{lcp:2} in Section \ref{sec:alL_bads} provides the basic framework needed to handle the general mixed manna setting. As discussed in Section \ref{sec:goods_and_bads}, we can identify which items are goods and which are bads by examining the sign of the utility for the first segment of each agent. Note that when dealing with mixed manna, we assume the strong connectivity conditions (see Conditions \ref{cond1} and \ref{cond3}).

For clarity, we first write all complementarity conditions with minimal change of variables. Let $M^-$ and $M^+$ denote the set of bads and goods respectively. For all $j\in M^-$, prices, spending, and utilities are negative. We introduce non-negative variables $p_j$ and $f_{ijk}$ for all $j\in M$. We interpret $p_j$ as the price of good $j\in M^+$, and $(-p_j)$ as the price of bad $j\in M^-$. Similarly, $f_{ijk}$ gives agent $i$'s spending on the segment $(i,j,k)$ for good $j$, while $-f_{ijk}$ is $i$'s spending on the segment $(i,j,k)$ of bad $j\in M^-$.
We ensure market clearing with the following complementarity conditions
\begin{subequations}\label{lcp:3}
\begin{eqnarray}
\forall i \in N: & \displaystyle\sum_{k,j\in M^+} f_{ijk} - \sum_{k,j\in M^-} f_{ijk} \leq \sum_{j\in M^+} W_{ij} p_j - \sum_{j\in M^-} W_{ij} p_j  & \perp \ \ r_i\label{eq:budget_3} \\
\forall j\in M^- : & \displaystyle\sum_{i,k} f_{ijk} \leq p_j & \perp \ \ p_j\label{eq:spending_3_bad}\\
\forall j\in M^+ : & p_j \leq \displaystyle\sum_{i,k} f_{ijk} & \perp \ \ p_j\enspace . \label{eq:spending_3_good}
\end{eqnarray}
\end{subequations}
Note that we treat the spending constraints for bads \eqref{eq:spending_3_bad} and goods \eqref{eq:spending_3_good} differently. Further, if all items are bads, then we recover \eqref{eq:budget_1} and \eqref{eq:spending_1}. 

\begin{lemma} \label{lem:clearing_2}
	If $\p^*$ is an equilibrium price vector, then $\exists \ \f$ such that $(\p,\f)$ satisfies \eqref{eq:budget_3}, \eqref{eq:spending_3_bad}, and \eqref{eq:spending_3_good}, where $\p = |\p^*|$. Further, if $\p$ and $\f$ satisfy \eqref{eq:budget_3},
	\eqref{eq:spending_3_bad}, \eqref{eq:spending_3_good} and $\p > 0$, then the market clears.
\end{lemma}
Proof of the above lemma is in Appendix \ref{sec:lem:clearing_2}. The next step is to make sure agents purchase optimal bundles. Let $(i,j,k)$ be a segment of agent $i$'s flexible partition for item $j$. We want the variable $r_i$ to satisfy, where $D_{ijk}=|U_{ijk}|$ for bad $j$,
\begin{equation}\label{eq:lambda_mbb}
\frac{1}{r_i} = \frac{U_{ijk}}{p_j}\ \ \mbox{if $j\in M^+$, and }\ \frac{1}{r_i} = \frac{D_{ijk}}{p_j}\ \  \mbox{if $j\in M^-$.}
\end{equation}
Recall that forced segments of goods and bads correspond to slightly different conditions. For any forced segment $(i,j,k')$ of bad $j$, we have $ppb_{ijk'} < ppb_{ijk}$. For any forced segment $(i,j,k')$ of good $j$, we have $bpb_{ijk'} > bpb_{ijk}$. Again, we compensate for this by introducing a variable $s_{ijk} \geq 0$ into each segment $(i,j,k)$ of $i$'s utility function, leading to the following complementarity conditions
\begin{align}
\hspace{3cm}\forall j\in M^-, \ \forall (i,j,k): & \ \  p_j -s_{ijk}\leq D_{ijk} r_i & \perp \ \ f_{ijk}\ \ \hspace{3cm}\label{eq:mbb_3_bad} \tag{11d}\\
\hspace{3cm}\forall j\in M^+, \ \forall (i,j,k): & \ \ U_{ijk} r_i \leq p_j +s_{ijk} & \perp \ \ f_{ijk}\ \ \hspace{3cm}\label{eq:mbb_3_good} \tag{11e}\\
\hspace{3cm}\forall (i,j,k): & \ \ f_{ijk} \leq L_{ijk} p_j & \perp \ \  s_{ijk}\enspace . \hspace{3cm}\label{eq:segment_3}\tag{11f}
\end{align}
Observe that, if all items are bads, then we recover \eqref{eq:ppb_1} and \eqref{eq:segment_1}. 

\begin{lemma} \label{lem:market_solution_2}
	Any competitive equilibrium of mixed manna gives a solution to LCP~\eqref{lcp:3}.
\end{lemma}

The proof of the above lemma is in Appendix \ref{sec:lem:mkt_sol_2}. Similar to the case of all bads, LCP~\eqref{lcp:3} admits solutions that are not competitive equilibria, e.g., the trivial solution $\p = \f = \r = \s = \0$. We use the same change of variables as before. We fix a maximum price (in absolute value) $P$, and define the relative prices: $(P - p_j)$ for all goods $j\in M^+$, and $-(P-p_j)$ for all bads $j\in M^-$. This bounds each agent's $ppb$ or $bpb$ in her flexible partition, i.e., $1/bpb, 1/ppb \leq P /U_{min}$ where $U_{min}=\min_{i,j,k: U_{ijk}\neq 0} |U_{ijk}|$. Using this we define a constant $R> P /U_{min}$ and replace $r_i$ with $(R - r_i)$. That is, we want $\frac{1}{(R-r_i)}$ to represent $ppb$ and $bpb$ of agent $i$'s flexible partitions. Substituting $p_j$ with $(P-p_j)$ for each $j\in M$ and $r_i$ with $(R-r_i)$ for each agent $i\in N$ into LCP~\eqref{lcp:3} yields

\begin{subequations}\label{lcp:4}
{\small
\begin{eqnarray}
\hspace{-0.5cm}\forall i \in N:  \displaystyle\sum_{j\in M^+} W_{ij} p_j- \sum_{j\in M^-} W_{ij} p_j + \sum_{k,j\in M^+} f_{ijk} - \displaystyle\sum_{k,j\in M^-} f_{ijk} \leq P ( \sum_{j\in M^+} W_{ij}-\sum_{j\in M^-} W_{ij}) & \perp \ \ r_i\label{eq:budget_4} 
\end{eqnarray}
}
\begin{eqnarray}
\forall j\in M^-: & \displaystyle\sum_{i,k} f_{ijk} + p_j  \leq P & \perp \ \ p_j\label{eq:spending_4_bad}\\
\forall j\in M^+: & -\displaystyle\sum_{i,k} f_{ijk} - p_j  \leq -P & \perp \ \  p_j\label{eq:spending_4_good}\\
\forall j \in M^-, \ \forall i,k: & D_{ijk} r_i - p_j -s_{ijk}\leq D_{ijk}R - P & \perp \ \ f_{ijk}\label{eq:mbb_4_bad} \\
\forall j \in M^+, \ \forall i,k: & -U_{ijk} r_i + p_j -s_{ijk}\leq P-U_{ijk}R & \perp \ \ f_{ijk}\label{eq:mbb_4_good} \\
\forall (i,j,k): & f_{ijk} + L_{ijk}p_j \leq L_{ijk} P & \perp \ \ s_{ijk}\label{eq:segment_4} \enspace .
\end{eqnarray} 
\end{subequations}
In the case of all bads, LCP~\eqref{lcp:4} is equivalent to LCP~\eqref{lcp:2} from Section \ref{sec:alL_bads}.

Similar to LCP~\eqref{lcp:2}, LCP~\eqref{lcp:4} still allows (at least) one non competitive equilibrium. By setting $p_j = P, \ \forall j\in M$, $r_i = R, \ \forall i \in N$, and all other variables $(\f,\s)=\0$ we get a solution to LCP~\eqref{lcp:4}. However, this solution does not correspond to a competitive equilibrium, but rather, a degenerate solution where all prices are zero and no items are allocated. We show in Section \ref{sec:algo_convergence} that the algorithm never reaches this degenerate solution. Assuming $p_j < P, \ \forall j\in M$, and $r_i < R, \ \forall i\in N$, it is straightforward to verify that Lemmas \ref{lem:clearing_2} and \ref{lem:market_solution_2} still hold -- see Appendix \ref{sec:optimaL_bundles_2} for the formal proof.

\begin{lemma} \label{lem:optimaL_bundles_2}
	In any solution to LCP~\eqref{lcp:4} with $p_j < P, \ \forall j\in M$, and $r_i < R, \ \forall i \in N$, all agents receive an optimal bundle and market clears w.r.t. prices $\p^*$ where $p^*_j=(P-p_j)$ for all goods $j\in M^+$ and $p^*_j=-(P-p_j)$ for all chores $j\in M^-$. 
\end{lemma}

The next theorem follows using Lemmas~\ref{lem:market_solution_2} and~\ref{lem:optimaL_bundles_2} together with the way LCP~\eqref{lcp:4} is constructed from LCP~\eqref{lcp:3}. Its proof is in Appendix~\ref{athm:equiL_soL_correspondence_2}.

\begin{theorem}\label{thm:equiL_soL_correspondence_2}
	The solutions to LCP~\eqref{lcp:4} with $p_j < P, \ \forall j\in M$, and $r_i < R, \ \forall i\in N$, exactly captures competitive equilibrium of mixed manna (up to scaling).
\end{theorem}

\subsubsection{Augmented LCP and Nondegeneracy}\label{sec:alcp}
Observe that LCP~\eqref{lcp:4} has the same form as \eqref{0lcp} in Section \ref{sec:prel-lcp}. We now give the augmented LCP for this problem. By choice of $P$ and $R>P/\min_{i,j,k: U_{ijk}\neq 0}|U_{ijk}|$ we have that for all bads $j\in M^-$, $D_{ijk}R-P>0, \ \forall i,k$. This makes the right hand side of \eqref{eq:mbb_4_bad} positive for all bads. Standard LCP techniques~\cite{Eaves76,GargMSV15,GargMV18} add the variable $z$ only in constraints with negative right hand side. We make two changes. First, we include the variable $z$ to any constraints with negative right hand side, and \emph{all} budget constraints \eqref{eq:budget_4}. Second, when adding $z$ into spending constraints for goods \eqref{eq:spending_4_good}, we use a coefficient $\delta_j = 1+ \epsilon_j$, where $\epsilon_j$  is a uniform random number from $(0,1/m)$. This change is necessary to ensure that the polyhedron corresponding to the augmented LCP remains nondegenerate, as discussed shortly. Adding $z$ to any constraints with negative right hand side, and \emph{all} budget constraints \eqref{eq:budget_4} yields
\begin{subequations}\label{lcp:5}
{\small 
\begin{eqnarray}
\forall i \in N: \  \sum_{j\in M^+} W_{ij} p_j-\sum_{j\in M^-} W_{ij} p_j + \sum_{k,j\in M^+} f_{ijk} - \sum_{k,j\in M^-} f_{ijk} - z \leq P ( \sum_{j\in M^+} W_{ij}-\sum_{j\in M^-} W_{ij})\ \perp\ r_i\label{eq:budget_4a}
\end{eqnarray}
}
\begin{eqnarray}
\forall j\in M^-: & \displaystyle\sum_{i,k} f_{ijk} + p_j  \leq P & \perp \ \ p_j\label{eq:spending_4a_bad}\\
\forall j\in M^+: & -\displaystyle\sum_{i,k} f_{ijk} - p_j -\delta_j z \leq -P & \perp \ \  p_j\label{eq:spending_4a_good}\\
\forall j \in M^-, \ \forall i,k: & D_{ijk} r_i - p_j -s_{ijk}\leq D_{ijk}R - P & \perp \ \  f_{ijk}\label{eq:mbb_4a_bad} \\
\forall j \in M^+, \ \forall i,k: & -U_{ijk} r_i + p_j -s_{ijk} -z \leq P-U_{ijk}R & \perp \ \  f_{ijk}\label{eq:mbb_4a_good} \\
\forall (i,j,k): & f_{ijk} + L_{ijk}p_j \leq L_{ijk} P & \perp \ \  s_{ijk}\enspace . \label{eq:segment_4a}
\end{eqnarray} 
\end{subequations}

Let $\cP$ be the polyhedron corresponding to LCP~\eqref{lcp:5}. Lemke's algorithm requires nondegeneracy of the polyhedron $\cP$, i.e., if $\cP$ is defined on $k$ variables, then at any $d$ dimensional face of $\cP$ exactly $(k-d)$ inequalities hold with equality. 
In that case, the solutions of LCP~\eqref{lcp:5} are paths and cycles on the 1-skeleton of $\cP$. However, there is an inherent degeneracy present when $z=0$, i.e., solutions of LCP~\eqref{lcp:4}. 
\medskip

\noindent{\em Inherent Degeneracy in LCP~\eqref{lcp:4}.} Summing \eqref{eq:budget_3} over all $i \in N$, and \eqref{eq:spending_3_bad} over all $j\in M^-$ and \eqref{eq:spending_3_good} over all $j \in M^+$ yields two identical equations; see the proof of Lemma \ref{lem:clearing_2} for details. That is, there is an inherent degeneracy in $\cP$. 
\medskip

Clearly, the inherent degeneracy of LCP~\eqref{lcp:4} is still present in $\cP$ when $z=0$. We need to show that no other degeneracies exist which crucially relies on $\delta_j = 1+\epsilon_j$ for all goods $j\in M^+$.\footnote{Suppose $\delta_j = 1, \ \forall j\in M^+$, and that $p_j = 0, \ \forall j\in M^+$, and $f_{ijk} = 0$ for all segments $(i,j,k)$ of each good. If $z=P$, then \eqref{eq:spending_4a_good} also holds with equality $\forall j \in  M^+$. Therefore, we have double labels for each good $j\in M^+$.}

If there is a degenerate vertex $\vv\in \cP$ with $z>0$, $p_j < P, \ \forall j\in M$, and $r_i < R, \ \forall i\in N$, then using the extra tight inequalities at $\vv$ we can derive a polynomial relation between the input parameters $\mathbf{U}$, $\mathbf{W}$ and $\mathbf{L}$, and $\epsilon_j$'s. Therefore, we get the following theorem (see Appendix \ref{sec:thm:non_deg_2} for the formal proof).

\begin{theorem} \label{thm:non_degenerate_2}
	If the instance parameters $\mathbf{U}$, $\mathbf{W}$, and $\mathbf{L}$ have no polynomial relation among them, then every vertex of $\mathcal{P}$ with $z > 0$, $p_j < P, \ \forall j\in M$, and $r_i < R, \ \forall i\in N$, is nondegenerate.
\end{theorem}

By the similar argument, together with Theorem \ref{thm:equiL_soL_correspondence_2} and the fact that solutions of LCP~\eqref{lcp:5} with $z=0$ are solutions of LCP~\eqref{lcp:4}, we get the following (see Appendix \ref{sec:thm:one_to_one_2} for a formal proof). 

\begin{theorem} \label{thm:one_to_one_2}
If the instance parameters $\mathbf{U}$, $\mathbf{W}$, and $\mathbf{L}$ have no polynomial relation among them, then solutions of LCP~\eqref{lcp:5} with $z=0$, $p_j < P, \ \forall j\in M$, and $r_i < R, \ \forall i\in N$, are in one to one correspondence with competitive equilibria.
\end{theorem}

\begin{remark}
A degenerate instance can be solved by perturbing the input parameters using the technique in~\cite{DuanGM16} that preserves the solution structure.
\end{remark}

\section{Algorithm}\label{sec:algo}
In Section~\ref{sec:alcp}, we designed the augmented LCP~\eqref{lcp:5} that permits use of Lemke's scheme, and we showed a one to one correspondence between competitive equilibria and solutions to LCP~\eqref{lcp:5} with $z=0$, and $p_j < P, \ \forall j\in M$ and $r_i < R, \ \forall i\in N$ so long as the polyhedron $\mathcal{P}$ defined by LCP~\eqref{lcp:5} is nondegenerate.  We now add the slack variable $v_a$ to the $a$-th constraint to of LCP~\eqref{lcp:5} to obtain
\begin{equation*}
A\y+\v-\c z = \q, \ \y \geq 0, \ \v \geq 0,\ z \geq 0, \ \text{and} \ \y \cdot \v = 0,
\end{equation*}
where $\c$ is the vector of coefficients for $z$ in LCP~\eqref{lcp:5}. Note that the condition $\v \geq 0$ follows from $\q-A\y \geq 0$. We call the new formulation LCP~(\ref{lcp:5}'). Of course when $v_k = (\q-A\y+\c z)_k =0$, the $k$-th constraint holds. Therefore, at any fully labeled vertex solution $S$ of the polyhedron defined by LCP~(\ref{lcp:5}'); see Section \ref{sec:prel-lcp}, either $v_k= 0$ or $y_k= 0$. At a double label, $v_k = y_k = 0$. Using the notation of Section \ref{sec:prel-lcp}, we let $\y = (\p,\f,\r,\s)$ be a vertex solution to LCP~(\ref{lcp:5}').

Recall from Section \ref{sec:prel-lcp} that Lemke's algorithm explores a certain path of the 1-skeleton of $\cP$, traveling from vertex solution to vertex solution along the edges of $\cP$. Note that we chose $R$ such that $R > P/\min_{i,j,k} |U_{ijk}|$, which ensures the right hand side of \eqref{eq:mbb_4a_bad} is positive, i.e., $D_{ijk}R - P >0$, for all segments $(i,j,k)$, $\forall i,k$, $\forall j\in M^-$, and that the right hand side of \eqref{eq:mbb_4a_good} is negative, i.e., $P-U_{ijk}R < 0$, for all segments $(i,j,k)$, $\forall i,k$, $\forall j\in M^+$. Further, for sufficiently large $R$, we have $\min_{i} P(\sum_{j\in M^+}W_{ij}-\sum_{j\in M^-} W_{ij}) > P- \max_{j\in M^+,i,k}U_{ijk}R$. Then, we get the primary ray (initial solution) $S_0$ by setting
\begin{equation*}\label{eq:S0_mixed}
	S_0 = \{ \y_0 =\0,\ z = \max_{j\in M^+,i,k} U_{ijk}R - P, \ \v_0 = \q + \c z   \}\enspace . 
\end{equation*}
Clearly, this initialization gives the unique double label $y_{(ijk)^*} = v_{(ijk)^*}=0$, for $(i,j,k)^* = \arg\max_{(j\in M^+,i,k)} U_{ijk}R - P$. 

Algorithm \ref{algo:mixed} gives a formal description of the Lemke's algorithm applied to LCP~(\ref{lcp:5}'). Assuming the input parameters $\mathbf{U}$, $\mathbf{W}$, and $\mathbf{L}$ have no polynomial relationship among them, Theorem \ref{thm:non_degenerate_2} guarantees that any vertex with $z>0$ is nondegenerate. Therefore, a unique double label, say $k$, such that $y_k = v_k = 0$, always exists. Algorithm \ref{algo:mixed} pivots at the double label by relaxing one constraint, and traveling along the corresponding edge of $\cP$ to the next vertex solution. 

\begin{theorem} \label{thm:mixed_converges}
	If the input parameters $\mathbf{U}$, $\mathbf{W}$, and $\mathbf{L}$ have no polynomial relationship among them, then Algorithm \ref{algo:mixed} terminates at a competitive equilibrium in finite time.
\end{theorem}

\begin{algorithm}[tb!] 
	\KwData{A, $\q$}
	\KwResult{A competitive equilibrium }
	$S \gets S_0$\;
	\While{$z>0$}{
		Let $k$ be the double label in solution $S$, i.e., $y_k = v_k = 0$. \;
		\eIf{$v_k$ just became 0}{
			Pivot by relaxing $y_k = 0$.\;
		}{
			Pivot by relaxing $v_k = 0$.\;
		}
		Let $S'$ be the next vertex solution to LCP~(\ref{lcp:5}') reached, $S\gets S'$\;
	}
	\caption{Algorithm for Competitive Equilibrium of Goods and Bads.}
	\label{algo:mixed}
\end{algorithm} 

\subsection{Convergence of Lemke's Algorithm}\label{sec:algo_convergence}
We now show that Algorithm \ref{algo:mixed} always finds an equilibrium. We note that, unlike earlier works that consider only good manna~\cite{Eaves76,GargMSV15}, our LCP formulation allows for secondary rays and one non-equilibrium solution. This makes the proof that Lemke's algorithm finds a competitive equilibrium significantly more complex.

For ease of presentation, we perform the subsequent analysis using LCP~\eqref{lcp:5}, i.e., without any slack variables. Let $\cP$ be the corresponding polyhedron. To verify Algorithm \ref{algo:mixed} terminates at a competitive equilibrium we need to examine two potential problems. First, we need to show that the algorithm never finds a secondary ray. Second, we need to show that, starting from the primary ray, Algorithm \ref{algo:mixed} never reaches the degenerate solution where $p_j = P, \ \forall j\in M$, and $r_i = R, \ \forall i \in N$. 

First, we consider secondary rays. Recall that a ray $\mathcal{R}$ is a unbounded edge of $\cP$ incident to the vertex $(\y^*,z^*)$ with $z^*>0$ 
\begin{equation*}
\mathcal{R} = \{ [\y^*,z^*] + \alpha [\y',z'] \ | \ \forall \alpha \geq 0 \}.
\end{equation*}
Clearly, all points on $\mathcal{R}$ solve LCP~\eqref{lcp:5}. Algorithm~\ref{algo:mixed} begins at the primary ray $S_0$, and all others are called secondary rays. The major issue is that if Algorithm \ref{algo:mixed} finds a secondary ray, then it fails to terminate. Observe that setting $p_j = P$ for some $j\in M^-$ leads to secondary rays. Suppose we set $p_j = P$ for some subset of bads $B\subseteq M^-$, and $p_j =0$ otherwise, and make all other variables $(\f,\r,\s) =\0$. Then, we may select sufficiently large $z^*$ to satisfy all constraints of the form \eqref{eq:budget_4a}, \eqref{eq:spending_4a_good}, and \eqref{eq:mbb_4a_good}. Let $\y^* = (\p,\f,\r,\s)$ be this vertex solution, and consider the ray $\mathcal{R} = [\y^*,z^*] +\alpha [\0,1]$ incident to $(\y^*,z^*)$. It is easily verified that $\mathcal{R}$ solves LCP~\eqref{lcp:5} for all $\alpha > 0$, and, therefore, is a secondary ray. We want to show the path traced by Algorithm \ref{algo:mixed} never reaches these problematic vertices. 

We begin with a simple observation. Notice that setting $p_j = P$ for any good requires that $z =0$, by \eqref{eq:spending_4a_good} and (\ref{eq:spending_4a_good}'). Therefore, Algorithm \ref{algo:mixed} stops at a vertex where any $p_j = P$, for any $j\in M^+$. We follow this result with a few useful facts.

\begin{claim} \label{claim:no_s_finaL_segment}
	Let $S$ be any solution to LCP~\eqref{lcp:5} with $p_j < P$, $\forall j \in M$. Pick any agent $i \in N$, and any item (good or bad) $j\in M$, and let $k = |u_{ij}|$ be $i$'s final segment for item $j$. If $j\in M^-$, then $s_{ijk} = 0$. If $j\in M^+$ and $p_j > 0$, then $s_{ijk} = 0$.
\end{claim}

\begin{proof}
	Recall that the length of the final segment $(i,j,k)$ is infinite, however, we set $L_{ijk} = 1 + \epsilon$, for some small $\epsilon > 0$ since there is a unit amount of each item. We consider two cases: $j$ is a bad or a good. First suppose $j \in M^-$, and for contradiction assume $s_{ijk} > 0$. By complementarity condition (\ref{eq:segment_4a}'), \eqref{eq:segment_4a} holds with equality. Then, $f_{ijk} = L_{ijk} (P-p_j) > 0$, since $p_j < P$ at $S$. Consider the constraint \eqref{eq:spending_4a_bad}. From the above observations, we see that
	\begin{equation*}
	L_{ijk}(P-p_j)  = f_{ijk} \leq \sum_{i',k'} f_{i'jk'} \leq P-p_j,
	\end{equation*}
	a contradiction, since $L_{ijk} >1$, and $p_j < P$.
	
	Now suppose $j\in M^+$ and $p_j > 0$. For contradiction assume $s_{ijk} >0$. Again (\ref{eq:segment_4a}') requires that \eqref{eq:segment_4a} holds with equality so that $f_{ijk} = L_{ijk}(P-p_j) > 0$. Since $p_j > 0$, then (\ref{eq:spending_4a_good}') requires that 
	\begin{equation*}
	P-p_j = \sum_{i',k'} f_{i'jk'} + \delta_jz \geq f_{ijk} + \delta_j z =  L_{ijk}(P-p_j) + \delta_jz,
	\end{equation*}
	a contradiction since $L_{ijk} > 1$, $\delta_j,z>0$.
\end{proof}

\begin{lemma}\label{lem:bound_on_p_and_r}
	At any solution to LCP~\eqref{lcp:5}, $p_j \leq P, \ \forall j\in M$ and $r_i\le R,\ \forall i\in N$. 
	Further, if $r_i = R$ for some $i \in N$, then $p_j = P, \ \forall j\in M^-$. 
\end{lemma}

\begin{proof}
	First suppose $p_j > 0$, for some $j\in M^-$. Complementarity condition (\ref{eq:spending_4a_bad}') requires that \eqref{eq:spending_4a_bad} holds with equality. Thus, $\sum_{i,k} f_{ijk} + p_j = P$. Then $p_j \leq P$, since $P$, and $f_{ijk}$'s are non-negative. The case of $j\in M^+$ follows similarly, using complementarity condition (\ref{eq:spending_4a_good}'). 

	Next if $r_i>R$ for some $i\in N$, then all her $f_{ijk}$'s have to be zero since both \eqref{eq:mbb_4a_bad} and \eqref{eq:mbb_4a_good} are strict. Then, using the fact that $p_j\le P$ for all $j\in M$, \eqref{eq:budget_4a} is also strict. This violates the corresponding complementarity condition (\ref{eq:budget_4a}') since $r_i>R>0$, a contradiction. 
	
	For the second claim, we show contrapositive. Suppose that $p_j < P$ for any $j\in M^-$, and pick any agent $i\in N$. Recall that $0 < D_{ij1} < \dots < D_{ijk}$, where $k$ is the final segment of $i$'s utility function for $j$. By Claim \ref{claim:no_s_finaL_segment}, $s_{ijk} = 0$, so that constraint \eqref{eq:mbb_4a_bad} for the segment $(i,j,k)$ becomes: $D_{ijk}r_i -p_j \leq D_{ijk}R - P$. Since $p_j < P$, it follows that $r_i < R$. 
\end{proof}
 
\begin{lemma} \label{lem:no_p_0_goods}
	Starting from the primary ray, if Algorithm \ref{algo:mixed} reaches a vertex where $p_j = P$ for some good $j\in M^+$, then $p_{j'} < P$ for all other items $j'\in M$, $r_i = R$, for all $i\in N$, and $z=0$.
\end{lemma} 

\begin{proof}
	For contradiction, let $T$ be the vertex solution to LCP~\eqref{lcp:5} where $p_j = P$ for some good $j$ for the first time, and assume that $p_{j'} <  P$ for some item $j' \in M$. Let $S$ be the vertex the that precedes $T$ starting from the primary ray, and $E$ be the edge between $S$ and $T$. Note that such a $S$ exists since we start from the primary ray where $\p = \r = \0$. Let $M_1$ be the set of goods for which $p_j \rightarrow P$ on $E$, and $N_1 = \{i\in N: \ \exists j\in M_1 \text{ s.t. } \ U_{ijk} > 0, \ k = |u_{ij}| \}$ be the set of agents that are non-satiated for some good in $M_1$. 
	
	\begin{claim}\label{claim:r_to_R}
		At $T$, $r_i = R,\ \forall i\in N_1$.
	\end{claim}
	
	\begin{proof}
		Let $j\in M_1$, and let $i\in N_1$ be an agent that is not satiated for good $j$. Let $(i,j,k)$ be the $i$'s final segment for good $j$. Note that $p_j > 0$ on $E$ so that $p_j$ can increase to $P$. By Claim 1, $s_{ijk} = 0$. Consider the constraint \eqref{eq:mbb_4a_good} for this segment the edge $E$
		\begin{equation} \label{eq:lambda_to_zero}
		U_{ijk} (R - r_i) - (P-p_j) - z \leq 0.
		\end{equation}
		Along $E$, both $z\rightarrow 0$, and $p_j \rightarrow P$.
		Therefore, \eqref{eq:lambda_to_zero} implies that $r_i \rightarrow R$, since $U_{ijk} > 0$.
	\end{proof}
	
	\begin{claim}
		If $p_j \rightarrow P$ for some good $j \in M^+$, then $p_{j'} \rightarrow P, \ \forall j' \in M^-$.
	\end{claim}
	
	\begin{proof}
		Since item $j$ is a good, at least one agent, say $i$, is non-satiated for $j$. Therefore, by Claim \ref{claim:r_to_R}, $r_i \rightarrow R$, on $E$. Consider any $j' \in M^-$. Let $k'=|u_{ij'}|$ be $i's$ final segment of $j'$. By Claim \ref{claim:no_s_finaL_segment}, $s_{ij'k'} =0$, on $E$. Then, constraint \eqref{eq:mbb_4a_bad} requires that 
		\begin{equation*}
		 (P-p_{j'}) \leq D_{ij'k'} (R-r_i)\enspace ,
		\end{equation*}
		which implies that $p_{j'} \rightarrow P$, since $r_i \rightarrow R$.
	\end{proof}
	
	\begin{claim}
		The agents of $N_1$ purchase no items at $T$, i.e., $f_{ijk} = 0, \ \forall j,k$, $\forall i \in N_1$.
	\end{claim}
	
	\begin{proof}
		At $T$, $ p_j = P, \ \forall j \in M^-$, by Claim 3. Therefore, (\ref{eq:spending_4a_bad}') requires that $\sum_{i,k}f_{ijk} +p_j = P,\ \forall j\in M^-$ at $T$. It follows that no agents purchase any bad at $T$, i.e., $f_{ijk}=0 \ \forall i,k, \ \forall j \in M^-$. A similar argument shows that no agents purchase any goods $j\in M_1$ at $T$.
		
		Let $i \in N_1$, and $j$ be any good such that $p_j < P$ at $T$. For contradiction, suppose $i$ purchases $j$ at $T$, i.e., at least $f_{ij1} > 0$. Then,  (\ref{eq:mbb_4a_good}') requires that $(P-p_j) + s_{ij1} = 0$, since $r_i = R$, and $ z=0$ at $T$. Thus, we obtain a contradiction since $s_{ij1}\geq 0$, and $p_j < P$. Therefore, the agents of $N_1$ purchase no items (bads or goods) at $T$.
	\end{proof}
	
	\begin{claim}
		Agents of $N_1$ are not endowed with any fraction of any good with a positive price, i.e., $\forall i\in N_1$, $W_{ij} = 0$ for all $j\in M_0= M^+ \setminus M_1$. Therefore, the budget of each agent $i\in N_1$ is equal to 0 at $T$.
	\end{claim}
	
	\begin{proof}
		At $T$ the following conditions hold for all $i \in N_1$. First, $r_i = R$, by Claim 2. Then, (\ref{eq:budget_4a}') requires that \eqref{eq:budget_4a} holds with equality. Next, Claim 3 shows that $p_j = P, \ \forall j\in M^-$, and Claim 4 states that $f_{ijk} = 0, \forall j,k$. Recalling, that $z=0$ at $T$, then \eqref{eq:budget_4a} simplifies to
		\begin{equation*}
		\sum_{j\in M_0} W_{ij} (P-p_j) = 0.
		\end{equation*}
		Clearly, $W_{ij} = 0, \ \forall j \in M_0$, for any $i\in N_1$ since $p_j < P, \ \forall j \in M_0$. It follows that agents of $N_1$ are only endowed with items in $M^- \cup M_1$. All of these items have price $|P-p_j| = 0$ at $T$. Thus, the budget of agents in $N_1$ equals 0 at $T$.
	\end{proof}
	
	We now prove the lemma. Suppose $p_j < P$ for some item at $T$. Claim 3 shows that $p_j = P, \ \forall j \in M^-$. Therefore, $j\in M_0$. Define $N_0 = N \setminus N_1$. Observe that $|N_0| > 0$, otherwise $|M_0| =0$, by Claim 5. It follows from Claim 2 that any agent $i \in N_0$ is satiated for all $j\in M_1$, i.e., the final segment $(i,j,k)$ has $U_{ijk} = 0$. Further, the agents of $N_1$ start with only goods of $M_1$, by Claim 5. Therefore, in the economic graph described in Section \ref{sec:Assumptions}, there are no edges from the agents of $N_1$ to any agents of $N_0$. That is, the economic graph is not strongly connected, a contradiction. Therefore, $|M_0| = 0$, and $p_j = P$ at $T$. $M_1 = M$, and so $N_1 = N$. By Claim \ref{claim:r_to_R}, $r_i = R,\ \forall i \in N$, at $T$.
\end{proof}

Next, we show that, starting from the primary ray, Algorithm \ref{algo:mixed} never reaches secondary rays where $p_j = P, \ \forall j\in S\subset M^-$, while $p_j < P, \ \forall j\in M^- \setminus S$. For this, we first prove the following lemma.

\begin{lemma} \label{lem:no_sec_ray_first_step}
	Starting from the primary ray, if Algorithm \ref{algo:mixed} reaches a vertex where $p_j =P$ for some bad $j\in M^-$, then $p_{j'} = P$, for all $j'\in M^-$.
\end{lemma} 

\begin{proof}
	For the sake of contradiction, suppose $T$ is the solution to LCP~\eqref{lcp:5} where $p_j = P$ for some bad $j\in M^-$ for the first time. Now consider the vertex $S = (\p,\f,\r,\s,z)$ which precedes $T$. That is, Algorithm \ref{algo:mixed} pivots at the vertex $S$ and travels along the edge $E$ to $T$.
	
	At $S$, $0\leq p_j < P, \forall j \in M^-$, since $T$ is the first time $p_j = P$ for some $j\in M^-$. In addition, complementarity condition (\ref{eq:spending_4a_bad}') requires that constraint \eqref{eq:spending_4a_bad} holds with equality for bad $j$ along the entire edge $E$ so that $p_j$ may increase to $P$. Then, the conditions $\sum_{i,k} f_{ijk} + p_j = P$, and $p_j < P$, imply that at least one agent, say $i$, spends on some segment $(i,j,k)$ along $E$. Recall that we select $R$ large enough that the right hand side of \eqref{eq:mbb_4a_bad} is positive for all segments $(i,j,k)$. Observe that this implies $r_i >0$, otherwise \eqref{eq:mbb_4a_bad} holds with strict inequality which forces $f_{ijk}=0, \ \forall j,k$ by complementarity condition (\ref{eq:mbb_4a_bad}'). Thus, the segment $(i,j,k)$ is either forced or flexible for $i$, $r_i > 0$, and \eqref{eq:mbb_4a_bad} holds with equality for segment $(i,j,k)$ along $E$.
	
	Let $j'$ be a bad such that $p_{j'} < P$ at $T$. By Claim 1, $i$'s final segment $k' = |u_{ij'}|$ has $s_{ij'k'} = 0$. Since \eqref{eq:mbb_4a_bad} holds along $E$ for the segment $(i,j,k)$, then $D_{ijk}(R-r_i) = P-p_j -s_{ijk}<  P-p_j$. Also, since $s_{ij'k'} = 0$, then $d_{ij'k'} (R-r_i) \geq P-p_{j'}$ holds along $E$. Equivalently, we have
	\begin{equation*}
	0< \frac{D_{ijk}}{(P-p_j)} \leq \frac{D_{ijk}}{(P-p_j)-s_{ijk}} = \frac{1}{R-r_{i}} \leq \frac{D_{ij'k'}}{(P-p_{j'})}.
	\end{equation*}
	Thus, we obtain a contradiction since $p_j \rightarrow P$, but $p_{j'} < P$.
\end{proof}

Lemma~\ref{lem:no_sec_ray_first_step} implies that we can not have $p_j = P$ for some subset of bads, while $p_j < P$ for all remaining bads. We still need to rule out the case where $p_j = P, \ \forall j\in M^-$. The argument follows similar reasoning to that of Lemma \ref{lem:no_sec_ray_first_step}. Setting $p_j = P$, i.e., the price of all bads equals 0, requires that at least one agent purchases some bad as $p_j\rightarrow P$, sending $ppb \uparrow \infty$. The more complicated portion of the proof lies in showing that this agent also must purchase some goods. However, if $p_j < P, \ \forall j\in M^+$, then $bpb$ remains bounded. This gives a contradiction since $bpb \geq ppb$ whenever an agent purchases both bads and goods.

\begin{lemma} \label{lem:no_sec_ray_first_step_2}
	Starting from the primary ray, if Algorithm \ref{algo:mixed} reaches a vertex where $p_j =P, \ \forall j \in M^-$, then $p_{j'} = P, \ \forall j'\in M^+$ and $z=0$. 
\end{lemma}

\begin{proof}
	For contradiction, let $T$ be a solution to LCP~\eqref{lcp:5} where $p_j = P, \ \forall j\in M^-$ for the first time, but $p_j < P, \ \forall j\in M^+$. Note that Lemma \ref{lem:no_p_0_goods} shows that $p_j < P, \ \forall j\in M^+$, otherwise all prices are set to zero, i.e., $p_j = P, \ \forall j\in M$. Let $S$ be the vertex which precedes $T$.
	
	At $S$, $p_j > 0, \forall j\in M^-$ so that $p_j$ may increase to $P$. Then, the conditions $\sum_{i,k} f_{ijk} + p_j = P$, and $p_j < P, \ \forall j \in M^-$, imply that at least one agent, say $i$, spends in her first segment $(i,j,1)$ for some bad $j$. Note that $r_i>0$, otherwise \eqref{eq:mbb_4a_bad} holds with strict inequality, and so (\ref{eq:mbb_4a_bad}') requires $f_{ijk} =0$ for all bads. Thus, the segment $(i,j,1)$ is either forced or flexible for $i$, $r_i > 0$, and \eqref{eq:mbb_4a_bad} holds with equality for segment $(i,j,1)$ along edge. We want to show that these conditions imply that the agent also purchases some good.
	
	Observe that on the edge $E$ from $S$ to $T$, every agent's budget eventually becomes strictly positive, since $p_j \rightarrow P, \ \forall j\in M^-$. Fix $\epsilon > 0$, and pick a point $T'$ on $E$ so that $2|M^-| \max_{j\in M^-} (P-p_j) <  \epsilon$. At $T'$, it follows that 
	\begin{equation*}
	\sum_{k,j\in M^-} |f_{ajk} - W_{aj}(P-p_j)| \leq 2|M^-| \max_{j\in M^-} (P-p_j) \leq \epsilon, \ \forall a\in N,
	\end{equation*}
	since $W_{aj}\leq 1$, and $f_{ajk} \leq P-p_j$, by \eqref{eq:spending_4a_bad}. Recall that $r_i > 0$, so that (\ref{eq:budget_4a}') requires that $\sum_{j' \in M^+}W_{ij'}(P-p_{j'}) +z- \sum_{j\in M^-} W_{ij}(P-p_j) +  \sum_{k,j\in M^-} f_{ijk} = \sum_{k,j' \in M^+} f_{ij'k},$ or	
	\begin{equation*}
	\sum_{j'\in M^+}W_{ij'}p_{j'} + z -\epsilon \leq \sum_{k,j'\in M^+} f_{ij'k} \leq \sum_{j'\in M^+}W_{ij'}p_{j'} + z + \epsilon,
	\end{equation*}
	at $T'$. Therefore, we must have $f_{ij'k'} >0 $ at least for some segment $(i,j',k')$ of some good $j'$, since $ \sum_{j\in M^+}W_{ij}(P-p_{j}) >0$, $z\geq 0$, and $\epsilon >0$ was arbitrary. For this segment, complementarity condition (\ref{eq:mbb_4a_good}') requires that $U_{ij'k}(R- r_i) = (P-p_{j'}) + z + s_{ij'k}$. Note that $U_{ij'k} > 0$ since $(P-p_j) > 0$ on $E$, and $z,s_{ijk} \geq 0$. For the bad $j$, (\ref{eq:mbb_4a_bad}') requires $D_{ijk} (R-r_i) = (P-p_j)-s_{ijk}$, since $f_{ijk} > 0$. Further, these conditions hold along the edge from $T'$ to $T$ where $p_j \rightarrow P, \ \forall j\in M^-$. But then
	\begin{equation*}
	\frac{U_{ij'k'}}{P-p_{j'}+z+s_{ij'k'}} = \frac{1}{R-r_i} = \frac{D_{ijk}}{(P-p_j)-s_{ijk}} \geq \frac{D_{ijk}}{P-p_j}, 
	\end{equation*}
	so that $U_{ij'k'} /(P-p_{j'} + z+ s_{ij'k'}) \rightarrow \infty$, since $p_j \rightarrow P$. Then, me must have $p_{j'} \rightarrow P$, $s_{ij'k'}\rightarrow 0$, and $z \rightarrow 0$, along the edge from $T'$ to $T$, since $s_{ij'k'},z\geq 0$ and $p_{j'} \leq P, \ \forall j\in M$ by Lemma \ref{lem:bound_on_p_and_r}. A contradiction since $p_{j'} < P, \ \forall j\in M^+$.
\end{proof}

Lemma \ref{lem:no_sec_ray_first_step_2} rules out the possibly of secondary rays where $p_j = P$, $\forall j\in S \subseteq M^-$. We still need to show that Algorithm \ref{algo:mixed} never reaches the degenerate solution.

\begin{lemma} \label{lem:no_z_0_solution_2}
	Starting from the primary ray, Algorithm \ref{algo:mixed} never reaches the solution $p_j = P, \ \forall j\in M$, and $r_i = R, \ \forall i \in N$, with all other variables, including $z$, equal to zero.
\end{lemma} 

\begin{proof}
	Let $T$ be the degenerate solution. By Lemma \ref{lem:no_sec_ray_first_step_2}, Algorithm \ref{algo:mixed} never reaches a vertex with $p_j = P, \ \forall j \in M^-$, while $p_j < P, \ \forall j\in M^+$. Therefore, the only possibility is that \emph{all} $p_j$ are set to $P$ simultaneously.
	
	Consider the vertex $S$ that precedes $T$. At $S$, $0 < p_j < P, \ \forall j\in M$, so that $p_j$'s can increase to $P$. Thus, (\ref{eq:spending_4a_bad}') and (\ref{eq:spending_4a_good}') require that \eqref{eq:spending_4a_bad} and \eqref{eq:spending_4a_good} hold with equality at $S$. Summing these equalities over all $j\in M$ shows that the total spending is
	\begin{equation*}
	\sum_{i,k,j\in M^+} f_{ijk} - \sum_{i,k,j\in M^-} f_{ijk} = \sum_{i,j\in M^+} (P-p_j) -\sum_{i,j\in M^-}(P-p_j) - z \sum_{j\in M^+} \delta_j.
	\end{equation*}
	
	Note that $z=0$ at $T$ so that Algorithm \ref{algo:mixed} stops there. This implies, $0 < r_i < R, \ \forall i \in N$, at $S$ so that the $r_i$ can increase to $R$, as required by Lemma \ref{lem:no_p_0_goods}. Therefore, (\ref{eq:budget_4a}') requires that \eqref{eq:budget_4a} holds with equality for all $i\in N$. Summing over all $i$ yields
	\begin{equation*}
	 \sum_{i,k, j' \in M^+} f_{ij'k} - \sum_{i,k,j\in M^-} f_{ijk} =  \sum_{i,j' \in M^+}W_{ij'}(P-p_{j'}) -\sum_{i,j\in M^-} W_{ij}(P-p_j)+z n.
	\end{equation*}
	Or, since there is a unit amount of each item, i.e., $\sum_i W_{ij} = 1$, we see that
	\begin{equation*}
	\sum_{i,k, j' \in M^+} f_{ij'k} - \sum_{i,k,j\in M^-} f_{ijk} =  \sum_{j' \in M^+}(P-p_{j'}) -\sum_{j\in M^-} (P-p_j)+z n.
	\end{equation*}
	Which implies that $z(n+\sum_{j\in M^+} \delta_j )=0$, at $S$. Thus, $z = 0$, since the $\delta_j >0, \ \forall j\in M^+$. This means that the algorithm stops at $S$, which is a competitive equilibrium by Theorem \ref{thm:one_to_one_2}.
\end{proof}

\begin{theorem} \label{thm:no_secondary_rays}
	Starting from the primary ray, the Algorithm \ref{algo:mixed} never reaches a secondary ray.
\end{theorem}

\begin{proof}
	Here, we need to impose conditions on the choices of $P$ and $R$. After fixing any $P \in \bR_+$, select $R$ large enough to ensure that: the right hand side of \eqref{eq:mbb_4a_bad} is positive, i.e., $D_{ijk}R - P >0$, for all segments $(i,j,k)$, $\forall i,k$, $\forall j\in M^-$, and that the right hand side of \eqref{eq:mbb_4a_good} is negative, i.e., $P-U_{ijk}R < 0$, for all segments $(i,j,k)$, $\forall i,k$, $\forall j\in M^+$. Recall that the ray $\mathcal{R} = \{ [\y^*,z^*] + \alpha [\y',z'] \ | \ \forall \alpha \geq 0 \}$ begins at the vertex $(\y^*,z^*)$ and travels in the direction $(\y',z')$, where $\y = (\p,\f,\r,\s)$.
	
	First, we show that $\y' = \0$, starting with $\p' = \0$. Consider constraints \eqref{eq:spending_4a_bad}, \eqref{eq:spending_4a_good}, and complementarity conditions (\ref{eq:spending_4a_bad}') and (\ref{eq:spending_4a_good}'). For contradiction, suppose $p_j' >0$ for some $j \in M$. Then, $p_j' >0, \ \forall \alpha >0$, so \eqref{eq:spending_4a_bad} or \eqref{eq:spending_4a_good} must hold with equality. Since $P$ is fixed, $f_{ijk}'\geq 0,\ \forall i,j,k$, and $z>0$, then eventually \eqref{eq:spending_4a_bad} or \eqref{eq:spending_4a_good} is violated. Therefore, $\p' = \0$. Similarly, by \eqref{eq:spending_4a_bad}, $\f' = \0, \ \forall j \in M^-$. Note that if $\p' = \0$, then the price of each item is constant along $E$. Recall from Claim 1, that $s_{ijk} = 0$ for the final segment $k = |u_{ij}|$ of any bad. Therefore, $\r' = \0$, otherwise \eqref{eq:mbb_4a_bad} is eventually violated for the final segment $(i,j,k)$ of any bad $j\in M^-$ for any agent $i\in N$. Also, since $\f'=\0$ for all bads $j \in M^-$, the spending on bads is constant. If $f_{ijk}'>0$ for some good $j\in M^+$, then $z$ must increase to ensure inequality \eqref{eq:budget_4a} holds. Further, \eqref{eq:mbb_4a_good} must hold for this segment $(i,j,k)$ by complementarity condition (\ref{eq:mbb_4a_good}') since $f_{ijk}' >0$. However, since $r_i'$ and $p_j'$ are constant and $s_{ijk} \geq 0$, \eqref{eq:mbb_4a_good} can not hold with equality as $z$ increases. This shows that $\p'$, $\f'$, and $\r'$ are constant. Observe that these variables determine $\s$ by \eqref{eq:mbb_4a_bad} and \eqref{eq:mbb_4a_good}. Therefore, $\s'= \0$. It follows that $z' >0$, otherwise no variables change.
	
	Finally, we show $\y^* = \0$. Notice that along the ray $E$, the money earned and spent by each agent remains constant. However, $z$ increases. Thus, complementarity condition (\ref{eq:budget_4a}') implies that $\r^* =\0$. It follows that \eqref{eq:mbb_4a_bad} holds with strict inequality for all bads $j\in M^-$, forcing $\f^*=\0$ for bads $j\in M^-$, by (\ref{eq:mbb_4a_bad}'). Now, (\ref{eq:spending_4a_bad}') requires that $\p^* = \0$ for all bads $j\in M^-$, since $p_j <P, \ \forall j\in M$ and $\f^* =\0$ for all bads $j\in M^-$. Since $z$ increases while $r_i$ and $p_j$ remain fixed for all goods $j \in M^+$, complementarity conditions (\ref{eq:spending_4a_good}') and (\ref{eq:mbb_4a_good}') require that both $\p^*$ and $\f^*$ are equal to $\0$ for all goods $j\in M^+$. As a result, $\s^* =\0$, by (\ref{eq:segment_4a}') as \eqref{eq:segment_4a} holds with strict inequality $\forall j \in M$, since $p_j < P, \ \forall j \in M$. Therefore, $\y^* = \0$, and the ray is $\mathcal{R} = [\0,z^*] + \alpha [\0,1]$, i.e., the primary ray. 
\end{proof}

\noindent{ \textit{Proof} (of Theorem \ref{thm:mixed_converges}).} Theorem \ref{thm:non_degenerate_2} shows that every vertex solution to LCP~\eqref{lcp:5} with $\p < P$, $\r < R$, and $z>0$, is nondegenerate as long as there is no polynomial relation between $\mathbf{U}$, $\mathbf{W}$, and $\mathbf{L}$. Lemmas \ref{lem:no_p_0_goods}, \ref{lem:no_sec_ray_first_step}, and \ref{lem:no_sec_ray_first_step_2} shows that we never reach a vertex where $p_j = P$ for any $j \in M$, or $r_i = R$ for any $i\in N$. Therefore, there is always a unique double label for Algorithm \ref{algo:mixed} to pivot at. Theorem \ref{thm:no_secondary_rays} establishes that Algorithm \ref{algo:mixed} never reaches a secondary ray, so that eventually it reaches a solution with $z=0$, $\p < P$, and $\r < R$, which is an equilibrium by Theorem \ref{thm:one_to_one_2}.

\subsection{Results}\label{sec:results}
Theorem~\ref{thm:mixed_converges} directly yields the following results on existence, membership in PPAD, and rational-valued property. 

\begin{theorem}
If the fair division instance of a mixed manna under SPLC utilities satisfies strong connectivity, as defined in Section~\ref{sec:Assumptions}, then there exists a competitive allocation, and the Algorithm~\ref{algo:mixed} terminates with one. Furthermore, Algorithm~\ref{algo:mixed} finds a rational-valued solution if all input parameters are rational numbers. 
\end{theorem}

We note that for the bads only case the above theorem does not apply because the {\em strong connectivity} assumption defined using goods is inapplicable. Therefore, for this case, we separately show the convergence of our algorithm in Appendix \ref{sec:convergence_all_bads}. This together with Theorem \ref{thm:equiL_soL_correspondence} shows that Algorithm \ref{algo:mixed} finds an equilibrium for bads with SPLC utility functions, and therefore the remaining theorems hold for this case as well. 

\begin{theorem}
If the fair division instance of a mixed manna under SPLC utilities satisfies strong connectivity, as defined in Section~\ref{sec:Assumptions}, then the problem of computing a competitive allocation is in $\classPPAD$.
\end{theorem}

\begin{proof}
The proof of this theorem follows from the Todd's result~\cite{Todd76} on orientability of the path followed by a complementary pivot algorithm, and is exactly same as the proof of Theorem 6.2 in~\cite{GargMSV15}.
\end{proof}

\begin{theorem} \label{thm:odd_number}
	If the fair division instance of a mixed manna under SPLC utilities satisfies strong connectivity, as defined in Section~\ref{sec:Assumptions}, and the input parameters $\mathbf{U}$, $\mathbf{W}$, and $\mathbf{L}$ have no polynomial relationship between them, then there are an odd number of competitive equilibria.
\end{theorem}

\begin{proof}
	Since the parameters $\mathbf{U}$, $\mathbf{W}$, and $\mathbf{L}$ have no polynomial relationship between them, Theorem \ref{thm:one_to_one_2} shows that all solutions to LCP~\eqref{lcp:5} with $p_j < P$, $\forall j \in M$, $r_i < R$, $\forall i\in N$, and $z =0$ are competitive equilibria. Theorem \ref{thm:mixed_converges} establishes that Algorithm 1 always terminates at one of these solutions. We now argue that all other equilibria are paired up on paths of the polyhedron corresponding to LCP~\eqref{lcp:5}.
	
	 Theorem \ref{thm:non_degenerate_2} shows that every vertex solution of LCP~\eqref{lcp:5} with $p_j < P$, $\forall j\in M$, and $r_i < R$, $\forall i \in N$ is nondegenerate. Therefore, a unique double label exists. Lemmas \ref{lem:no_p_0_goods}, \ref{lem:no_sec_ray_first_step}, \ref{lem:no_sec_ray_first_step_2}, and \ref{lem:no_z_0_solution_2} show that starting from a solution with $p_j < P$, $\forall j\in M$, and $r_i < R$, $\forall i \in N$ and traveling along the edge incident to the double label, we always reach another solution where $p_j < P$, $\forall j\in M$, and $r_i < R$, $\forall i \in N$. Thus, a set of paths connect these solutions. Moreover, Theorem \ref{thm:no_secondary_rays} shows that these paths never reach a secondary ray. Therefore, if one starts from an equilibrium, then the subsequent path of solutions with $p_j < P$, $\forall j\in M$, and $r_i < R$, $\forall i \in N$ must eventually end at a vertex where $z=0$, i.e., another equilibrium. Then, all other equilibria, besides the one found starting from the primary ray, must be paired. Thus, there are an odd number of equilibria.
\end{proof}

\section{Strongly Polynomial Bound}\label{sec:strongly}
Devanur and Kannan~\cite{DevanurK08} offered a strongly polynomial time algorithm for exchange model for goods with SPLC utilities when either the number of goods or the number of agents is constant, which~\cite{GargK15} extended to more general Arrow-Debreu model with production. The approach uses a \emph{cell decomposition} technique and the fact that $n$ hyperplanes in $\bR^d$ form at most $O(n^d)$ nonempty regions, or cells. Garg et al.~\cite{GargMSV15} adapted this argument to bound the number of fully label vertices in their LCP formulation for exchange model for goods under SPLC utilities. We follow their analysis and obtain a strongly polynomial bound on runtime for the case of all bads as well. 

The idea is as follows. Suppose the number of bads, i.e., $m$, is a constant. We decompose $(\p,z)$ space, i.e., $\bR^{m+1}_+$, into cells by a set of polynomially many hyperplanes such that each cell corresponds to unique setting of forced, flexible, and undesirable partitions. Then, we show that each fully labeled vertex maps into a cell by projection. Further, at most two vertices map to any given cell. 
Consider the LCP~\eqref{lcp:5} from Section~\ref{sec:splc} with $M^- = M$ (i.e., $M^+=\emptyset$). That is,
\begin{subequations}\label{lcp:6}
\begin{eqnarray}
\forall i \in N: & -\sum_{j\in M} W_{ij} p_j - \sum_{k,j\in M} f_{ijk} - z \leq -P \sum_{j\in M} W_{ij} & \perp\ \ r_i\label{6eq:budget_4a}\\
\forall j\in M: & \displaystyle\sum_{i,k} f_{ijk} + p_j  \leq P & \perp \ \ p_j\label{6eq:spending_4a_bad}\\
\forall j \in M, \ \forall i,k: & D_{ijk} r_i - p_j -s_{ijk}\leq D_{ijk}R - P & \perp \ \  f_{ijk}\label{6eq:mbb_4a_bad} \\
\forall (i,j,k): & f_{ijk} + L_{ijk}p_j \leq L_{ijk} P & \perp \ \  s_{ijk}\enspace . \label{6eq:segment_4a}
\end{eqnarray} 
\end{subequations}

The main result of this section is the following theorem.

\begin{theorem} \label{thm:run_time_bound}
	If the fair division instance of a mixed manna under SPLC utilities that contains only bads has either constantly many agents or constantly many bads, then Algorithm~\ref{algo:mixed} runs in strongly polynomial time.
\end{theorem}

\paragraph{Constantly Many Bads} We consider $\bR_+^{m+1}$ with coordinates $p_1,\dots,p_m,z$. For each tuple $(i,j,j',k,k')$ where $i\in N$, $j\neq j' \in M$, $k \leq |u_{ij}|$, and $k' \leq |u_{ij'}|$, create a hyperplane $D_{ijk}(P-p_{j'}) - D_{ij'k'}(P-p_{j}) = 0$. This divides $\bR_+^{m+1}$ in cells where each region has one of the signs $\leq$, $=$, or $\geq$. For any agent $i\in N$, the sign of each cell gives a partial order on the pain per buck of her segments. Thus, in any cell, we can sort the segments $(j,k)$ of agent $i$ in increasing order of pain per buck, and create equivalence classes $B_1^i,\dots,B_l^i$ with same pain per buck. Let $B_{<l}^i = B_1^i \cup \dots \cup B_{l-1}^i$, and define $B_{\leq l}^i$ and $B_{\geq l}^i$ similarly. 

Next, we show how to represent the flexible partition. We further subdivide each cell by adding the hyperplanes $\sum_{(j,k)\in B_{<l}^i} L_{ijk} (P-p_{j}) = \sum_j W_{ij} (P-p_{j}) - z$, for each agent $i\in N$, and each of her partitions $B_l^i$. For each subcell, let $B_{l_i}^i$ be the rightmost partition such that $\sum_{(j,k)\in B_{<l}^i} L_{ijk} (P-p_{j}) < \sum_j W_{ij} (P-p_{j}) - z$ for agent $i$. Then, $B_{l_i}^i$ is her flexible partition. Finally, we add the hyperplanes $p_{j} = 0, \ \forall j\in M$ and $z = 0$, so that we only consider the cells where $p_{j} \geq 0$, and $z\geq 0$. Since every vertex on the path followed by Algorithm~\ref{algo:mixed} satisfies $p_j < P$, we only consider the cells where $p_j < P$. Observe that any vertex $(\y,z)$ traced by our algorithm maps to a cell by projecting it onto $(\p,z)$ space. 

\begin{lemma} \label{lem:runtime_bads_1}
	Let $\cP$ be the polyhedron corresponding to LCP~\eqref{lcp:6}. Then, at most two fully labeled vertices of $\cP$ map onto any given cell. Further, if two vertex map to the same cell, then they are adjacent.
\end{lemma}

\begin{proof}
	Each fully labeled vertex and each cell correspond to their own settings of forced, flexible, and undesirable partitions for each agent. Therefore, if a vertex maps to a given cell, then these two setting must match. If a vertex $S = (\p,\f,\r,\s,z)$ maps to a certain cell, then the following inequalities are satisfied
	\begin{itemize}
		\item If $p_j > 0$, then $\sum_{i,k} f_{ijk} = P-p_{j}$, else $p_{j} =0$ at $S$.
		\item If $\sum_j W_{ij}(P-p_j) -z \geq 0$, then $\sum_j W_{ij}(P-p_{j}) - \sum_{j,k} f_{ijk} -z = 0$, else $r_i = 0$ at $S$.
		\item If $D_{ij'k'}(P-p_j) - D_{ijk}(P-p_{j'})\ge 0$ for $(j',k')\in B_{l_i}^i$, then $-D_{ijk}(R-r_i) + (P-p_{j})-s_{ijk} = 0$, else $f_{ijk} = 0$ at $S$.
		\item If $D_{ij'k'}(P-p_j) - D_{ijk}(P-p_{j'})> 0$ for $(j',k')\in B_{l_i}^i$, then $f_{ijk} = L_{ijk} (P-p_{j})$, else $s_{ijk} = 0$.
	\end{itemize}
	
	In each of the complementarity conditions above, one inequality is enforced. Therefore, their intersection forms a line. If this line does not intersect $\cP$, then no vertex maps to this cell. If it does, then intersection is either a fully labeled vertex, or a fully labeled edge on which the solution is fully labeled along the entire edge. In the former case, only the vertex $S$ maps to the cell. In the latter, only the endpoints of the fully labeled edge map to the cell. Clearly, these vertices are adjacent. 
\end{proof}

Notice that the total number of hyperplanes we created is strongly polynomial. Therefore, this creates a strongly polynomial number of cells as well. 

\paragraph{Constantly Many Agents} 
In this case, we consider the space $\bR_+^{n}$ using to the coordinates $\r$. Then, we create a partitioning of the segments corresponding to the bads. Besides this change, the remaining analysis is similar. 

Every fully labeled vertex $S = (\p,\f,\r, \boldsymbol{s},z)$ maps to $\bR_+^n$ by taking the projection on $\r$. Given a fully labeled vertex, for each bad $j$ sort all of its segments $(i,j,k)$ by increasing order of $D_{ijk} (R-r_i)$ and partition them into equivalence classes $B^j_1, \dots, B^j_l$. Observe that, at this vertex, bad $j$ gets allocated in order of these partitions. If segment $(i,j,k) \in B^j_l$ is allocated, i.e., $f_{ijk} >0$, then all the segments in partitions before $B^j_l$ must also be allocated. We call the last allocated partition the flexible segment, all partitions before it forced partitions, and all partitions after it the undesirable partitions of bad $j$. Suppose that segment $(i,j,k)$ is in the flexible partition of bad $j$. Then, $D_{ijk}(R - r_i) = P-p_{j}$, otherwise all segments in this partition are either undesirable or all of them are forced for the corresponding agents. Therefore, the flexible partition defines the price of each bad.

Now we decompose the space $\bR_+^n$ into cells in a way that captures the segment configuration of each bad. For each tuple $(i,i',j,k,k')$ where $i\neq i' \in N$, $j\in M$, $k\leq |u_{ij}|$, and $k' \leq |u_{i'j}|$, we introduce the hyperplane $D_{ijk}(R-r_i) - D_{i'jk'} (R-r_{i'}) = 0$. In any cell, the signs of these hyperplanes gives a partial order of segments $(i,k)$ for each agent $i$ based on $D_{ijk} (R-r_i)$. Sort the segments of each bad $j$ in increasing order of $D_{ijk}(R - r_i)$, and partition them into equality classes $B^j_1, \dots, B^j_l$.

Next, we capture the flexible partition of each bad. If the bad is fully sold, then simply sum the lengths of the segments starting from the first until it becomes 1. An undersold bad requires more work. If a bad is undersold, then $p_{j} =0$. Thus, segments of its flexible partition satisfy $D_{ijk} (R-r_i) = P$. To capture this we add the hyperplanes $D_{ijk} (R-r_i) -P = 0$, for all $(i,j,k)$. Observe that flexible partition of a bad is either: the partition when it becomes fully sold, or where $D_{ijk} (R-r_i) = P$. This can easily be deduced from the signs of the hyperplanes. Finally, we add the hyperplanes $r_i = 0, \forall i$ and consider only those cells for which $0 \le r_i < R, \forall i$. 

From the above discussion it is clear that the fully labeled vertices which map to a given cell may be worked out similarly to Lemma \ref{lem:runtime_bads_1}. Further, we obtain one equality for each complementarity condition, since each cell captures complete segment configuration, status of bads, and agents of the instance. 

\begin{lemma} \label{lem:runtime_bads_2}
	Let $\cP$ be the polyhedron corresponding to LCP~\eqref{lcp:6}. Then, at most two fully labeled vertices of $\cP$ map onto any given cell. Further, if two vertex map to the same cell, then they are adjacent.
\end{lemma}

Clearly, our algorithm follows a systematic path rather than a brute force enumeration of every cell configuration like in~\cite{BranzeiS19,GargM20}. Theorem~\ref{thm:run_time_bound} follows from the above discussion since the number of hyperplanes is strongly polynomial in both cases. 

\begin{remark}
	It is not clear how to show a strongly polynomial bound for the case of mixed manna when the number of agents (or items) is a constant. This is due to the additional variable $z$ appearing in~\eqref{eq:mbb_4a_good} (constraint to force an optimal bundle). This makes the bpb condition unusable as a segment configuration at an arbitrary fully-labeled vertex.
\end{remark}

\section{$\classPPAD$-Hardness of all Bads with SPLC utilities} \label{sec:hardness}
In this section, we show that finding $1/poly(n)$-approximate equilibrium for bads under SPLC utility functions is $\classPPAD$-hard. Our proof relies on reducing the problem of finding an approximate Nash equilibrium, which is known to be $\classPPAD$-hard~\cite{DaskalakisGP09,ChenDT09}, to finding a competitive equilibrium. Our reduction is motivated from the construction of~\cite{ChenDDT09}, which shows a similar result for the goods case. The main challenge in extending the reduction to the bads case is that~\cite{ChenDDT09} crucially uses utility values of $0$  to prevent allocating certain items to agents. This translates to the disutility (cost) of $\infty$ in case of bads. However, this breaks our proof of existence of an equilibrium (since we require finite utility functions), and therefore can not be used. In fact, an equilibrium may not even exist if we allow $\infty$ disutility values~\cite{ChaudhuryGMM20}. Thus, in the bads case every agent can possibly be assigned any bad and utility functions must be designed in such a way that only certain desired allocations happen at equilibrium. 

We first show the hardness for the exchange model that we later extend to the Fisher (and CEEI) model. Let us start by defining approximate equilibrium of exchange setting.

\paragraph{Approximate Exchange Equilibrium} Recall that an exchange equilibrium $(\p^*,\x^*)$ satisfies the following two conditions:
\begin{enumerate}
	\item[C1.] Optimal bundle: For each $i \in N$, $\x^*_i \in \argmin\{ f_i(\x)\ \text{ s.t. }\ \x\ge 0;\ \x^*_i\cdot \p^* \ge \w_i \cdot \p^*\}$.
	\item[C2.] Supply-Demand: For each $j \in M$, $\sum_{i\in N} x^*_{ij} = \sum_{i\in N} W_{ij}$. 
\end{enumerate}
The equilibrium is said to be $\epsilon$-approximate, for $\epsilon>0$, if the optimal bundle condition holds as above and the demand meets supply approximately:
\begin{itemize}
	\item[C2'.] $\epsilon$-Supply-Demand: For each $j \in M$, $|\sum_{i\in N} x^*_{ij} - \sum_{i\in N} W_{ij}| \le \epsilon (\sum_{i\in N} W_{ij})$.
\end{itemize}

In our proof of hardness, we reduce finding approximate well-supported Nash equilibrium of a two-player game to finding approximate competitive equilibrium, defined as follows.

\paragraph{2-Nash} A two-player game, where each player has $n$ moves to chose from, can be represented by two $n\times n$ payoff matrices $(R,C)$. This is because in a play one of the player can be thought of choosing a row (row-player) and the other player choosing a column (column-player), and the corresponding entry in $R$ and $C$ are their payoffs respectively. The Nash equilibrium problem can be states as follows: Given payoff matrices $R,C \in [0,1]^{n \times n}$ and $\epsilon>0$, find mixed-strategy (probability distribution) $\balpha$ for the row-player and $\bbeta$ for the column-player such that, $\balpha,\bbeta \in \{\zz\ge 0| \sum_{s \in [n]} z_s=1\}$ and,
\begin{equation}\label{eq:ne}
\forall s,t\in [n],\ \ \ \ (R\bbeta)_s \le  (R\bbeta)_t -\epsilon \Rightarrow \alpha_s=0\ \ \mbox{ and }\ \ (\balpha^TC)_s \le  (\balpha^TC)_t -\epsilon \Rightarrow \beta_s=0.
\end{equation}

A remarkable series of results in 2006 settled the complexity of finding (approximate) Nash equilibrium in finite games \cite{DaskalakisGP09,ChenDT09}. In particular, \cite{ChenDT09} show that for any $c>0$, finding $1/n^c$-approximate well-supported Nash equilibrium in two player games is $\classPPAD$-hard. We provide a reduction from finding a $1/n$-approximate Nash equilibrium to finding a $1/poly(n)$-approximate equilibrium in the exchange setting for bads under SPLC cost functions to show Theorem \ref{thm:hardness}. And then reduce exchange to CEEI to show Theorem \ref{thm:ceei-hardness} below.

\begin{theorem}\label{thm:hardness}
	Finding a $\frac{1}{poly(n)}$-approximate equilibrium in a bads exchange setting with SPLC cost is $\classPPAD$-hard, where $m$ is the number of bads. This holds even if every utility function has at most two segments with $O(1)$ costs. 
	\end{theorem}

\begin{theorem}\label{thm:ceei-hardness}
	Finding a $\frac{1}{poly(n)}$-approximate equilibrium in a bads Fisher setting with equal budgets (CEEI) under SPLC costs is $\classPPAD$-hard, where $m$ is the number of bads. This holds even if every utility function has at most three segments with $O(1)$ costs.
\end{theorem}

\subsection{Exchange Setting Construction}
Given an $n\times n$ game $(R,C)$ we construct an exchange setting $\mathcal M(R,C)$ of bads with SPLC utilities as follows:
\medskip

\noindent{\bf Bads.} The setting has $m=2n+2$ bads, $M={1,\dots,m}$. Intuitively prices of items $\{1,\dots, n\}$ will correspond to vector $\boldsymbol{\alpha}$, and that of items $\{n+1,\dots, 2n\}$ will correspond to vector $\boldsymbol{\beta}$ of the Nash equilibrium. 
\medskip

\noindent{\bf Agents, their endowments, and utility functions.}
Let $H$ be a large constant that we will set later, and let $(c)^+$ denote $\max\{c,0\}$. Next, we will describe four sets of agents, and $N$ is the union of these four. If the cost function is linear we only specify the slope and if it is PLC then we represent it by a list of (slope, length) pairs. Note that we work with disutility values $D_{ijk} = |U_{ijk}|$. Also note that we only specify non-zero endowments.
\begin{itemize}
\item {\em Price-regulating agents $N_{PR}$.} For every pair of items $(j, j') \in M\times M$ where $j\neq j'$, there is an agent $a_{j,j'}\in N_{PR}$. For notational simplicity let $i = a_{j,j'}$, then $i$'s endowment and utility functions are
\begin{itemize}
	\item Endowment: $W_{i,j} = 1/n$.
	\item Utility Functions: $u_{ij} = 1$, $u_{ij'} =2$, and $u_{ik} = H$ otherwise.
\end{itemize}
\item {\em Deficit agents $N_D$.} For every item $j \in [2n]$, there is an agent $a_j \in N_D$. Let $i = a_{j}$, then $i$'s endowment and utility functions are 
\begin{itemize}
	\item Endowment: $W_{i,(m-1)} = W_{i,(2n+1)} = 1/n^8$.
	\item Utility Functions: $u_{ij} = 1$, and $u_{ik} = H$ otherwise.
\end{itemize}
\item {\em Row-player's agents $N_R$.} For every pair $s,s'\in [n]$, there are two agents, $a_{s,s',R} \in N_R$ and $a_{s',s,R}\in N_R$. Let us denote these by $i$ and $i'$ respectively. Our goal is to capture the NE conditions for the row-player using the optimal bundles of these two agents, and therefore their endowments and utility functions are closely related. 

\begin{itemize}
	\item Endowment: For each $k\in [n]$, let $r_k=(R_{s'k} - R_{sk})$, and define $r=\sum_k r_k$.
	\begin{itemize}
		\item $W_{is}=W_{i's'}=1/n^4$.
		\item For each $k\in [n]$, $W_{i(n+k)} = \frac{(r_k)^+}{n^6}$, and $W_{i'(n+k)}= \frac{(-r_k)^+}{n^6}$.
		\item $W_{i(m-1)}=\frac{(-r)^+}{n^6}$ and $W_{i'(m-1)}=\frac{(r)^+}{n^6}$. 
	\end{itemize}
	\item Utility Functions:
	\begin{itemize}
		\item $u_{i,s}=u_{i',s'} = \{(1, 1/n^4), (H,\infty)\}$
		\item For each $k\in [n]$, $u_{i,(n+k)}=\{(1/3,(-r_k)^+/n^6), (H,\infty)\}$ and 
		
		$u_{i',(n+k)}=\{(1/3,(r_k)^+/n^6), (H,\infty)\}$.
		\item $u_{i,(m-1)}=\{(1/3,(r)^+/n^6), (H,\infty)\}$ and $u_{i',(m-1)}=\{(1/3,(-r)^+/n^6), (H,\infty)\}$.
		\item $u_{i,m} = u_{i',m} = 3$.
		\item $u_{i,k} = u_{i',k}= H$, otherwise.
	\end{itemize}
\end{itemize}

\item {\em Column-player's agents $N_C$.} For every pair $s,s'\in [n]$, there are two agents, $a_{s,s',C} \in N_C$ and $a_{s',s,C}\in N_C$. Let us denote these by $i$ and $i'$ respectively. Again, the goal is to capture the NE conditions for the column-player using the optimal bundles of these two agents, and therefore their endowments and utility functions are constructed in similar way as above.
\begin{itemize}
	\item Endowment: For each $k\in [n]$, let $c_k=(C_{ks'} - C_{ks})$, and set $c=\sum_k c_k$.
	\begin{itemize}
		\item $W_{i(n+s)}=W_{i'(n+s')}=1/n^4$.
		\item For each $k\in [n]$, $W_{i,k} = \frac{(c_k)^+}{n^6}$, $W_{i',k}= \frac{(-c_k)^+}{n^6}$.
		\item $W_{i(m-1)}=\frac{(-c)^+}{n^6}$ and $W_{i'(m-1)}=\frac{(c)^+}{n^6}$. 
	\end{itemize}
	\item Utility Functions:
	\begin{itemize}
		\item $u_{i,(n+s)}=u_{i',(n+s')} = \{(1, 1/n^4), (H,\infty)\}$.
		\item For each $k\in [n]$, $u_{i,k}=\{(1/3,(-c_k)^+/n^6), (H,\infty)\}$ and 
		
		$u_{i',k}=\{(1/3,(c_k)^+/n^6), (H,\infty)\}$.
		\item $u_{i,(m-1)}=\{(1/3,(c)^+/n^6), (H,\infty)\}$ and $u_{i',(m-1)}=\{(1/3,(-c)^+/n^6), (H,\infty)\}$.
		\item $u_{i,m} = u_{i',m} =3$.
		\item $u_{i,k} = u_{i',k} = H$, otherwise.
	\end{itemize}
\end{itemize}
\end{itemize}

Note that, each of the $c_k$'s and $r_k$'s above are in $[-1,1]$, and therefore for any bad $j\in M$, total endowment of it in the setting is $\sum_{i\in A} W_{ij} \le \frac{(m-1)}{n} + \frac{n}{n^4} + \frac{2n^2}{n^6} + \frac{2n}{n^8} \le 3$. Thus, the $\epsilon$-approximate supply-demand condition the above setting $\mathcal M(R,C)$. can be re-written as, 

\begin{equation}\label{eq.amc}
\forall j \in M, \ \ \ |\sum_{i\in N} x^*_{ij} - \sum_{i\in N} W_{ij}| \le 3\epsilon.
\end{equation}

\subsection{Correctness}
In this section, we will show the following theorem.

\begin{theorem}\label{thm.red}
Computing $\frac{1}{n}$-approximate well-supported Nash equilibrium in game $(R,C)$ reduces to $\frac{1}{m^8}$-approximate exchange setting equilibrium of $\mathcal M(R,C)$. 
\end{theorem}

Let $(\ps,\x^*)$ be a $\frac{1}{m^8}$-approximate exchange setting equilibrium of $\mathcal M(R,C)$. Using the fact that $m=(2n+2)$ and equation \eqref{eq.amc}, the approximate supply-demand condition satisfied by the equilibrium is, 

\begin{equation}\label{eq.amc2}
\forall j \in M, \ \ \ |\sum_{i\in A} x_{ij}^* - \sum_{i\in A} W_{ij}| \le 3/m^8 < 1/n^8.
\end{equation}

Before starting the reduction, we introduce some notation. We work with absolute value of prices, but for notational simplicity we use $\ps_j > 0$ to denote the price of bad $j$ in absolute value at equilibrium. Also, for any subset of agents $\tilde{N} \subset N$, and any bad $j\in M$, let $W_j (\tilde{N}) = \sum_{i\in \tilde{N}} W_{ij}$. Similarly, define $x^*_j(\tilde{N}) = \sum_{i\in \tilde{N}} x_{ij}^*$.

The first step in each the reduction is to ensure that the prices are ``well-behaved'' \cite{ChenDDT09}. In particular, we will show that ratio of prices for any two bads is bounded above by two. The agents of $N_{PR}$ are constructed mainly for this task. Recall that the price of all bads are negative an equilibrium. 

\begin{lemma}\label{lem.pr}
For any $j,k \in M$, $\frac{\ps_j}{\ps_{k}}\le 2$.
\end{lemma}
To prove the above claim, first we need to understand consumption of $N_R$ and $N_C$ agents on the first segments of their cost functions. If the function $u_{ij}$ is non-linear then let $x_{ijk}$ denote the amount agent $i$ is assigned of bad $j$ on segment $k$. 

\begin{claim}\label{cl.1}
For any $s\neq s' \in [n]$, let $i=a_{(s,s',R)}$ and $i'=a_{(s',s,R)}$. Then for every bad $j \in [m-1]$, $W_{ij} + W_{i'j} \ge x^*_{ij1} + x^*_{i'j1}$. Similarly, for $i=a_{(s,s',C)}$ and $i'=a_{(s',s,C)}$.
\end{claim}
\begin{proof}
The claims follow from the construction of $N_R$ and $N_C$ respectively. 
Note that, allocation on the first segment can not be more than its length. Therefore, $x^*_{is1} \le 1/n^4 =W_{i,s}$ and for the remaining $j\in[n]$ both $x^*_{ij1}$ and $W_{ij}$ are zero. Similarly, $x^*_{i',s',1}\le 1/n^4=W_{i',s'}$ and $x^*_{i'j}=W_{i'j}=0$ for $j\in[n], j\neq s'$. For $j\in\{n+1,\dots, 2n\}$, we have $x^*_{ij1} \le (-r_{j-n})^+/n^6 = W_{i'j}$ and $x^*_{i'j1} \le (r_{j-n})^+/n^6 = W_{ij}$, where $r_k=(R_{s'k} - R_{sk})$. Similarly for $j=(m-1)$, $x^*_{ij1} \le (\sum_{k\in[n]} r_k)^+/n^6= W_{i'j}$, and $x^*_{i'j1} \le (-\sum_{k\in[n]} r_k)^+/n^6= W_{ij}$. The proof for $i=(s,s',C)$ and $i'=(s',s,C)$ follows similarly.
\end{proof}

The main intuition behind the proof of Lemma \ref{lem.pr} is that if the maximum to minimum price ratio is more than $2$, then $N_{PR}$ will leave out a $1/n$ amount of the minimum price bad to be consumed by the remaining agents. And among the remaining agents, if they have to consume this bad at high-slope segment (with disutility $H$) then they would prefer maximum priced bad instead. Claim to consume the mm \ref{cl.1} will come handy in if agents outside $N_{PR}$ have to consume the $1/n$ amount any bad from $[(m-1)]$ then they ``have to'' consume it on the high-slope segment.
\medskip

\noindent{\em Proof of Lemma \ref{lem.pr}.} 
Let $b'\in \argmax_{j\in B} \ \ps_{j}$, and $b\in \argmin_{j \in B}\ \ps_j$. Due to scale invariance of equilibria, we may assume without loss of generality that $\ps_{b} =1$.  

Assume for contradiction that $p_{b'} > 2p_{b}$. Note that $W_{b}(N_{PR}) = \frac{(m-1)}{n}$. In $N_{PR}$ all agents except of type $(b,j)$ and $(j,b)$ have slope $H$ for bad $b$, $\forall j \in M, \ j\neq b$. All of these agents will strictly prefer bad $b'$ over bad $b$, as it has lower pain per buck. Agents of type $(j,b)$ have cost $1$ for $j$ and $2$ for bad $b$, and therefore will strictly prefer to consume $j$ over bad $b$. Agent $(b,b')$ has cost $1$ for $b$ and $2$ for $b'$ but still strictly prefers $b'$ since $\frac{\ps_{b'}}{\ps_b}> 2\Rightarrow \frac{1}{\ps_b} > \frac{2}{\ps_{b'}}$. The remaining agents $(b,j)$ for $j\neq b,b'$, they bring $1/n$ units of $b$ and can consume only as much. Therefore $x^*_b(N_{PR}) \le (m-2)/n$. Due to \eqref{eq.amc2}, $N\setminus N_{PR}$ has to consume $(1/n -1/n^8)$ amount of bad $b$. We consider two cases: $b\in [m-1]$, or $b=m$.\\

\noindent\textbf{Case I:} $b\in [m-1]$. Any agent $a_j \in N_D$, other than $a_b$, strictly prefers bad $b'$ over $b$ since it provides lower pain per buck. Also, agent $a_b$ can only buy $\ps_{(m-1)}/n^8$ amount of $b$. Observe that $\ps_{(m-1)}\le \ps_{b'}\le 3H$. Indeed, if $\ps_{b'} > 3H$, then all agents strictly prefer bad $b'$ over any segment of $b$ and the setting cannot clear, even approximately. Thus, $x_b(N_D) \leq O(H/n^8)$, which leaves the remaining $O(1/n-H/n^8)$ of bad $b$ to be purchased by the agents of $N_R \cup N_C$.

By Claim \ref{cl.1}, $x^*_{b}(N_R\cup N_C) \le w_b(N_R\cup N_C) = O(1/n^3)$. Therefore, the agents of $N_R \cup N_C$ must consume $b$ at a disutility of $H$. However, in that case they prefer $b'$ over $b$, and the setting cannot clear, even approximately.\\

\noindent\textbf{Case II:} $b=m$. First, note that none of the agents of $N_D$ will consume $b=m$ since they strictly prefer $b'$. Thus, all demand for $m$ outside of $N_{PR}$ comes from $N_R \cup N_C$.

Observe, all agents of $N_C$ and $N_R$ must consume the bads $\{1,\dots,(m-1)\}$ on their first segments with disutility 1/3, before consuming any of bad $m$ with disutility $3$. Consider agent $i = a_{(s,t,R)} \in N_R$. Her total earnings are

\begin{equation} \label{eq:NR_earnings}
	\sum_{j \in M} W_{ij} \ps_j = \frac{\ps_s}{n^4} + \frac{1}{n^6} \sum_{k\in [n]}(R_{tk} - R_{sk})^+ \ps_{(n+k)} + \frac{1}{n^6} (\sum_{k\in[n]} R_{sk} - R_{tk})^+ \ps_{(m-1)}.
\end{equation}

The total cost of her first segments of bads $\{(n+1),\dots,(m-1)\}$ is
\begin{equation} \label{eq:NR_spending}
	\frac{1}{n^6} \sum_{k\in [n]}(R_{sk} - R_{tk})^+ \ps_{(n+k)} + \frac{1}{n^6} (\sum_{k\in[n]} R_{tk} - R_{sk})^+ \ps_{(m-1)}.
\end{equation}

After purchasing her first segments of $\{ (n+1),\dots,(m-1) \}$, her remaining money which can be used to purchase bad $m$ is 
\begin{align}
\frac{\ps_s}{n^4} &+ \frac{1}{n^6} \sum_{k\in [n]}(R_{tk} - R_{sk})^+ \ps_{(n+k)} + \frac{1}{n^6} (\sum_{k\in[n]} R_{sk} - R_{tk})^+ \ps_{(m-1)} \nonumber \\
&- \bigg(\frac{1}{n^6} \sum_{k\in [n]}(R_{sk} - R_{tk})^+ \ps_{(n+k)} + \frac{1}{n^6} (\sum_{k\in[n]} R_{tk} - R_{sk})^+ \ps_{(m-1)} \bigg) \nonumber \\
&= \frac{\ps_s}{n^4} +\frac{1}{n^6} \left[(\sum_{k \in [n]} R_{sk} - R_{tk}) \ps_{(m-1)} - \sum_{k \in [n]} (R_{sk} - R_{tk}) \ps_{(n+k)}\right] \nonumber 
\end{align}
\begin{align}
&= \frac{\ps_s}{n^4} +\frac{1}{n^6} \left[\sum_{k \in [n]} (R_{sk} - R_{tk}) (\ps_{(m-1)}-\ps_{(n+k)}) \right]. \label{eq:NR_money_left}
\end{align}

Now consider the agent $i' = a_{(t,s,R)} \in N_R$ paired with $i$. By the same argument, after purchasing the first segments of bads $\{(n+1),\dots,(m-1)\}$, $i'$ has the remaining money
\[ \frac{\ps_t}{n^4} +\frac{1}{n^6} \left[\sum_{k \in [n]} (R_{tk} - R_{sk}) (\ps_{(m-1)}-\ps_{(n+k)}) \right]. \]
Then the total money $i$ and $i'$ can spend on $m$ is $(\ps_s+\ps_t)/n^4 = O(H/n^4)$, since $\ps_{b'} < O(H)$. A similar argument shows that for any pair $i=a_{(s,t,C)}, i' = a_{(t,s,C)}$ in $N_C$, after spending on their first segments of bads $\{1,\dots,n \}$, their remaining money to spend on $m$ is at most $O(H/n^4)$. Summing over all pairs $i$ and $i'$ shows that the total amount the agents of $N_R \cup N_C$ can spend on $m$ is $O(H/n^2) < (1/n-1/n^8)$, for $n$ sufficiently large. Then setting cannot clear, even approximately.\qed\\

Let us set $H=10$. Scale $\ps$ so that the minimum price is $1$. Then using Lemma \ref{lem.pr} we have that $\ps\in[1,2]^{m}$. 
Define $u_s = 2-\ps_s$ and $v_i = 2-\ps_{(n+s)}$ for all $s \in[n]$. Note that, $\uu,\vv\in [0,1]^n$ but they need not be probability distributions. We will show Nash equilibrium conditions \eqref{eq:ne} for $(\uu,\vv)$ first. 

\begin{lemma}\label{lem:ne-cond}
For any $s, t \in [n]$ if $(R\vv)_s - (R\vv)_t < -1/n$ then $u_s=0$, and if $(\uu^T C)_s - (\uu^T C)_t < -1/n$ then $v_s=0$.
\end{lemma}

To show the above lemma, we will first make a couple of simple observations. Define agent set $\bAPR = N\setminus N_{PR}$ to be everyone outside $N_{PR}$. In order to get $u_s=0$ we need to force $\ps_s=2$. Next, we show that if $\bAPR$ leave leave a large enough amount of bad $s$ to be consumed by $N_{PR}$ then it must be that $\ps_s=2$.

\begin{observation}\label{ob.1}
No agent consumes any bad at cost $H$. 
\end{observation}
\begin{proof}
Note that for every agent $i$ there is at least one bad $j$ such that the cost function $f_{ij}$ is linear with slope at most $3$. Given that $H=10$ and $10/3 >2$, the claim follows using Lemma $\ref{lem.pr}$. 
\end{proof}

\begin{observation}\label{ob.2}
For any bad $b\in M$, $(i)$ $W_b(\bAPR) - x^*_b(\bAPR) \ge 1/n^8 \Rightarrow \ps_b=2$, and $(ii)$ $W_b(\bAPR) - x^*_b(\bAPR) \le -1/n^8 \Rightarrow \ps_b=1$
\end{observation}
\begin{proof}
If $W_b(\bAPR) - x^*_b(\bAPR) \ge 1/n^8$, then the supply-demand condition \eqref{eq.amc} requires that $W_b(N_{PR}) - x^*_b(N_{PR}) <0$. Recall that, for any agent $(j,j')\in N_{PR}$ the cost for bad $j$ is $1$, that of bad $j'$ is $2$, and $H$ for all other bads. At equilibrium she only purchases bads that minimizes cost/price. Therefore, if $\ps_b<2$ then from $N_{PR}$, only the agents $(b,j'),\ \forall j'\neq b$ are willing to consume bad $b$. Each of them brings exactly $1/n$ amount of bad $b$ and can consume only what they bring, which contradicts $W_j(N_{PR}) - x^*_j(N_{PR}) <0$. Instead, if $\ps_b=2$, then agents of type $(j',b)$ may also be able to consume bad $b$. 

Now suppose $W_b(\bAPR) - x^*_b(\bAPR) < -1/n^8$, then the supply-demand condition \eqref{eq.amc} requires that $W_b(N_{PR}) - x^*_b(N_{PR}) > 0$. Notice that agents $(b,j)$, $\forall j \neq b$ of bring all the of the bad $b$ from $N_{PR}$. Moreover, these agents strictly prefer $b$ over $j$ when $\ps_b > 1$, since all $\ps_j \in [1,2]$ by Lemma \ref{lem.pr} (given that the minimum price is 1). Therefore $W_b(N_{PR}) - x^*_b(N_{PR}) \leq 0$, which contradicts approximate supply-demand.
\end{proof}

\begin{observation}\label{ob.3}
$\ps_{(m-1)}=2$. 
\end{observation}
\begin{proof}
Claim \ref{cl.1} together with Observation \ref{ob.1} implies that $W_{(m-1)}(N_R \cup N_C) \ge x^*_{(m-1)}(N_R \cup N_C)$. And from Observation \ref{ob.1} we also know that $x^*_{(m-1)}(N_D)=0$ while $W_{(m-1)}(N_D) = 2n/n^8 > 1/n^7$. Both of these together gives, $W_{(m-1)}(\bAPR) - x^*_{(m-1)}(\bAPR) > 1/n^7$ and then Observation \ref{ob.2} implies $\ps_{(m-1)}=2$.
\end{proof}

Using Observation \ref{ob.1} next we show a stronger version of Claim \ref{cl.1}.

\begin{observation}\label{ob.4}
For any $s\neq s' \in [n]$, let $i=a_{(s,s',R)}$ and $i'=a_{(s',s,R)}$. Then for every bad $j \in [m-1]$, $W_{ij} + W_{i'j} \ge x^*_{i,j} + x^*_{i',j}$, and the condition holds with equality for $j\in \{(n+1),\dots,(m-1)\}$. Similarly, for $i=a_{(s,s',C)}$ and $i'=a_{(s',s,C)}$ where the equality holds for $j\in [n] \cup \{(m-1)\}$.
\end{observation}
\begin{proof}
The inequality follows using the fact that no agent consumes any bad at cost $H$ by Observation \ref{ob.1}. To show equality for $j>n$ when $i=a_{(s,s',R)}$ and $i'=a_{(s',s,R)}$, we only need to make sure that $i$ and $i'$ consumes the first segment of bad $j$ completely. 

For agent $i$, recall that the cost of any $j\in\{(n+1),\dots,(m-1)\}$ on the first segment is $1/3$, compared to that the next lowest cost of $1$ (on the first segment of bad $s$). Since price ratios of any two bads is not more than $2$, she prefers to buy all the segments with cost $1/3$ first. To consume all of them she needs money of,
\[
\frac{1}{n^6} \sum_{k\in [n]}(R_{sk} - R_{tk})^+ \ps_{(n+k)} + \frac{1}{n^6} (\sum_{k\in[n]} R_{tk} - R_{sk})^+ \ps_{(m-1)} \le \frac{4n}{n^6}
\]

The above inequality follows from the fact that $R\in[0,1]^{n\times n}$ and every price is at most $2$ (given that the minimum price is $1$). 
Her endowment of $1/n^4$ amount of bad $s$ earns her at least $1/n^4$ units of money, which is enough to buy the above bundle. Similar argument follows for agent $i'$. 

The case for for $i=a_{(s,s',C)}$ and $i'=a_{(s',s,C)}$ follows similarly.
\end{proof}

Now to prove Lemma \ref{lem:ne-cond} we only need to show that agents of set $\bAPR$ leave at least $1/n^8$ amount of bad $s$ to be consumed by $N_{PR}$. We will crucially use the above observations for the same.
\medskip

\noindent{\em Proof of Lemma \ref{lem:ne-cond}.}
Let $i=a_{(s,t,R)}\in N_R$. The total earning of $i$ is given by \eqref{eq:NR_earnings}. Given that price ratio of any two bads is bounded above by $2$ (Lemma \ref{lem.pr}), agent $i$ will first buy the first segments of bads $\{(n+1),\dots,(m-1)\}$, then the first segment of bad $s$, and then consume bad $m$ if any money left. By Observation \ref{ob.4}, agent $i$ spends exactly \eqref{eq:NR_spending} on the bads $\{(n+1),\dots, (m-1)\}$.

The left over money after the above spending is given by \eqref{eq:NR_money_left}
\[
\begin{array}{l} 
\ \ \ \frac{\ps_s}{n^4} +\frac{1}{n^6} \left[\sum_{k \in [n]} (R_{sk} - R_{tk}) (\ps_{(m-1)}-\ps_{(n+k)}) \right] \\[10pt]
\ \ \ = \frac{\ps_s}{n^4} +\frac{1}{n^6} \sum_{k \in [n]} (R_{sk} - R_{tk}) v_k  \\[10pt]
\ \ \ = \frac{\ps_s}{n^4} +\frac{1}{n^6} \big[(R\vv)_s - (R\vv)_t\big] \\[10pt]
\ \ \ < \frac{\ps_s}{n^4} - \frac{1}{n^7} \le \ps_s \left(\frac{1}{n^4} - \frac{1}{2n^7}\right)\\
\end{array}
\]

The last inequality above follows from $\ps_s\le 2 \Rightarrow -1/\ps_s \le -1/2$. Thus, agent $i$ can only consume at most $\left(\frac{1}{n^4} - \frac{1}{2n^7}\right)$ amount of bad $s$ from the $1/n^4$ units that she brings. By Observation \ref{ob.4}, agents of $N_R$ and $N_C$ are unable to consume any of the remaining $\frac{1}{2n^7}$ amount of $s$. Among the agents of $N_D$ only agent $a_s$ can consume $s$, but only up to the amount of $\frac{\ps_{(m-1)}}{n^8\ps_s} \le \frac{2}{n^8}$. Thus, $W_s(\bAPR) - x_s^*(\bAPR) \ge \frac{1}{2n^7} - \frac{2}{n^8} > \frac{1}{n^8}$, and then using Observation \ref{ob.2} it follows that $\ps_s=2$. This in turn implies $u_s=2-\ps_s =0$. 

The second part of the lemma follows similarly. \qed
\medskip

\begin{observation}\label{ob.5}
Let $s\in \argmax_{k \in [n]} (R\vv)_k$, then $\ps_s=1$. 
Similarly, for $s\in \argmax_{k \in [n]} (\uu C)_k$, $\ps_{(n+s)}=1$. 
\end{observation}
\begin{proof}
As observed in the proof of Lemma \ref{lem:ne-cond}, for all $t \in [n], t\neq s$, the corresponding agent $i=a_{(s,t,R)}\in N_R$, the total earning minus spending on the segments with slope $1/3$ is,
\[
\frac{\ps_s}{n^4} +\frac{1}{n^6} \big[(R\vv)_s - (R\vv)_t\big] \ge \frac{\ps_s}{n^4},
\]
since $(R\vv)_s - (R\vv)_t\ge 0$. Therefore, $i$ will consume all of the bad $s$ that she brings, i.e., $1/n^4$ amount. By Observation \ref{ob.4} remaining agents of $N_R$ and $N_C$ together will consume all of $s$ that they bring to the setting. Among $N_D$, the agent $a_s$ brings $1/n^8$ units of $(m-1)$ and wants to consume only bad $s$. Price of $2$ for bad $(m-1)$ (Obs. \ref{ob.3}), gives earning of $2/n^8$ to agent $a_s$. Since price of $s$ can be at most $2$ (by Lemma \ref{lem.pr} assuming the minimum price is 1), she will demand at least $1/n^8$ units of bad $s$.   
Putting supply and demand of all the agents in $\bAPR$ together we get,
\[
W_s(\bAPR) - x_s^*(\bAPR) \le -1/n^8
\]

From the above inequality, Observation \ref{ob.2} gives us $\ps_s=1$. The second part follows similarly.
\end{proof}

The above observation ensures that vectors $\uu$ and $\vv$ are not zero vectors. Using this together with Lemma \ref{lem:ne-cond} we will show the main theorem next. 
\medskip

\noindent{\em Proof of Theorem \ref{thm.red}}
Given an equilibrium $(\ps,x^*)$ of setting $\mathcal M(R,C)$ construct vectors $\uu$ and $\vv$ as defined above,
\[
\forall s \in [n], \ \ u_s= 2- p_s^*, \ \ \mbox{ and } v_s=2-p_{(n+s)}^*
\]

Observation \ref{ob.5} ensures that $\uu, \vv\neq \0$. Construct strategy vectors 
\[
\balpha=\frac{\uu}{\sum_{s\in[n]} u_s},\ \ \ \mbox{ and } \ \ \ \bbeta=\frac{\vv}{\sum_{s\in[n]} v_s}
\]

We will show that $(\balpha,\bbeta)$ satisfies Nash equilibrium conditions \eqref{eq:ne} for $\epsilon=1/n$. For any pair of strategies $s,t\in [n]$
\[
\begin{array}{lcl}
(R\bbeta)_s - (R\bbeta)_t < -1/n & \Rightarrow & (\sum_{s\in[n]} v_s)^{-1} ((R\vv)_s - (R\vv)_t) < -1/n\\
& \Rightarrow & ((R\vv)_s - (R\vv)_t) < -1/n\ \ \ \ (\because (\sum_{s\in[n]} v_s)\ge 1 \mbox{ using Obs. \ref{ob.5}})\\
& \Rightarrow & \ps_s=2 \ \ \ \ (\because \mbox{ Lemma \ref{lem:ne-cond}})\\
& \Rightarrow & u_s=0 \ \ \ \ (\because u_s = 2 - \ps_s)\\
& \Rightarrow & \alpha_s=0 \ \ \ \ (\because u_s = 2 - \ps_s)
\end{array}
\]

By the similar argument as above we can show that for any pair $s,t \in [n]$ if $(\balpha^TC)_s - (\balpha^T C)_t < -1/n$ then $v_s=0$. \qed
\medskip

\noindent{\em Proof of Theorem \ref{thm:hardness}.}
	Clearly, our construction is polynomial (in number of agents and bads) with respect to the input size of the 2-player game $(R,C)$. Theorem \ref{thm.red} shows that a $1/m^8$-approximate exchange setting equilibrium yields a $1/n$-approximate well-supported NE of $(R,C)$. Moreover, this conversion from exchange equilibrium to NE runs in polynomial time. As finding a $1/n$-approximate well-supported NE is $\classPPAD$-hard, it follows that finding a $1/poly(n)$-approximate exchange equilibrium is $\classPPAD$-hard. \qed
\medskip

\noindent{\em Proof of Theorem \ref{thm:ceei-hardness}.}
	We will show how to convert our exchange setting construction into a Fisher setting with equal budgets in a way such that the equilibria of the two settings are in one-to-one correspondence. Moreover, a $O(\epsilon/n)$ approximate equilibrium in the Fisher setting is a $O(\epsilon)$ approximate equilibrium in the exchange setting. Both the reduction from exchange to Fisher setting, and conversion of Fisher to exchange equilibrium take polynomial time which establishes the $\classPPAD$-hardness of the Fisher case. We use the notation $(\w,\x,\p)$ and $(\tilde{\w}, \tilde{\x}, \tilde{\p})$ to denote endowments, allocations, and prices in the exchange and Fisher settings respectively.
	
	First we show how to convert our exchange setting construction into a Fisher setting. We ensure equal budgets by giving each agent an equal amount of each bad. Recall in the construction of the exchange setting, the agents of $N_{PR}$ bring the largest amount of each bad $j\in [m]$, specifically a $1/n$ amount. In the Fisher setting, we give each agent a $\tilde{W}_{ij} =1/n$ amount of each bad $j$, by increasing $i$'s initial endowment of each bad $j$ by $1/n-W_{ij}$. As all agents bring the same amount of each bad, the new instance is a Fisher setting with equal budgets.
	
	Next, we want to ensure a one-to-one correspondence between the equilibria of the two settings. For each agent $i$ and bad $j$ we add an extra segment to $f_{ij}$ with disutility $\tilde{d}_{ij1}=0$ and length $\tilde{L}_{ij1} = 1/n-W_{ij}$, the additional amount of bad $j$ given to agent $i$. Note that this new segment becomes the first segment of $f_{ij}$, and `shifts' all other segments to the right. Moreover, for any set of prices $p$, these newly created first segments are: optimal purchases since they have 0 disutility, and can be fully purchased since the cost $\tilde{L}_{ij1} p_j = (1/n-W_{ij})p_j$ is exactly equal to extra money $i$ earns from bad $j$ in the Fisher setting. Thus, after purchasing all the new first segments, $i$ is left with a remaining budget of $\sum_{j\in M} \tilde{W}_{ij} p_j - \sum_{j \in B} \tilde{L}_{ij1}p_j = \sum_{j\in M} W_{ij} p_j$, i.e., her budget in the exchange setting. Finally, since we shift all remaining segments of $f_{ij}$ to the right, it is easy to check that $i$'s forced, flexible, and undesirable segments are the same (up to the newly added first segments) in both the exchange and Fisher settings for any prices $p$. Thus, if $x_i = (x_{ij})_{j\in M}$ is an affordable and optimal bundle in the exchange setting, then $\tilde{x}_i = (L_{ij1} + x_{ij})_{j\in M}$ is an affordable and optimal bundle in the Fisher setting. It follows that the set of equilibria of the two settings are in one-to-one correspondence.
	
	Finally, we show that if $(\tilde{\x}^*,\tilde{\p}^*)$ is an $\epsilon/(6n)$ approximate equilibrium of the Fisher setting, then $(\x^*,\p^*)$ is an $\epsilon$ approximate equilibrium of the exchange setting where
	\begin{equation} \label{eq:new_eq}
		x_{ij}^* = \tilde{x}_{ij}^* - \tilde{L}_{ij1}, \quad p_j^* = \tilde{p}_j^*, \ \forall i,j.
	\end{equation} 
	As previously argued, if $\tilde{x}_i^*$ is an affordable and optimal bundle for the Fisher setting, then $x_i^*$ is an affordable and optimal bundle for the exchange setting. Thus, if condition C1 holds for $(\tilde{\x}^*,\tilde{\p}^*)$, then it holds for $(\x^*,\p^*)$. It remains to check condition C2'. Observe, that there are $$|N| = 2\frac{m(m-1)}{2} + 2n +2\frac{n(n-1)}{2} + 2\frac{n(n-1)}{2} = 6n^2+6n+2$$ 
	agents used in the construction, and $\tilde{W}_{ij} = 1/n$ for all $i,j$. Thus, $\sum_{i\in N} \tilde{W}_{ij} = 6n+6+2/n < 6n+7$, for all $j\in M$. By \eqref{eq:new_eq}, $|\sum_{i\in N} (\tilde{W}_{ij} - \tilde{x}_{ij}^*) | = |\sum_{i\in N} (W_{ij} - x_{ij}^*)|$, since $\tilde{L}_{ij1} = 1/n-W_{ij}$. Therefore, if $(\tilde{\x^*},\tilde{\p^*})$ is $\epsilon/(6n)$ approximate equilibrium of the Fisher setting, then
	\begin{equation}
		|\sum_{i\in N} (W_{ij} - x_{ij}^*)| = |\sum_{i\in N} (\tilde{W}_{ij} - \tilde{x}_{ij}^*) | \leq \frac{\epsilon}{6n} \sum_{i\in N} \tilde{W}_{ij} \leq \epsilon(1+7/(6n)) <  3\epsilon, \ \forall j \in M,
	\end{equation}
	so that $(\x^*,\p^*)$ is an $\epsilon$ approximate equilibrium for the exchange setting by \eqref{eq.amc}.
	
	It follows that computing a Fisher setting equilibrium (with equal budgets) is $\classPPAD$-hard, since both the reduction from exchange to Fisher setting, and the conversion from Fisher to exchange equilibrium take polynomial time (and space). \qed

\section{Numerical Experiments}\label{sec:experiments}
Table \ref{tab:numerical} summarizes the results of numerical experiments conducted on randomly generated trials using a Matlab implementation of our algorithm. Note that we used the same number of segments, shown as \#Seg in the table, for each agent and each item. We drew the $U_{ijk}$'s, $L_{ijk}$'s, and $W_{ij}$'s uniformly at random from the intervals $[-1,0]$, $[0,1/\#\text{Seg}]$ , and $[0,1]$ respectively. Then, we rescaled the $W_{ij}$ values to ensure a unit amount of each bad. Finally, for each agent $i$ and each bad $j$, we sorted the $U_{ijk}$'s decreasing order to generate an SPLC utilities. 

\begin{table}[tbh!]
	\centering
	\begin{tabular}{|l|l|l|l|l|}
		\hline
		$N\times M \times$ \#Seg & Instances & Min iters & Mean iters & Max iters \\ \hline
		$5\times 5\times 5$ & 1000 & 85 & 137.3 & 297 \\ \hline
		$10 \times 5 \times 5$ & 1000 & 107 & 170.9 & 395 \\ \hline
		$10 \times 10 \times 5$ & 1000 & 130 & 369.1 & 609 \\ \hline
		$15 \times 15 \times 5$ & 50 & 168 & 750.3 & 1393 \\ \hline
		$20 \times 20 \times 5$ & 10 & 1127 & 1398.2 & 2001 \\ \hline
	\end{tabular}
	\caption{Experimental results conducted on random instances.}
	\label{tab:numerical}
\end{table}

\begin{figure}[tbh!]
	\centering
	\includegraphics[scale=0.75]{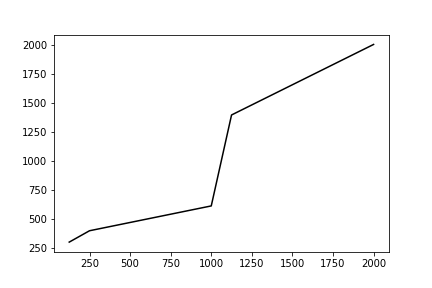}
	\caption{Plot of $N\times M \times$ \#Seg versus maximum number of iterations.}
	\label{fig:results}
\end{figure}

Figure \ref{fig:results} compares the maximum number of iterations versus the total number of segments in the agent's utility functions, i.e., $N \times M \times$ \#Seg $= \sum_{i,j} |u_{ij}|$. Note that even in the worst case, the maximum number of iterations is on the order of the total number of segments of the agents' utility functions.

\appendix
\section{Converting Bads into Goods?}\label{sec:bads_to_goods}
Bogomolnaia et al.~\cite{BogomolnaiaMSY17} propose a method to convert a competitive division problem with bads into a problem with only goods. Note that their argument only applies to the Fisher setting and uses linear utility functions. The approach relies on the interpretation of leisure as the opposite of work. Therefore, if agent $i$ is assigned an $x_{ij}$ fraction of bad $j$, then we can equivalently view this as a good representing an exemption from completing a $1-x_{ij}$ fraction of the task. 

The reduction from bads to goods proposed by~\cite{BogomolnaiaMSY17} is as follows. Assume there are $n$ agents in the competitive division problem. For each bad $j$, we create $n-1$ units of a good $j'$ representing an exemption from completing bad $j$. Suppose agent $i$ has utility $D_{ij} <0$ for bad $j$, then $i$'s utility for good $j'$ is SPLC with two segments. The first segment has slope $|D_{ij}| > 0$ and length $L_{ij} = 1$, and the second segment has slope 0. Note that this means $i$ values up to 1 unit of exemption to the bad $j$. 

Bogomolnaia et al.~\cite{BogomolnaiaMSY17} state that given an equilibrium $(\x',\p')$ in the problem of goods, one can obtain an equilibrium in bads by setting $p_j^* = -p_j'$, and $x_{ij}^* = 1-x_{ij}'$. \\

\noindent\textbf{Counter Example:} Consider a competitive division problem with two agents $a$ and $b$, and three bads 1, 2, and 3. The agents' utility functions are: $u_a(x) = -10x_{a1}-2x_{a2}-x_{a3}$, and $u_b(x) = -x_{a1} -100x_{a2}-100x_{a3}$. We create 1 unit of exemption for each bad. The utility functions for agent $a$ are SPLC where the first segment has slope $(10,2,1)$ for goods 1, 2, and 3 respectively, and are capped at 1 unit of good. One can verify that the prices $\p' = (4/3,1/3,1/3)$, and the allocations $x_a' = (3/4,0,0)$ and $x_b = (1/4,1,1)'$ are an equilibrium in goods. In bads, this becomes $\p^* = (-4/3,-1/3,-1/3)$, with the allocation $x_a^* =(1/4,1,1)$ and $x_b^* = (3/4,0,0)$. However, this is not a competitive equilibrium since $a$ does not receive the same $pbp$ for all bads. One can check that the prices $\p^* = (-20/13,-4/13,-2/13)$, along with the allocation $x_a^* =(7/20,4/13,2/13)$ and $x_b^*=(13/20,0,0)$ give an equilibrium.

\section{Approach of~\cite{Eaves76,GargMSV15} Gets Stuck on Secondary-Rays}\label{app:sr}  
Previous works of Eaves~\cite{Eaves76} and Garg et al.~\cite{GargMSV15} developed complementary pivot algorithms based on Lemke's scheme for all goods under SPLC utilities. The basic structure of our LCP is similar to prior works. However, they use a different change of variables. Both~\cite{Eaves76,GargMSV15} use a \emph{lower bound} on prices by making the price of good $j$: $1+p_j$, where $p_j \geq 0$. Thus, the \emph{minimum} price is 1 (in absolute value). In addition,~\cite{Eaves76,GargMSV15} make no changes to the variable $r_i= 1/ppb_i$, where $ppb_i$ is the pain per buck of agent $i$'s flexible segment. In this section, we examine this change of variables when applied to the special case of all bads with linear utilities. The resulting formulation is as follows:
\begin{align}
&\forall i \in N,\ \ \sum_{j\in M} W_{ij} p_j - \sum_{j \in M} f_{ij}-\epsilon_i z \le -\sum_{j\in M} W_{ij} & \perp  r_i \label{eq:bL_sr_budget}\\
&\forall j \in M,\ \ \sum_{i \in N} f_{ij} - p_j\le 1 & \perp  p_j \label{eq:bL_sr_spending}\\
&\forall i \in N,\ \forall j \in M,\ \ \ p_j - D_{ij}r_i -\delta_{ij} z \le -1 & \perp  f_{ij} \label{eq:bL_sr_mpb}
\end{align}

The constraints have the same interpretation as before: a budget constraint for all agents \eqref{eq:bL_sr_budget}, a constraint on the total spending of agents for each bad \eqref{eq:bL_sr_spending}, and a minimum pain per buck constraint for each agent, for each bad \eqref{eq:bL_sr_mpb}. Note that we add coefficients $\epsilon_i$, and $\delta_{ij}$ to $z$ for all terms with negative rhs for two purposes. First, this provides a degree of control over the primary ray, i.e., the initial double label, and therefore  how the algorithm starts. Second, we require $\delta_{ij}$'s coefficients to ensure nondegeneracy of LCP when $z>0$. To see this, suppose $p_j = 0$ for some $j\in M$, and $r_i = 0, \ \forall i\in N$. Then by setting $z = 1$, the constraints \eqref{eq:bL_sr_mpb} become tight (hold with equality) for this $j$, $\forall i \in N$. Thus, there is no unique double label. 

We now examine the behavior of Lemke's algorithm when starting from a constraint \eqref{eq:bL_sr_budget} or \eqref{eq:bL_sr_mpb}. We show that in both cases the algorithm quickly reaches a secondary ray.

\subsection{Starting from \eqref{eq:bL_sr_budget}}
Suppose we select $\epsilon_i = 1, \ \forall i \in N$, and ensure $1/\delta_{ij} < \max_k \sum_j W_{kj}$. By setting $z = \max_k \sum_j W_{kj}$ and all other variables $(\p,\r,\f) = \0$, we obtain a unique double label for constraint \eqref{eq:bL_sr_budget} for agent $a = \arg\max_k \sum_j W_{kj}$. Specifically, all constraints \eqref{eq:bL_sr_mpb} hold with strict inequality.

Lemke's algorithm then fixes $z= \max_k \sum_j W_{kj} = \sum_j W_{aj}$, and increases $r_a$. Observe that $r_a$ only appears in the constraints \eqref{eq:bL_sr_mpb}. However, since $\delta_{aj}z>1$, and $D_{aj} >0, \ \forall j\in M$, increasing $r_a$ never makes any inequality \eqref{eq:bL_sr_mpb} tight for any $j\in M$. That is, we arrived at secondary ray. Notice that the same problem arises regardless of which budget constraint \eqref{eq:bL_sr_budget} we start from (assuming appropriate choice of $\epsilon_i$'s and $\delta_{ij}$'s). Therefore, starting from a budget constraint \eqref{eq:bL_sr_budget} \emph{always} leads to a secondary ray.

\subsection{Starting from \eqref{eq:bL_sr_mpb}}
Suppose we fix $\epsilon_i =1, \ \forall i\in N$, and $\delta_{ij}$'s such that $\max_k \sum_j W_{kj} < \max_{ij} 1/\delta_{ij}$. Then, setting $z = \max_{ij} 1/\delta_{ij}$ and all other variables $(\p,\r,\f) = \0$, yields the unique double label at constraint \eqref{eq:bL_sr_mpb} for the pair $(a,b) = \arg\min_{ij} \delta_{ij}$, i.e., the agent $a\in N$ and bad $b\in M$ that achieve $\max_{ij} 1/\delta_{ij}$. Further, all constraints \eqref{eq:bL_sr_budget} hold with strict inequality.  

Lemke's algorithm fixes $z = 1/\delta_{ab}$, and increases $f_{ab}$ until some other inequality becomes tight. Note that \eqref{eq:bL_sr_budget} can not become tight due to our choice of $z$. Then, \eqref{eq:bL_sr_spending} becomes tight for bad $b$. At this point, we may change $p_j$ subject to the following constraints
\begin{align*}
f_{ab} &= 1 + p_b\\
1+p_b &= \delta_{ab}z.
\end{align*}
We check whether a constraint of form \eqref{eq:bL_sr_budget} or \eqref{eq:bL_sr_mpb} can become tight.\\

\noindent\textbf{Starting from Pain per Buck Constraints \eqref{eq:bL_sr_mpb}.}
For any $i \neq a$, we would require that \eqref{eq:bL_sr_mpb} becomes tight while observing the relationship between $f_{ab}$, $p_b$, and $z$ above. Thus, we need
\begin{equation*}
1+p_b = \delta_{ib}z = \frac{\delta_{ib}}{\delta_{ab}} (1+p_b) > 1+p_b,
\end{equation*}
since $\delta_{ab} = \min_{ij} \delta_{ij}$. Thus, no constraint of the form \eqref{eq:bL_sr_spending} can become tight.\\

\noindent\textbf{Budget Constraints \eqref{eq:bL_sr_budget}.} 
For any $i\in N$ (including $a$), we would require that \eqref{eq:bL_sr_budget} becomes while maintaining the relationship between $f_{ab}$, $p_b$, and $z$ above. Thus, for $i \neq a$ we need
\begin{equation*}
W_{ib}\ p_b + \sum_j W_{ij} = z = \frac{1+p_b}{\delta_{ab}},
\end{equation*}
or after rearranging 
\begin{equation*}
\underbrace{\Sigma_j W_{ij} - 1/\delta_{ab}}_{<0} =  \underbrace{\big(1/\delta_{ab} - W_{ib} \big) }_{>0} \ p_b,
\end{equation*}
where the inequalities of the coefficients follow from $1/\delta_{ab} > \max_k \sum_j W_{kj} \geq W_{ib}$. However, no value of $p_b \geq 0$ suffices. Similarly, if we want \eqref{eq:bL_sr_budget} to become tight for agent $a$, then we need
\begin{equation*}
\underbrace{\Sigma_j W_{ij} - 1/\delta_{ab}-1}_{<0} =  \underbrace{\big(1/\delta_{ab} - W_{ib} +1\big) }_{>0} \ p_b,
\end{equation*}
and again no value of $p_b > 0$ works.\\

\noindent\textbf{Conclusion:} The examples demonstrate that for this relationship between $f_{ab}$, $p_b$, and $z$, no constraints can become tight, i.e., we have reached a secondary ray.

\section{Omitted Proofs}\label{sec:proofs}

\subsection{Proof of Lemma~\ref{lem:market_solution_1}}\label{sec:lem:clearing_1}
We will need the following lemma to prove the result. 
\begin{lemma}\label{lem:clearing_1}
	If $\p^*$ is an equilibrium price vector, then $\exists \ \f$ such that \eqref{eq:budget_1} and \eqref{eq:spending_1} hold. Further, if $\p$ and $\f$ satisfy \eqref{eq:budget_1} and \eqref{eq:spending_1}, and $\p > 0$, then the market clears.
\end{lemma}

\begin{proof}
	Let $\p^*$ be an equilibrium price vector and set $\p = |\p^*|$. Let $\x^*$ be the equilibrium allocation for $\p^*$. For each agent $i$ and each bad $j$, we distribute $x_{ij}^*$ among individual segments by filling starting from the first segment until all of $x_{ij}^*$ is used
	\begin{equation} \label{eq:x_from_allocation}
		x_{ijk}^* = \min \bigg( \max \bigg( x_{ij}^* - \sum_{k' < k} L_{ijk}, \ 0  \bigg) ,\ L_{ijk} \bigg).
	\end{equation}
	The market clearing conditions ensures that setting $f_{ijk} = x_{ijk}^* p_j$ together with $\p$ satisfies \eqref{eq:budget_1} and \eqref{eq:spending_1}.
	
	Next, suppose $\p,\f$ satisfy \eqref{eq:budget_1} and \eqref{eq:spending_1} and $\p>\0$. Summing \eqref{eq:budget_1} over all $i\in N$ and \eqref{eq:spending_1} over all $j \in M$ gives
	\begin{equation*}
		\sum_{j} p_j = \sum_{i,j} W_{ij} p_j \leq \sum_{i,j,k} f_{ijk} \leq \sum_{j} p_j,
	\end{equation*}
	 where the first equality uses the fact that there is a unit amount of each bad, i.e., $\sum_{i} W_{ij} =1$, $\forall j \in M$. It follows from the non-negativity of all variables that all constraints \eqref{eq:budget_1} and \eqref{eq:spending_1} hold with equality. Therefore, setting $x_{ijk} = f_{ijk}/p_j$ ensures the market clears. 
\end{proof}
\begin{proof}(of Lemma~\ref{lem:market_solution_1})
	Let $(\x^*,\p^*)$ be a competitive equilibrium. Define $\p$ and $\f$ as in Lemma \ref{lem:clearing_1}, and set $r_i$ according to \eqref{eq:lambda_ppb} for any segment $(i,j,k)$ of $i$'s flexible partition. Note that $r_i > 0$ since $0 < D_{ij1} \leq D_{ij2} \leq \dots$ for all bads $j\in M$. By Lemma \ref{lem:clearing_1}, \eqref{eq:budget_1} and \eqref{eq:spending_1} hold with equality since $(\x^*,\p^*)$ clears the market. Therefore, so do (\ref{eq:budget_1}') and (\ref{eq:spending_1}').
	
	We set the variables $s_{ijk}$ as follows: if $(i,j,k)$ is undesirable or flexible set $s_{ijk} = 0$, if $(i,j,k)$ is a forced segment set $s_{ijk}$ to satisfy
	\begin{equation*}
		\frac{1}{r_i} = \frac{D_{ijk}}{p_j - s_{ijk}} \ \ \Rightarrow \ \ s_{ijk} = p_j - D_{ijk}r_i.
	\end{equation*}
	Note that $s_{ijk} \geq 0$ since $D_{ijk} > 0$ and $D_{ijk}/p_j \le \frac{1}{r_i}$ for forced segments. It can be easily verified that in each case: the segment is forced, flexible, or undesirable; the constraints \eqref{eq:ppb_1} and \eqref{eq:segment_1}, and corresponding complementarity conditions (\ref{eq:ppb_1}') and (\ref{eq:segment_1}') are satisfied.
\end{proof}

\subsection{Proof of Theorem~\ref{thm:equiL_soL_correspondence}}\label{sec:lem:equiv_sol}
We will need the following lemma to prove the theorem. 
\begin{lemma} \label{lem:optimaL_bundles}
	In any solution to LCP~\eqref{lcp:2} with $p_j < P, \ \forall j\in M$, and $r_i < R, \ \forall i\in N$, agents receive an optimal bundle of bads.
\end{lemma}

\begin{proof} 
	Recall that $D_{ijk} = |U_{ijk}| >0, \ \forall\ (i,j,k)$ since for each agent, her utility for each bad is a concave, decreasing function, i.e., $0<D_{ij1} \leq D_{ij2} \leq \dots$. Due to scale invariance of competitive equilibria, we may pick any maximum price (in absolute value) $P$. Given the choice of $P$, we selected $R$ such that $R > P/ \min_{ijk} D_{ijk}$. This ensures $D_{ijk}R - P > 0$, $\forall (i,j,k)$, which makes the right hand side of \eqref{eq:ppb_2} positive for all segments $(i,j,k)$. Notice this implies that $r_i > 0, \ \forall\ i \in N$, otherwise \eqref{eq:ppb_2} is a strict inequality. In turn, (\ref{eq:ppb_2}') forces $f_{ijk} = 0, \ \forall\ (i,j,k)$. Then, for each $j\in M$, $p_j < P$ implies \eqref{eq:budget_2} is strict and therefore $p_j=0$ which will violate inequality \eqref{eq:spending_2}. 
	
	Let $(i,j,k)$ be a segment with highest $ppb$ that agent $i$ spends on in the solution to LCP~\eqref{lcp:2}. Define $\sigma_i$ as the inverse of the pain per buck of the this segment
	\begin{equation*}
		\sigma_i= \frac{P-p_j}{D_{ijk}} = \frac{1}{ppb_{ijk}}.
	\end{equation*}
	Observe that $\sigma_i > 0$, since $p_j<P$, and $D_{ijk} > 0$, for all segments of all bads. 
	
	We want to show that $(R-r_i) \leq \sigma_i$. Since agent $i$ spends on the segment $(i,j,k)$, i.e., $f_{ijk} >0$, complementarity condition (\ref{eq:ppb_2}') requires constraint \eqref{eq:ppb_2} holds with equality. Since $s_{ijk} \geq 0$, this yields
	\begin{equation} \label{eq:lambda_bound}
		D_{ijk} (R-r_i)  = P-p_j - s_{ijk} \leq P-p_j = D_{ijk} \sigma_i.
	\end{equation}
	Thus, $(R-r_i) \leq \sigma_i$ since $D_{ijk} > 0$.
	
	Let $Q_i$ denote all segments of $i$'s utility function with $ppb = 1/\sigma_i$, and call this the flexible partition. Similarly, let the forced partition be all segments with strictly lower $ppb$ than $1/\sigma_i$, and the undesirable partition be all segments with strictly higher $ppb$ than $1/\sigma_i$. We show that these segments correspond to forced, flexible, and undesirable partitions described in Section \ref{sec:optimaL_bundles}. 
	
	Observe that undesirable partitions are unallocated by construction, since we selected $\sigma_i$ based on the segment receiving a positive allocation with highest $ppb$. Now consider any segment $(i,j,k)$ in agent $i$'s forced partition. We have,
	\[
	\frac{D_{ijk}}{P-p_j} < \frac{1}{\sigma_i} \Rightarrow D_{ijk} (R-r_i) \le D_{ijk} \sigma_i < P-p_j.
	\]
	Hence to satisfy \eqref{eq:ppb_2}, it must be that that $s_{ijk} > 0$. Therefore, \eqref{eq:segment_2} must hold with equality to satisfy (\ref{eq:segment_2}'). That is, the segment is fully allocated. 
	
	Finally, let $(i,j,k) \in Q_i$. If $(R-r_i) < \sigma_i$ then all the segments of this partition are also fully allocated by the similar argument as above. In other words, the agent exhausts her budget when she is done consuming $Q_i$ as the last partition. It follows from the characterization in Section \ref{sec:optimaL_bundles} that each agent receives an optimal bundle of bads. 
\end{proof}

\begin{proof}(of Theorem~\ref{thm:equiL_soL_correspondence})
	By Lemma \ref{lem:clearing_1} at prices $p^*_j=-(P-p_j)$ for all $j$ the market clears. And by Lemma \ref{lem:optimaL_bundles}, each agent receives an optimal bundle of goods in any solution to LCP~\eqref{lcp:2}, i.e., it is a competitive equilibrium. Further, upto change of variables from LCP~\eqref{lcp:1} to LCP~\eqref{lcp:2}, Lemma \ref{lem:market_solution_1} shows that every competitive equilibrium yields a solution to LCP~\eqref{lcp:2}.
\end{proof}

\subsection{Proof of Lemma~\ref{lem:clearing_2}}\label{sec:lem:clearing_2}
\begin{proof}
	The proof follows similarly to Lemma \ref{lem:clearing_1} in Appendix~\ref{sec:lem:clearing_1}. Let $(\x^*,\p^*)$ be an equilibrium. Set $p_j = |p_j^*|, \ \forall j\in M^-$, and $p_j = p_j^*, \ \forall j\in M^+$.  For each agent $i$ and each item $j$, we distribute $x_{ij}^*$ among individual segments by filling starting from the first segment until all of $x_{ij}^*$ is used, according to \eqref{eq:x_from_allocation}. The market clearing conditions for the equilibrium ensures that setting $f_{ijk} = x_{ijk}^* p_j$ together with $\p$ satisfies \eqref{eq:budget_3}, \eqref{eq:spending_3_bad}, and \eqref{eq:spending_3_good}.
	
	Next, suppose $\p,\f$ satisfy \eqref{eq:budget_3}, \eqref{eq:spending_3_bad}, \eqref{eq:spending_3_good} and $\p > 0$. Summing \eqref{eq:budget_3} over all $i\in N$, \eqref{eq:spending_3_bad} over all $j \in M^-$, and \eqref{eq:spending_3_good} over all $j \in M^+$ gives
	\begin{equation*}
	\sum_{j\in M^+} p_j -\sum_{j\in M^-}p_j = \sum_{i,k,j\in M^+} f_{ijk} - \sum_{i,k,j\in M^-} f_{ijk} \leq \sum_{i,j\in M^+} W_{ij} p_j - \sum_{i,j\in M^-} W_{ij} p_j = \sum_{j\in M^+} p_j -\sum_{j\in M^-}p_j
	\end{equation*}
	since there is a unit amount of each item, i.e., $\sum_{i} W_{ij} =1$. Since all variables are non-negative, it follows that \eqref{eq:budget_3}, \eqref{eq:spending_3_bad}, and \eqref{eq:spending_3_good} hold with equality. Therefore, setting $x_{ijk} = f_{ijk}/p_j$ ensures the market clears. 
\end{proof}

\subsection{Proof of Lemma~\ref{lem:market_solution_2}}\label{sec:lem:mkt_sol_2}
\begin{proof}
	Let $(\x^*,\p^*)$ be a competitive equilibrium. In the LCP set $p_j=|p^*_j|,\ \forall j\in M$ and $\f$ as done in the proof of Lemma \ref{lem:clearing_2}. By Lemma \ref{lem:clearing_2}, \eqref{eq:budget_3}, \eqref{eq:spending_3_bad}, and \eqref{eq:spending_3_good} hold with equality since $(\x^*,\p^*)$ clears the market. Therefore, so do (\ref{eq:budget_3}'), (\ref{eq:spending_3_bad}'), and (\ref{eq:spending_3_good}'). 
	For each agent $i$, if $i$ purchases any goods, then set $r_i$ according to \eqref{eq:lambda_mbb} for any flexible segment of goods. Otherwise, set $r_i$ according to \eqref{eq:lambda_mbb} for any flexible segment of bads. In either case $r_i > 0$ since $i$ is non-satiated for some good $j$, i.e., $U_{ijk} >0$, and $0 < D_{ij1} \leq D_{ij2} \leq \dots$, for all bads $j\in M^-$. 
	
	We set the variables $s_{ijk}$ as follows. If $(i,j,k)$ is undesirable or flexible set $s_{ijk} = 0$, whether $j$ is a good or a bad. Recall that, for a forced segment $(i,j,k)$, if $j\in M^+$ then $\frac{1}{r_i} < \frac{U_{ijk}}{p_j}$, and if $j\in M^-$ then $\frac{1}{r_i} > \frac{U_{ijk}}{p_j}$. Using this, set it's $s_{ijk}$ to satisfy
	\begin{equation*}
	\frac{1}{r_i} = \frac{U_{ijk}}{p_j + s_{ijk}}, \mbox{ if $j\in M^+$}, \quad \text{or} \quad \frac{1}{r_i} = \frac{D_{ijk}}{p_j - s_{ijk}}, \mbox{ if $j \in M^-$}.
	\end{equation*}
	It is easy to verify that in each case: the segment is forced, flexible, or undesirable; the constraints \eqref{eq:mbb_3_bad}, \eqref{eq:mbb_3_good}, and \eqref{eq:segment_3} are satisfied, as well as the corresponding complementarity conditions (\ref{eq:mbb_3_bad}'), (\ref{eq:mbb_3_good}'), and (\ref{eq:segment_3}').
\end{proof}

\subsection{Proof of Lemma~\ref{lem:optimaL_bundles_2}}\label{sec:optimaL_bundles_2}
\begin{proof} 
Given a solution of LCP~\eqref{lcp:4}, define $p^*_j=(P-p_j),\ \forall j\in M^+$ and $p^*_j=-(P-p_j),\ \forall j\in M^-$. And allocation $x^*_{ij} = \sum_k f_{ijk}/p^*_j,\ \forall (i,j)$. We want to show that $(\p^*,\x^*)$ gives a competitive equilibrium. It is easy to show that market clears at $(\p^*,\x^*)$ using \eqref{eq:budget_4}, \eqref{eq:spending_4_bad}, and \eqref{eq:spending_4_good} using similar argument as in Lemma \ref{lem:clearing_2}. Next, we show that every agent receives an optimal bundle as per $\x^*$ at prices $\p^*$.

	 Recall that we have picked $P$ and $R$ such that 
	 $\min_{j\in M^-,i,k} D_{ijk}R-P > 0$, and $P-\min_{j\in M^+,i.k}$ $ U_{ijk}R  <0$. While similar to the proof of Lemma \ref{lem:optimaL_bundles}, we now rely on the assumption that each agent $i$ is non-satiated for some good $j$, i.e., the final segment $(i,j,k)$ of good $j$ satisfies $U_{ijk} >0$. Notice that since $D_{ijk} > 0, \forall i,k$ for any bad $j$, and the above assumption on goods implies that $r_i > 0, \ \forall\ i \in N$. Consider two cases: an agent purchases some bads, or they purchase only goods. In the first case, if $r_i = 0$, then \eqref{eq:mbb_4_bad} is a strict inequality. Then (\ref{eq:mbb_4_bad}') requires $f_{ijk} = 0, \ \forall\ i,k$, $\forall j\in M^-$, contradicting market clearing, Lemma \ref{lem:clearing_2}. Similarly, in the second case, if $r_i =0$, then \eqref{eq:mbb_4_good} can not hold for the non-satiated segment with infinite length. 
	
	Here, we diverge from the case of all bads, depending on whether an agent purchases any goods (or bads). Consider any agent $i$. There are three cases, $i$ purchases: a) only goods, b) only bads, or c) goods and bads. We focus on the last case as it is the most complicated. The first two cases can be handled in a similar manner. Let $(i,j,k)$ be the segment of goods with lowest bang per buck that agent $i$ spends on, i.e., with $f_{ijk}>0$. Define $\nu_i$ as the inverse $bpb$ of this segment
	\begin{equation*}
	\nu_i= \frac{P-p_j}{U_{ijk}}.
	\end{equation*}
	Note that $0< \nu_i < \infty$, since $p_j < P$, and each agent is non-satiated for some good $j$. Similarly, let $(i,j',k')$ be the segment of bads with the highest pain per buck, $ppb$, that $i$ spends on, and define $\sigma_i$ as in Lemma \ref{lem:optimaL_bundles}. We want to show that $\nu_i \leq (R-r_i) \leq \sigma_i$. Therefore, $bpb_{ijk} \geq ppb_{ij'k'}$ for any good $j$ and any bad $j'$ that $i$ spends on. From \eqref{eq:lambda_bound}, we have $(R-r_i) \leq \sigma_i$. By a similar argument, for the segment $(i,j,k)$ with lowest bang per buck
	\begin{equation*}
	U_{ijk} (R-r_i) = (P-p_j) + s_{ijk} \geq (P-p_j) = U_{ijk} \nu_i.
	\end{equation*}
	Thus, $(R-r_i) \geq \nu_i$.
	
	Let $G_i$ denote the set of segments of goods with $bpb = 1/\nu_i$, and call this the flexible partition of goods. Similarly, let the forced partition of goods be all segments with strictly higher $bpb$ than $1/\nu_i$, and the undesirable partition of goods be all segments with strictly lower $bpb$ than $1/\nu_i$. Define the various partitions of bads as: forced for $ppb$ strictly less than $\sigma_i$, undesirable if $ppb$ strictly more than $\sigma_i$, and let $B_i$ be the flexible partition for bads where $ppb = 1/\sigma_i$. As per the optimal bundle characterization described in Section \ref{sec:optimaL_bundles} we need to show that $f_{ijk}$'s are zero for the segments in undesirable partitions, $f_{ijk}=L_{ijk}(P-p_j)$ for the segments in forced partitions, and $0\le f_{ijk} \le L_{ijk} (P-p_j)$ for segments in $G_i$ and $B_i$. 
	
	Observe that undesirable goods (bads) are unallocated by construction, since we selected $\nu_i$ ($\sigma_i$) based on the segment receiving a positive allocation with lowest $bpb$ (highest $ppb$). Consider any segment $(i,j,k)$ in agent $i$'s forced partition, whether a bad or a good. Observe that $s_{ijk} > 0$, in order to satisfy \eqref{eq:mbb_4_bad} or \eqref{eq:mbb_4_good}. Therefore, (\ref{eq:segment_4}') requires that \eqref{eq:segment_4} holds with equality. That is the segment is fully allocated. 
	
	For the flexible partition, if $\nu_i < (R-r_i)$, then for all $(i,j,k) \in G_i$ it must be that $s_{ijk}>0$ and hence $f_{ijk} = L_{ijk}(P-p_j)$, otherwise $(i,j,k)$ could be partially allocated. Similarly, if $(R-r_i) < \sigma_i$, then all the segments in $B_i$ are fully allocated, otherwise they could be partially allocated. 
	Thus, $i$ only purchases goods with $bpb \geq ppb$, and in the flexible partition of goods and bads $bpb = ppb$. It follows from the characterization in Section \ref{sec:optimaL_bundles} that each agent receives an optimal bundle of bads. 
\end{proof}

\subsection{Proof of Theorem~\ref{thm:equiL_soL_correspondence_2}}\label{athm:equiL_soL_correspondence_2}
\begin{proof}
	By Lemmas \ref{lem:clearing_2} and \ref{lem:optimaL_bundles_2}, the market clears and each agent receives an optimal bundle of goods in any solution to LCP~\eqref{lcp:4} with $p_j < P,\ \forall j\in M$ and $r_i<R,\ \forall i\in N$, i.e., it is a competitive equilibrium. Further, Lemma \ref{lem:market_solution_2} shows that every competitive equilibrium prices yields a solution to LCP~\eqref{lcp:4} with $p_j < P,\ \forall j\in M$ and $r_i<R,\ \forall i\in N$. Therefore, solutions to LCP~\eqref{lcp:4} with $p_j <P, \ \forall j\in M$, and $r_i < R, \ \forall i\in N$, exactly captures competitive equilibria up to scaling.
\end{proof}

\subsection{Proof of Theorem~\ref{thm:non_degenerate_2}}\label{sec:thm:non_deg_2}
We first show this theorem for the case of all bads, i.e., for LCP~\eqref{lcp:5} without \eqref{eq:spending_4a_good} and \eqref{eq:mbb_4a_good}.
\begin{theorem} \label{thm:non_degenerate}
In case of all bads, if the input parameters $\mathbf{D}$, $\mathbf{W}$, and $\mathbf{L}$ have no polynomial relation among them, then every vertex of $\mathcal{P}$ with $z > 0$, $p_j < P, \ \forall j\in M$, and $r_i<R,\ \forall i\in N$ is nondegenerate.
\end{theorem}

\begin{proof}
	Let $S = (\p,\q,\r,\s,z)$ be a vertex solution to LCP~\eqref{lcp:5} with $z>0$ and $p_j < P, \ \forall j\in M$. For contradiction, suppose $S$ is degenerate. Then, there are at least two double labels at $S$. Let $\cI$ be the set of inequalities of LCP~\eqref{lcp:5} which hold with equality at $S$. Remove all zero variables and their non-negativity conditions from $\cI$, as well as all conditions corresponding to double labels at $S$. Our goal is to write all non-zero variables as linear functions of $z$, where the coefficients are in terms of monomials of input parameters. Then, substituting these expressions into the double labels at $S$ yields a polynomial relation among input parameters. 
	
	For forced segments, i.e., $s_{ijk} > 0$, remove conditions \eqref{eq:mbb_4a_bad} and \eqref{eq:segment_4a} from $\cI$, and replace $f_{ijk}$ with $L_{ijk} (P-p_j)$. For undesirable segments, \eqref{eq:segment_4a} is a strict inequality and $f_{ijk} = 0$. Thus, $\cI$ contains no conditions \eqref{eq:mbb_4a_bad} or \eqref{eq:segment_4a} for undesired segments either. 
	
	Now we may write all non-zero variables as linear functions of $z$. All remaining $f_{ijk}$ correspond to spending in flexible segments. Clearly, for each agent $i$ and each bad $j$ only one such segment exists. To simplify notation we relabel $f_{ijk}$ for these flexible segments as $f_{ij}$, and the corresponding $D_{ijk}$ as $D_{ij}$.
	
	Let $\cE$ be the set of $(i,j)$ pairs such that agent $i$ has a flexible segment for bad $j$, i.e., where condition \eqref{eq:mbb_4a_bad} holds with $s_{ijk} = 0$. Then,
	\begin{equation} \label{eq:flexible_equal}
		D_{ij} r_i - p_j = D_{ij} R-P.
	\end{equation}
	By considering the pairs of $\cE$ as edges between $N$ and $M$, we obtain a bipartite graph, say $G$. Note that $G$ is acyclic, otherwise we obtain a polynomial relation between $D_{ij}$'s using \eqref{eq:flexible_equal} along the cycles to eliminate the $r_i$'s and $p_j$'s.
	
	Let $H$ be a connected component of $G$. We pick a representative bad for $H$. If there is an undersold bad, i.e., \eqref{eq:spending_4a_bad} is a strict inequality, then we pick this item, say $b$. Observe that for any bad $j\in H$, we may write $P-p_j=\frac{\phi_1(D)}{\phi_2(D)} (P-p_b) $, where $\phi_1(D)$ and $\phi_2(D)$ are monomials in terms of $D_{ij}'$s. Similarly, we may write $R-r_i$ in terms of monomials of $D_{ij}'$s. Now, since \eqref{eq:spending_4a_bad} is a strict inequality for bad $b$, the complementary condition (\ref{eq:spending_4a_bad}') requires $p_b = 0$. In addition, no other bad $j$ can be undersold in $H$, otherwise the above steps yield a polynomial relation between the $D_{ij}'$s.
	
	Suppose that for component $H$, the representative bad $b$ is not undersold, i.e., \eqref{eq:spending_4a_bad} holds with equality. Consider any leaf node $v_0$ of $H$, and remove the edge incident to it in $H$, say $(v_0,v_1)$ to create $H'$. Let $H'$ be rooted at $v_1$. Starting from leafs of $H'$ and working toward the root $v_1$, we can use market clearing conditions \eqref{eq:spending_4a_bad} and \eqref{eq:budget_4a} for bads and agents respectively, to write all $f_{ij}'$s for edges in $H'$ as linear functions of $z$ and the representative prices obtained in the first step. Market clearing conditions give two different expressions for $f_{ij}$ on the missing edge $(v_0,v_1)$. Thus, yielding a linear relation between the representative prices and $z$. This relation is non trivial because exactly one of them must contain a $W_{ij}$ not present in the other.
	
	If bad $b$ is undersold, then a similar approach using $b$ as the root allows us to write $f_{ij}'$s as linear functions of representative prices and $z$. This gives a system of linear equations: $p_b=0$ if $b$ is undersold, and $p_j$ is a linear function of representative prices and $z$ otherwise. Solving this system, we obtain $p_j$'s as linear functions of $z$. Substituting these expressions for representative prices in terms of $z$, we obtain expressions for $f_{ij}$'s, $r_i$'s, and remaining $p_j$'s. We preform the above steps for each connected component of $G$.
	
	Finally, consider the equalities of $G$ corresponding to double labels that we removed from $\cI$. Replace all variables by their linear functions of $z$. Use one double label to solve for $z$ in terms of input parameters $\mathbf{D}$, $\mathbf{W}$, and $\mathbf{L}$. Substitute this value of $z$ in to the other double label to get a polynomial relation among input parameters, a contradiction.
\end{proof}
\begin{proof}(of Theorem~\ref{thm:non_degenerate_2})
	The proof closely follows that of Theorem \ref{thm:non_degenerate}. We assume for contradiction, that the vertex is nondegenerate. Our goal is to write all non-zero variables as linear functions of $z$, where the coefficients are in terms of monomials of input parameters. Then, substituting these expressions into the double labels at $S$ yields a polynomial relation between input parameters. 
	Notice that we may still follow the steps in Theorem \ref{thm:non_degenerate} to solve for $\r$, as well as $\p$, $\f$, and $\s$ for all bads $j\in M^-$. Thus, it remains to solve for $\p$, $\f$, and $\s$ for all goods $j\in M^+$. Using similar arguments to the case of all bads, we find expressions for these variables as linear functions of $z$. Substituting these expressions into the two sets of double labels yields a polynomial relation between input parameters.
\end{proof}

\subsection{Proof of Theorem~\ref{thm:one_to_one_2}}\label{sec:thm:one_to_one_2}
To prove the theorem it suffices to show one to one correspondence between solutions of LCP~\eqref{lcp:4} and competitive equilibria. For this, we first show a similar result for the case of all bads, i.e., using LCP~\eqref{lcp:2}.

\begin{theorem} \label{thm:one_to_one}
	If the polyhedron of LCP~\eqref{lcp:2} is nondegenerate, then solutions to LCP~\eqref{lcp:2} with $p_j < P, \ \forall j\in M$, and $r_i < R, \ \forall i \in N$, are in one to one correspondence with competitive equilibria.
\end{theorem}

\begin{proof}
	Due to scale invariance, it suffices to show this for the set of competitive equilibria, say $\mathcal{E}$, where the minimum price is $-P$. In LCP~\eqref{lcp:2}, we represent the equilibrium price of a bad as $p_j^* = -P+p_j$. Therefore, we show a one to one correspondence between elements of $\mathcal{E}$ and solutions to LCP~\eqref{lcp:2} with $p_j = 0$ for some bad $j$. Let $(\x^*,\p^*) \in \mathcal{E}$. By Theorem \ref{thm:equiL_soL_correspondence}, any competitive equilibrium $(\x^*,\p^*)$ yields a solution to LCP~\eqref{lcp:2} using $\p = P+\p^*$, and $\f$ where $f_{ijk}=x^*_{ijk}(P-p_j)$. We show that this choice of $(\p,\f)$ yields exactly one solution to LCP~\eqref{lcp:2}. 
	
	For contradiction, suppose not. Then, there exists different choices of $\r$ and $\s$ that together with $(\p,\f)$ solve LCP~\eqref{lcp:2}. Observe that fixing $\p$, $\f$, and $\s$ also fixes $\r$. Therefore, it must be true that for some agent, say $i$, her flexible partition, say $Q_i$, is full allocated, i.e., $f_{ijk} = L_{ijk} (P-p_j), \forall (j,k)\in Q_i$. Set $r_i$ so that 
	\begin{equation*}
	\frac{1}{R-r_i} = \frac{D_{ijk}}{P-p_j},
	\end{equation*}
	for some segment $(j,k) \in Q_i$, and set $s_{ijk} = 0, \ \forall \ (j,k) \in Q_i$. Set the $\r$ and $\s$ for all other agents similarly. 
	
	Let $C = \sum_{ij} |u_{ij}|$, be the total number of segments over all agents and items. Observe that there are $n+m+2C$ variables in LCP~\eqref{lcp:2}. Further, the solution described above gives at least $n+m+2C+2$ inequalities of LCP~\eqref{lcp:2} hold with equality: Market clearing gives \eqref{eq:budget_2} $\forall i\in N$ and \eqref{eq:spending_2} $\forall j\in M$, and optimal bundles satisfies complementarity conditions (\ref{eq:ppb_2}') and (\ref{eq:segment_2}'). Plus the requirement $p_j = 0$ for some bad. Finally, all segments of agent $i$'s flexible partition $Q$ satisfy both \eqref{eq:segment_2} and $s_{ijk} = 0$. However, nondegeneracy of LCP~\eqref{lcp:2} means at most $n+m+2C+1$ inequalities hold with equality at any vertex. 
\end{proof}

\begin{proof} (Theorem~\ref{thm:one_to_one_2})
	Showing the one to one correspondence follows from a nearly identical argument to that of Theorem \ref{thm:one_to_one}. The only difference is that must consider the set of equilibria with maximum magnitude of price equal to $P$, i.e., $p_j = 0$ for some good or some bad. Assuming a bad has price with maximum magnitude price $P$ and follow the same argument as Theorem \ref{thm:one_to_one}. The case where a good has maximum magnitude price follows from a similar argument. 
\end{proof}

\section{Convergence of Algorithm \ref{algo:mixed} with All Bads} \label{sec:convergence_all_bads}
In this Section we prove that Algorithm \ref{algo:mixed} always converges to an equilibrium in the case of all bads, $M^+ = \emptyset$. The proofs are similar in spirit to the mixed manna case, but there are minor differences in some details. We still show that we the algorithm never sets a subset prices to zero, i.e., $p_j = P$, $\forall j\in \tilde{M} \subset M$, rather all prices are set to zero simultaneously. However, we can not rely on Lemma \ref{lem:no_p_0_goods} to ensure that $r_i = R, \ \forall i \in N$, as used in Lemma \ref{lem:no_z_0_solution_2} which shows that the algorithm stops at an equilibrium before setting $p_j  = P, \ \forall j\in M$. This is the only real difference between the proofs.

Recall LCP~\eqref{lcp:2} of Section \ref{sec:alL_bads} which gives the formulation for all bads. We require the augmented LCP which we create by adding $-z$ to the left hand side of \eqref{eq:budget_2} for all $i\in N$, yielding  
\begin{align}\label{eq:budget_2_z}
		\forall i \in N: & \ -\sum_{j} W_{ij} p_j - \sum_{j,k} f_{ijk} - z\leq -P \sum_{j} W_{ij} & \perp \ \ r_i.
\end{align}
Let $k = \arg\max_i \sum_{j\in M} W_{ij}$. Then we get the primary ray by setting $z = \sum_{j\in M} W_{kj}$, and all other variables equal to 0.

We now show that Algorithm \ref{algo:mixed} never reaches a secondary ray where $p_j = P$ for some subset of bads, and $z > 0$, and that the algorithm never reaches the degenerate solution where $p_j = P, \ \forall j\in M$, and all other variables equal to 0. 

Note that Claim \ref{claim:no_s_finaL_segment}, and Lemmas \ref{lem:bound_on_p_and_r}, and \ref{lem:no_sec_ray_first_step} still hold. Therefore, $p_j \leq P, \ \forall j\in M$, $r_i \leq R, \ \forall i\in N$, and if $p_j = P$ for some $j\in M$, then $p_j = P, \ \forall j\in M$. Thus, the algorithm never reaches a secondary ray where $p_j = P$ for some subset of bads, and $z > 0$. It remains to show that the algorithm never reaches the generate equilibrium where $p_j = P, \ \forall j\in M$. The idea is similar to Lemma \ref{lem:no_z_0_solution_2}. However, we can not use Lemma \ref{lem:no_p_0_goods} to show that $p_j = P, \ \forall j\in M$ implies $r_i = R, \ \forall i\in N$.

\begin{lemma}\label{lem:no_z_0_solution_bads}
	Starting from the primary ray, Algorithm \ref{algo:mixed} never reaches the degenerate solution where $p_j = P, \ \forall j\in M$, $r_i = R, \ \forall i\in N$, and all other variables equal to zero.
\end{lemma}

\begin{proof}
	Let $T$ be a vertex where $p_j = P, \ \forall j\in M$, $S$ be the vertex that precedes $T$, and $E$ be the edge between $S$ and $T$. At $S$, $p_j > 0$ so that all $p_j \rightarrow P$ on $E$. Therefore, complementarity condition requires that (\ref{eq:spending_2}') requires that \eqref{eq:spending_2} holds with equality on $E$, $\sum_{i,k} f_{ijk} = P-p_j, \ \forall j\in M$. Since $p_j <P$ at $S$, this requires that for each bad $j\in M$, at least on agent, say $i$, purchases some this bad, i.e., $f_{ijk} > 0$. Then complementarity condition (\ref{eq:ppb_2}') requires that \eqref{eq:ppb_2} is tight. Observe that this implies that $r_i > 0$, otherwise \eqref{eq:ppb_2} holds with strict inequality for all segments $(i,j,k)$. Therefore, for this agent, \eqref{eq:budget_2_z} holds with equality on $E$ by complementarity condition (\ref{eq:budget_2_z}'). 
	
	We want to argue that $r_i > 0, \forall i\in N$. If this condition holds then \ref{eq:budget_2_z} is tight $\forall i \in N$, and \eqref{eq:spending_2} is tight $\forall j\in M$. Summing over all of the constraints yields
	\begin{equation*}
		\sum_j P-p_j = \sum_{j} W_{ij} (P-p_{j}) = \sum_{ijk} f_{ijk} + nz = \sum_{j} P-p_j +nz,
	\end{equation*}
	at $S$, since $\sum_i W_{ij} =1$. Then, $z=0$ at $S$, which is a competitive equilibrium by Theorem \ref{thm:equiL_soL_correspondence}.
	
	For contradiction, assume that $r_k > 0$, for some strict subset of agents $k \in N_1 \subset N$. Note that for all agents $i\in N_0 = N\setminus N_1$, \eqref{eq:ppb_2} holds with strict inequality since $r_i =0$, and therefore complementarity condition (\ref{eq:ppb_2}') requires that $f_{ijk} = 0$ for all segments $(j,k)$ for all $i\in N_0$. Further, since $p_j > 0, \ \forall j \in M$, at $S$ then (\ref{eq:spending_2}') requires that \eqref{eq:spending_2} is tight for all $j\in M$. Then, we see that $\sum_{j,k,i\in N_1} f_{ijk} = \sum_{ijk} f_{ijk} = \sum_j (P-p_j)$.
	
	Next, observe that \eqref{eq:budget_2_z} is tight for all $i \in N_1$ by complementarity condition (\ref{eq:budget_2_z}'). Therefore, $\sum_{j,i\in N_1} W_{ij} (P-p_j) = \sum_{j,k,i\in N_1} f_{ijk} + |N_1|z$. Also, since every agent is endowed with some fraction of at least one bad and $p_j < P$ at $S$, $\sum_{j,i\in N_1} W_{ij} (P-p_j) < \sum_j (P-p_j)$. Combining the above results yields
	\begin{equation*}
		\sum_j (P-p_j) > \sum_{j,i\in N_1} W_{ij} (P-p_j)=\sum_{j,k,i\in N_1} f_{ijk} + |N_1|z = \sum_j (P-p_j) + |N_1|z,
	\end{equation*} 
	at $S$. Thus, we obtain a contradiction since $p_j < P, \ \forall j \in M$ and $z\geq 0$ at $S$.
\end{proof}

The only remaining step to show convergence of Algorithm \ref{algo:mixed} in the case of all bads is to show that the algorithm never reaches a secondary where $\p<P$, and $\r < R$. However, this follows the argument of Theorem~\ref{thm:no_secondary_rays}, while simply ignoring the steps that relate to goods. 

Then, Lemmas \ref{lem:no_sec_ray_first_step} and \ref{lem:no_z_0_solution_bads} show that starting from the primary ray, $\p < P$ and $\r < R$. Specifically, Algorithm \ref{algo:mixed} never reaches a secondary ray where $p_j = P$ for some subset of bads, and it never reaches the degenerate solution. Theorem~\ref{thm:no_secondary_rays} shows that the algorithm never reaches any other secondary ray. Therefore, eventually we reach a vertex where $\p < P$, $\r < R$, and $z =0$, which is an equilibrium by Theorem \ref{thm:equiL_soL_correspondence}.
\bibliographystyle{abbrv}
\bibliography{literature}

\begin{thebibliography}{10}

\bibitem{spliddit}
\url{www.spliddit.org}.

\bibitem{mitwohnen}
\url{www.mitwohnen.org}.

\bibitem{AnariGSS17}
N.~Anari, S.~O. Gharan, A.~Saberi, and M.~Singh.
\newblock {Nash Social Welfare, Matrix Permanent, and Stable Polynomials}.
\newblock In {\em 8th Innovations in Theoretical Computer Science Conference
  (ITCS)}, pages 1--12, 2017.

\bibitem{AnariMGV18}
N.~Anari, T.~Mai, S.~O. Gharan, and V.~V. Vazirani.
\newblock Nash social welfare for indivisible items under separable,
  piecewise-linear concave utilities.
\newblock In {\em Proc.\ 29th Symp.\ Discrete Algorithms (SODA)}, pages
  2274--2290, 2018.

\bibitem{AzizCIW19}
H.~Aziz, I.~Caragiannis, A.~Igarashi, and T.~Walsh.
\newblock Fair allocation of indivisible goods and chores.
\newblock In {\em Proc.\ 28th Intl.\ Joint Conf.\ Artif.\ Intell.\ (IJCAI)},
  2019.

\bibitem{AzizCL19a}
H.~Aziz, H.~Chan, and B.~Li.
\newblock Maxmin share fair allocation of indivisible chores to asymmetric
  agents.
\newblock In {\em Proc.\ 18th Conf.\ Auton.\ Agents and Multi-Agent Systems
  (AAMAS)}, pages 1787--1789, 2019.

\bibitem{AzizCL19}
H.~Aziz, H.~Chan, and B.~Li.
\newblock Weighted maxmin fair share allocation of indivisible chores.
\newblock In {\em Proc.\ 28th Intl.\ Joint Conf.\ Artif.\ Intell.\ (IJCAI)},
  2019.

\bibitem{AzizRSW17}
H.~Aziz, G.~Rauchecker, G.~Schryen, and T.~Walsh.
\newblock Algorithms for max-min share fair allocation of indivisible chores.
\newblock In {\em Proc.\ 31st Conf.\ Artif.\ Intell.\ (AAAI)}, pages 335--341,
  2017.

\bibitem{AzrieliS14}
Y.~Azrieli and E.~Shmaya.
\newblock Rental harmony with roommates.
\newblock {\em J. Economic Theory}, 153:128--137, 2014.

\bibitem{BudishC10}
E.~B.~Budish and E.~Cantillon.
\newblock The multi-unit assignment problem: Theory and evidence from course
  allocation at harvard.
\newblock {\em American Economic Review}, 102, 2010.

\bibitem{BarmanM17}
S.~Barman and S.~K. Krishnamurthy.
\newblock Approximation algorithms for maximin fair division.
\newblock In {\em Proc.\ 18th Conf.\ Economics and Computation (EC)}, pages
  647--664, 2017.

\bibitem{BarmanKV18}
S.~Barman, S.~K. Krishnamurthy, and R.~Vaish.
\newblock Finding fair and efficient allocations.
\newblock In {\em Proc.\ 19th Conf.\ Economics and Computation (EC)}, pages
  557--574, 2018.

\bibitem{Bataille99}
N.~Bataille, M.~Lema\^itre, and G.~Verfaillie.
\newblock Efficiency and fairness when sharing the use of a satellite.
\newblock In {\em Proceedings of the 5th International Symposium on Artificial
  Intelligence, Robotics and Automation in Space}, pages 465--470, 1999.

\bibitem{BogomolnaiaMSY17}
A.~Bogomolnaia, H.~Moulin, F.~Sandomirskiy, and E.~Yanovskaia.
\newblock Competitive division of a mixed manna.
\newblock {\em Econometrica}, 85(6):1847--1871, 2017.

\bibitem{BogomolnaiaMSY19}
A.~Bogomolnaia, H.~Moulin, F.~Sandomirskiy, and E.~Yanovskaia.
\newblock Dividing bads under additive utilities.
\newblock {\em Social Choice and Welfare}, 52(3):395--417, 2019.

\bibitem{BrainardS00}
W.~Brainard and H.~Scarf.
\newblock How to compute equilibrium prices in 1891.
\newblock {\em Cowles Foundation Discussion Paper}, 1270, 2000.

\bibitem{BramsT96}
S.~J. Brams and A.~D. Taylor.
\newblock {\em Fair division - from cake-cutting to dispute resolution}.
\newblock Cambridge University Press, 1996.

\bibitem{BranzeiS19}
S.~Branzei and F.~Sandomirskiy.
\newblock Algorithms for competitive division of chores.
\newblock arXiv:1907.01766, 2019.

\bibitem{Budish11}
E.~Budish.
\newblock The combinatorial assignment problem: Approximate competitive
  equilibrium from equal incomes.
\newblock {\em J. Political Economy}, 119(6):1061--1103, 2011.

\bibitem{CaragiannisKMPSW16}
I.~Caragiannis, D.~Kurokawa, H.~Moulin, A.~Procaccia, N.~Shah, and J.~Wang.
\newblock The unreasonable fairness of maximum {N}ash welfare.
\newblock In {\em Proc.\ 17th Conf.\ Economics and Computation (EC)}, pages
  305--322, 2016.

\bibitem{ChaudhuryCGGHM18}
B.~R. Chaudhury, Y.~K. Cheung, J.~Garg, N.~Garg, M.~Hoefer, and K.~Mehlhorn.
\newblock On fair division for indivisible items.
\newblock In {\em 38th {IARCS} Annual Conference on Foundations of Software
  Technology and Theoretical Computer Science, {FSTTCS}}, pages 25:1--25:17,
  2018.

\bibitem{ChaudhuryGMM20}
B.~R. Chaudhury, J.~Garg, P.~McGlaughlin, and R.~Mehta.
\newblock Dividing bads is harder than dividing goods: On the complexity of
  fair and efficient division of chores.
\newblock Available at: \url{http://jugal.ise.illinois.edu/chores.pdf}, 2020.

\bibitem{ChenDDT09}
X.~Chen, D.~Dai, Y.~Du, and S.~Teng.
\newblock Settling the complexity of {A}rrow-{D}ebreu equilibria in markets
  with additively separable utilities.
\newblock In {\em Proc.\ 50th Symp.\ Foundations of Computer Science (FOCS)},
  pages 273--282, 2009.

\bibitem{ChenDT09}
X.~Chen, X.~Deng, and S.-H. Teng.
\newblock Settling the complexity of computing two-player {N}ash equilibria.
\newblock {\em J. ACM}, 56(3), 2009.

\bibitem{ChenPY13}
X.~Chen, D.~Paparas, and M.~Yannakakis.
\newblock The complexity of non-monotone markets.
\newblock In {\em Proc.\ 45th Symp.\ Theory of Computing (STOC)}, pages
  181--190, 2013.

\bibitem{ChenT09}
X.~Chen and S.~Teng.
\newblock Spending is not easier than trading: {O}n the computational
  equivalence of {F}isher and {A}rrow-{D}ebreu equilibria.
\newblock In {\em Proc.\ 20th Intl.\ Symp.\ Algorithms and Computation
  (ISAAC)}, pages 647--656, 2009.

\bibitem{chen2009spending}
X.~Chen and S.-H. Teng.
\newblock Spending is not easier than trading: on the computational equivalence
  of fisher and arrow-debreu equilibria.
\newblock In {\em International Symposium on Algorithms and Computation}, pages
  647--656. Springer, 2009.

\bibitem{ColeDGJMVY17}
R.~Cole, N.~Devanur, V.~Gkatzelis, K.~Jain, T.~Mai, V.~Vazirani, and
  S.~Yazdanbod.
\newblock Convex program duality, {F}isher markets, and {N}ash social welfare.
\newblock In {\em Proc.\ 18th Conf.\ Economics and Computation (EC)}, 2017.

\bibitem{ColeG15}
R.~Cole and V.~Gkatzelis.
\newblock Approximating the {N}ash social welfare with indivisible items.
\newblock In {\em Proc.\ 47th Symp.\ Theory of Computing (STOC)}, pages
  371--380, 2015.

\bibitem{CottlePS92}
R.~Cottle, J.-S. Pang, and R.~Stone.
\newblock {\em The Linear Complementarity Problem}.
\newblock Academic Press, Boston, 1992.

\bibitem{Dantzig63}
G.~Dantzig.
\newblock {\em Linear Programming and Extensions}.
\newblock Princeton University Press, 1963.

\bibitem{DaskalakisGP09}
C.~Daskalakis, P.~Goldberg, and C.~Papadimitriou.
\newblock The complexity of computing a {N}ash equilibrium.
\newblock {\em SIAM J. Comput.}, 39(1):195--259, 2009.

\bibitem{DevanurGV16}
N.~Devanur, J.~Garg, and L.~V{\'{e}}gh.
\newblock A rational convex program for linear {A}rrow-{D}ebreu markets.
\newblock {\em ACM Trans.\ Econom.\ Comput.}, 5(1):6:1--6:13, 2016.

\bibitem{DevanurK08}
N.~Devanur and R.~Kannan.
\newblock Market equilibria in polynomial time for fixed number of goods or
  agents.
\newblock In {\em Proc.\ 49th Symp.\ Foundations of Computer Science (FOCS)},
  pages 45--53, 2008.

\bibitem{DevanurPSV08}
N.~Devanur, C.~Papadimitriou, A.~Saberi, and V.~Vazirani.
\newblock Market equilibrium via a primal--dual algorithm for a convex program.
\newblock {\em J. ACM}, 55(5), 2008.

\bibitem{DuanGM16}
R.~Duan, J.~Garg, and K.~Mehlhorn.
\newblock An improved combinatorial polynomial algorithm for the linear
  {A}rrow-{D}ebreu market.
\newblock In {\em Proc.\ 27th Symp.\ Discrete Algorithms (SODA)}, pages
  90--106, 2016.

\bibitem{DuanM15}
R.~Duan and K.~Mehlhorn.
\newblock A combinatorial polynomial algorithm for the linear {A}rrow-{D}ebreu
  market.
\newblock {\em Inf.\ Comput.}, 243:112--132, 2015.

\bibitem{Eaves76}
B.~C. Eaves.
\newblock A finite algorithm for the linear exchange model.
\newblock {\em J. Math.\ Econom.}, 3:197--203, 1976.

\bibitem{Eisenberg61}
E.~Eisenberg.
\newblock Aggregation of utility functions.
\newblock {\em Management Sci.}, 7(4):337--350, 1961.

\bibitem{EisenbergG59}
E.~Eisenberg and D.~Gale.
\newblock Consensus of subjective probabilities: {T}he {P}ari-{M}utuel method.
\newblock {\em Ann.\ Math.\ Stat.}, 30(1):165--168, 1959.

\bibitem{EtkinPT05}
R.~Etkin, A.~Parekh, and D.~Tse.
\newblock Spectrum sharing for unlicensed bands.
\newblock In {\em In Proceedings of the first IEEE Symposium on New Frontiers
  in Dynamic Spectrum Access Networks}, 2005.

\bibitem{GargHM18}
J.~Garg, M.~Hoefer, and K.~Mehlhorn.
\newblock Approximating the {N}ash social welfare with budget-additive
  valuations.
\newblock In {\em Proc.\ 29th Symp.\ Discrete Algorithms (SODA)}, 2018.

\bibitem{GargK15}
J.~Garg and R.~Kannan.
\newblock Markets with production: {A} polynomial time algorithm and a
  reduction to pure exchange.
\newblock In {\em EC}, pages 733--749, 2015.

\bibitem{GargKK20}
J.~Garg, P.~Kulkarni, and R.~Kulkarni.
\newblock Approximating {N}ash social welfare under submodular valuations
  through (un)matchings.
\newblock In {\em SODA}, 2020.
\newblock To appear.

\bibitem{GargM20}
J.~Garg and P.~McGlaughlin.
\newblock Computing competitive equilibria with mixed manna.
\newblock In {\em Proceedings of the 19th International Conference on
  Autonomous Agents and Multiagent Systems, {AAMAS} '20, Auckland, New Zealand,
  May 9-13, 2020}, pages 420--428, 2020.

\bibitem{GargMSV15}
J.~Garg, R.~Mehta, M.~Sohoni, and V.~V. Vazirani.
\newblock A complementary pivot algorithm for market equilibrium under
  separable, piecewise-linear concave utilities.
\newblock {\em SIAM J. Comput.}, 44(6):1820--1847, 2015.
\newblock Extended abstract appeared in STOC 2012.

\bibitem{GargMV18}
J.~Garg, R.~Mehta, and V.~V. Vazirani.
\newblock Substitution with satiation: {A} new class of utility functions and a
  complementary pivot algorithm.
\newblock {\em Math. Oper. Res.}, 43(3):996--1024, 2018.
\newblock Extended abstract appeared in STOC 2014.

\bibitem{GargV14}
J.~Garg and V.~V. Vazirani.
\newblock On computability of equilibria in markets with production.
\newblock In {\em SODA}, pages 1329--1340, 2014.

\bibitem{GargV19}
J.~Garg and L.~A. V{\'e}gh.
\newblock A strongly polynomial algorithm for linear exchange markets.
\newblock In {\em Proc.\ 51st Symp.\ Theory of Computing (STOC)}, 2019.

\bibitem{GhodsiZHKSS11}
A.~Ghodsi, M.~Zaharia, B.~Hindman, A.~Konwinski, S.~Shenker, and I.~Stoica.
\newblock Dominant resource fairness: Fair allocation of multiple resource
  types.
\newblock In {\em Proceedings of the 8th USENIX Conference on Networked Systems
  Design and Implementation}, NSDI'11, pages 323--336, 2011.

\bibitem{GhodsiHSSY17}
M.~Ghodsi, M.~HajiAghayi, M.~Seddighin, S.~Seddighin, and H.~Yami.
\newblock Fair allocation of indivisible goods: Improvement and generalization.
\newblock In {\em Proc.\ 19th Conf.\ Economics and Computation (EC)}, 2018.
\newblock Available on arXiv:1704.00222 since April 2017.

\bibitem{GoldmanP14}
J.~R. Goldman and A.~D. Procaccia.
\newblock Spliddit: unleashing fair division algorithms.
\newblock {\em SIGecom Exchanges}, 13(2):41--46, 2014.

\bibitem{HansenL18}
K.~A. Hansen and T.~B. Lund.
\newblock Computational complexity of proper equilibrium.
\newblock In {\em Proceedings of the 2018 {ACM} Conference on Economics and
  Computation, Ithaca, NY, USA, June 18-22, 2018}, pages 113--130, 2018.

\bibitem{HuangL19}
X.~Huang and P.~Lu.
\newblock An algorithmic framework for approximating maximin share allocation
  of chores.
\newblock arXiv:1907.04505, 2019.

\bibitem{Jain07}
K.~Jain.
\newblock A polynomial time algorithm for computing the {A}rrow-{D}ebreu market
  equilibrium for linear utilities.
\newblock {\em SIAM J. Comput.}, 37(1):306--318, 2007.

\bibitem{KleeM72}
V.~Klee and G.~Minty.
\newblock How good is the {S}implex algorithm?
\newblock In O.~Shisha, editor, {\em Inequalities {III}}, pages 159--175.
  Academic Press, 1972.

\bibitem{KollerMS94}
D.~Koller, N.~Megiddo, and B.~von Stengel.
\newblock Fast algorithms for finding randomized strategies in game trees.
\newblock In {\em Proceedings of the Twenty-Sixth Annual {ACM} Symposium on
  Theory of Computing, 23-25 May 1994, Montr{\'{e}}al, Qu{\'{e}}bec, Canada},
  pages 750--759, 1994.

\bibitem{LemkeH64}
C.~Lemke and J.~Howson.
\newblock Equilibrium points of bimatrix games.
\newblock {\em SIAM J. Appl.\ Math.}, 12:413--423, 1964.

\bibitem{Lemke65}
C.~E. Lemke.
\newblock Bimatrix equilibrium points and mathematical programming.
\newblock {\em Management Science}, 11(7):681--689, 1965.

\bibitem{LiptonMMS04}
R.~J. Lipton, E.~Markakis, E.~Mossel, and A.~Saberi.
\newblock On approximately fair allocations of indivisible goods.
\newblock In {\em Proc.\ 5th Conf.\ Economics and Computation (EC)}, pages
  125--131, 2004.

\bibitem{Moulin03}
H.~Moulin.
\newblock {\em Fair Division and Collective Welfare}.
\newblock MIT Press, 2003.

\bibitem{Moulin19}
H.~Moulin.
\newblock Fair division in the internet age.
\newblock {\em Annual Review of Economics}, 11, 2019.

\bibitem{Orlin10}
J.~Orlin.
\newblock Improved algorithms for computing {F}isher's market clearing prices.
\newblock In {\em Proc.\ 42nd Symp.\ Theory of Computing (STOC)}, pages
  291--300, 2010.

\bibitem{PrattZ90}
J.~W. Pratt and R.~J. Zeckhauser.
\newblock The fair and efficient division of the winsor family silver.
\newblock {\em Management Science}, 36(11):1293--1301, 1990.

\bibitem{RobertsonW98}
J.~Robertson and W.~Webb.
\newblock {\em Cake-Cutting Algorithms: Be Fair If You Can}.
\newblock AK Peters, MA, 1998.

\bibitem{SandomirskiyH19}
F.~Sandomirskiy and E.~Segal{-}Halevi.
\newblock Fair division with minimal sharing.
\newblock {\em CoRR}, abs/1908.01669, 2019.

\bibitem{SavaniS06}
R.~Savani and B.~von Stengel.
\newblock Hard-to-solve bimatrix games.
\newblock {\em Econometrica}, 74(2):397--429, 2006.

\bibitem{SonmezU10}
T.~S\"onmez and U.~Unver.
\newblock Course bidding at business schools.
\newblock {\em International Economic Review}, 51(1):99--123, 2010.

\bibitem{Sorensen12}
T.~B. S{\o}rensen.
\newblock Computing a proper equilibrium of a bimatrix game.
\newblock In {\em Proc.\ 13th Conf.\ Economics and Computation (EC)}, pages
  916--928, 2012.

\bibitem{steinhaus1948problem}
H.~Steinhaus.
\newblock The problem of fair division.
\newblock {\em Econometrica}, 16:101--104, 1948.

\bibitem{Su99}
F.~E. Su.
\newblock Rental harmony: Sperner's lemma in fair division.
\newblock {\em The American Mathematical Monthly}, 106(10):930--942, 1999.

\bibitem{Todd76}
M.~Todd.
\newblock Orientation in complementary pivot algorithms.
\newblock {\em Math.\ Oper.\ Res.}, 1(1):54--66, 1976.

\bibitem{Varian74}
H.~Varian.
\newblock Equity, envy and efficiency.
\newblock {\em J. Econom.\ Theory}, 29(2):217--244, 1974.

\bibitem{VaziraniY11}
V.~Vazirani and M.~Yannakakis.
\newblock Market equilibrium under separable, piecewise-linear, concave
  utilities.
\newblock {\em J. ACM}, 58(3):10, 2011.

\bibitem{Vegh14}
L.~V{\'{e}}gh.
\newblock Concave generalized flows with applications to market equilibria.
\newblock {\em Math.\ Oper.\ Res.}, 39(2):573--596, 2014.

\bibitem{Vegh17}
L.~A. V{\'{e}}gh.
\newblock A strongly polynomial algorithm for generalized flow maximization.
\newblock {\em Math. Oper. Res.}, 42(1):179--211, 2017.

\bibitem{Vossen02}
T.~W. Vossen.
\newblock {\em Fair allocation concepts in air traffic management}.
\newblock PhD thesis, University of Maryland, College Park, 2002.

\bibitem{Walras74}
L.~Walras.
\newblock {\em \'{E}l\'{e}ments d'\'{e}conomie politique pure, ou th\'{e}orie
  de la richesse sociale (Elements of Pure Economics, or the theory of social
  wealth)}.
\newblock Lausanne, Paris, 1874.
\newblock (1899, 4th ed.; 1926, rev ed., 1954, Engl. transl.).

\bibitem{Ye07}
Y.~Ye.
\newblock Exchange market equilibria with {L}eontief's utility: {F}reedom of
  pricing leads to rationality.
\newblock {\em Theoret.\ Comput.\ Sci.}, 378(2):134--142, 2007.

\end{thebibliography}

\end{document}